\documentclass[letterpaper,11pt]{article}
\usepackage[algo2e, linesnumbered,ruled,vlined] {algorithm2e}
\usepackage[T1]{fontenc}
\usepackage{amsmath,amsfonts,amsthm,amssymb}
\usepackage{mathtools}
\usepackage{thmtools}
\usepackage[normalem]{ulem}
\usepackage[margin=1in]{geometry}
\usepackage{multirow}
\usepackage{comment}
\usepackage{authblk}

\SetCommentSty{mycommfont}

\usepackage[dvipsnames,usenames]{xcolor}
\usepackage[colorlinks=true,pdfpagemode=UseNone,urlcolor=Black,linkcolor=RoyalBlue,citecolor=OliveGreen,pdfstartview=FitH,linktocpage]{hyperref}

\declaretheorem[numberwithin=section]{theorem}
\declaretheorem[sibling=theorem]{lemma}

\declaretheorem[sibling=theorem]{corollary}

\declaretheorem[sibling=theorem]{Example}

\theoremstyle{definition}

\declaretheorem[sibling=theorem]{Remark}

\title{Improved Approximations for Ultrametric Violation Distance}
\author[1]{Moses Charikar}%\thanks{moses@cs.stanford.edu}}
\author[1]{Ruiquan Gao}%\thanks{ruiquan@cs.stanford.edu}}
\affil[1]{Stanford University, \url{{moses, ruiquan}@cs.stanford.edu}}
\usepackage{dsfont}
\usepackage{xspace}
\usepackage{tikz}
\usepackage{tcolorbox}

\newcommand{\eps}{\epsilon}
\newcommand{\expect}{\mathbb{E}}
\newcommand*{\defeq}{\stackrel{\text{def}}{=}}
\newcommand{\ind}{\mathbf{1}}
\newcommand{\algoname}{\textsc{LP}\text{-}\textsf{UMVD}\text{-}\textsc{Pivot}\xspace}

\newcommand{\lowerr}{\mathsf{L}}
\newcommand{\higherr}{\mathsf{H}}
\newcommand{\highdet}{\mathsf{HD}}
\newcommand{\highrand}{\mathsf{HR}}
\newcommand{\UMVD}{\textsc{Ultrametric Violation Distance}\xspace}
\newcommand{\CClust}{\textsc{Correlation Clustering}\xspace}

\begin{document}

\begin{titlepage}
    \maketitle
    \thispagestyle{empty}
    \begin{abstract}
        \thispagestyle{empty}
        We study the \UMVD~problem
introduced by Cohen-Addad, Fan, Lee, and Mesmay [FOCS, 2022]. 
Given pairwise distances $x\in \mathbb{R}_{>0}^{\binom{[n]}{2}}$ as input, the goal is to 
%compute the minimum $\ell_0$ distance between $x$ and an ultrametric. 
modify the minimum number of distances so as to make it a valid ultrametric.
In other words, this is the problem of fitting an ultrametric to given data, where the quality of the fit is measured by the $\ell_0$ norm of the error; variants of the problem for the $\ell_\infty$ and $\ell_1$ norms are well-studied in the literature.

Our main result is a 5-approximation algorithm for \UMVD, improving the previous best large constant factor ($\geq 1000$) approximation algorithm.
%We give a randomized polynomial-time 
%The same algorithm gives an
We give an
$O(\min\{L,\log n\})$-approximation %algorithm 
for weighted \UMVD~where the weights satisfy triangle inequality and $L$ is the number of distinct values in the input.
%with triangle inequality constraints, where $L$ is the number of distinct values in the input.
%Beyond that, we also study the natural extension where there are some missing pairs in the input distances and the problem asks to compute an ultrametric that has a minimum disagreement with the input.
% \highlight{To the best of our knowledge, this implies the first $O(1)$-approximation for weighted \CClust where the weights satisfy the triangle inequality.}
We also give a
%a randomized polynomial-time 
$16$-approximation algorithm
for the problem on 
%if the existing pairs form
$k$-partite graphs,
where the input is specified on pairs of vertices that form a complete $k$-partite graph. 
%All three results use a unified algorithmic framework with small modifications. %for the three cases.
All our results use a unified algorithmic framework with small modifications for the three cases.
    \end{abstract}
\end{titlepage}

\newpage
\thispagestyle{empty}
\tableofcontents

\newpage 
\setcounter{page}{1}

\section{Introduction}
\label{sec:intro}
The study of ultrametrics arose in mathematics, in the work of Hensel introducing $p$-adic numbers (although the name was proposed later by Krasner~\cite{krasner1944nombres}).
These are metric spaces where the distances between any three points $i,j,k$ satisfy a stronger form of the triangle inequality (a.k.a., ultrametric inequality): $d(i,k) \leq \max(d(i,j),d(j,k))$.
In the 1960s, ultrametrics found applications in taxonomy~\cite{hartigan1967representation,johnson1967hierarchical,jardine1967structure,jardine1971mathematical} due to the natural correspondence between the classification of objects (represented by a hierarchical clustering) and ultrametrics.
(An ultrametric can be represented by the shortest path distances between leaves in a rooted tree where the root-leaf distances are equal for all leaves; the tree structure corresponds to a hierarchical clustering of the leaves).
Since then, ultrametrics have been studied and used in a wide range of fields, including %mathematics,
biology (e.g., %numerical taxonomy 
\cite{sneath1962numerical,cavalli1967phylogenetic}),
physics (see the survey by \cite{rammal1986ultrametricity}), 
finance (e.g., \cite{mantegna1999hierarchical}),
and computer science.

In this work, we consider the problem of fitting an ultrametric to given measurements of distances between pairs of points -- this is relevant to data-analysis settings where we have measurement noise, errors, or incomplete data.
%$x\in \mathbb{R}_{>0}^{\binom{[n]}{2}}$.
%to minimize the errors, i.e. the difference between the reconstructed ultrametric distances and the input.
The problem of fitting ultrametrics (and tree metrics) to observations originally arose in phylogenetic analysis. %and has been extensively studied.
It was introduced by Cavalli-Sforza and Edwards~\cite{cavalli1967phylogenetic} who were interested in minimizing the $\ell_2$ norm of the error, i.e., the difference between the reconstructed ultrametric/tree distances and the input.
Farris~\cite{farris1972estimating} proposed the problem of minimizing the $\ell_1$ norm of the error.
%The problems of minimizing the $\ell_1$ and $\ell_2$ norms 
Subsequently, these problems were shown to be NP-hard~\cite{kvrivanek1986np,day1987computational}.
Since those early results, the problem of fitting ultrametrics to given data, minimizing the $\ell_p$ norm of the error (for various values of $p$) has been extensively studied in the computer science community over the past three decades (see Section~\ref{sec:related}).
The current best-known results include a polynomial-time algorithm for the $\ell_\infty$ norm~\cite{farach1993robust}, an $O(1)$-approximation for the $\ell_1$-norm~\cite{DBLP:conf/focs/Cohen-Addad0KPT21} and an $O((\log n \log \log n)^{1/p})$-approximation for $\ell_p$ norms, $1 < p < \infty$ \cite{DBLP:journals/siamcomp/AilonC11}.

Given the long history of the study of this family of ultrametric fitting problems, it is somewhat surprising that the natural problem of minimizing the $\ell_0$ norm of the error was introduced and studied only very recently by \cite{DBLP:conf/focs/Cohen-AddadFLM22}.
This problem, called \UMVD, is the focus of our work. 
Given pairwise distances $x\in \mathbb{R}_{>0}^{\binom{[n]}{2}}$,
we study the problem of minimizing the number of entries of $x$ we need to modify to produce a valid ultrametric.
This problem is a generalization of the well-studied \CClust problem\footnote{The input to \CClust is a complete graph with $+$ and $-$ labels on edges, intended to represent similarity and dissimilarity of vertices. The goal is to output a disjoint partition into clusters which minimizes the number of edges that are misclassified, i.e., the number of $+$ edges that go across clusters and the number of $-$ edges that are inside clusters.}, which is equivalent to instances with two distinct values in the input. 
% \highlight{
%In the language of \CClust, the 
Edges with the larger distance correspond to the $-$ edges and edges with the smaller distance correspond to the $+$ edges
in \CClust. An optimal ultrametric for such instances corresponds to a partition into clusters.\footnote{Here, inter-cluster distances are set to the smaller value and intra-cluster distances are set to the larger value.}
% } 
%\Ruiquan{this sentence is moved from the footnote}
%\Moses{Now that we mention \CClust more prominently in the intro, should we also explain what the \CClust problem is?}
%\Ruiquan{I am not sure if this version is clear.}
On the other hand, \UMVD can be viewed as a collection of \CClust problems with hierarchical structure; this viewpoint was exploited in~\cite{DBLP:conf/focs/Cohen-AddadFLM22}.
%\Moses{Do we need to explain what correlation clustering is?}
%\Ruiquan{Trying to add somewhere: This problem is a generalization of the well-studied \CClust problem, where the input consists of two distinct values.\footnote{In the language of \CClust, the edges with the larger distance correspond to the $-$ edges while the edges with the smaller distance correspond to the $+$ edges.}}

% \highlight{
The previous best algorithm for \UMVD was a large constant factor approximation introduced by Cohen-Addad et al.~\cite{DBLP:conf/focs/Cohen-AddadFLM22}\footnote{the authors presented an 8,000,000 factor for their approach with unoptimized parameters, and they wrote, ``it is very likely that the current approach could easily lead to a 1000-approximation but at the expense of a more tedious proof''.}.
Their %novel 
approach %to approximate \UMVD within a constant factor 
involves new \CClust algorithms whose output has additional structure,
%These structural properties are somewhat 
similar to the ``clean'' properties used by~\cite{bansal2004correlation} to obtain the first large constant factor approximation for \CClust.
A natural idea 
%to improve the approximation ratio for \UMVD
for improvement
is to use the pivot-based approach 
used in all the best known algorithms for \CClust
\cite{DBLP:journals/jacm/AilonCN08,DBLP:conf/stoc/ChawlaMSY15,DBLP:conf/focs/Cohen-AddadLN22,cohen2023handling}: 
%The approach is established for \CClust 
%This works 
%in the following way: 
in each round, the algorithm chooses a random unclustered vertex as the pivot and decides the set of vertices in its cluster; the algorithm proceeds until all vertices are clustered.
%The introduction of this approach~\cite{DBLP:journals/jacm/AilonCN08} gives two different versions:
There are two natural instantiations of this approach:
\begin{description}
    \item[Combinatorial]: put all the $+$ neighbors of the pivot in its cluster.
    \item[LP based]: compute an optimal solution for an LP relaxation of \CClust, interpret the variables as probabilities for pairs of vertices being in the same cluster and use them for (randomly) determining cluster membership in the pivot approach.
    %and randomly decides if an unclustered vertex is in the pivot's cluster by dictating the LP variables as means of Bernoulli distributions.
\end{description}
The combinatorial version gave a 3-approximation, while the LP version initially gave a 2.5-approximation.\footnote{later improved to a $(2.06-\epsilon)$-approximation by~\cite{DBLP:conf/stoc/ChawlaMSY15} via a more elaborate treatment of the LP variables.}
Cohen-Addad et al.~\cite{DBLP:conf/focs/Cohen-AddadFLM22} considered a (natural) analog of the combinatorial version for \UMVD, where they freeze distances to the pivot and minimally modify the distances of other edges in each round, but showed a compelling negative result: even 
%under the fully optimized choice of the pivot 
for the best choice of pivot
in each round, the algorithm can only give an $\Omega(\log n)$-approximation on some instances.
However, the hard instances leave open the possibility that we can incorporate the LP-based version to obtain a good constant approximation for \UMVD. 
To construct the hard instances, they start with an ultrametric and carefully contaminate some edges.
Here, the pivot-based approach can potentially identify the contaminated edges using an LP solution.
Therefore, an interesting question to ask is
\begin{center}
    \emph{Can the pivot-based approach improve the approximation ratio for \UMVD by appropriately rounding the standard LP relaxation?}
\end{center}
Our answer to this question is positive, and we give significantly improved approximation ratios for \UMVD and its variants, using this approach.
% }

\subsection{Our Results}
%In this paper, we improve the approximation ratio of \UMVD to $5$. 
Our main result is a $5$-approximation for \UMVD.
% The previous best algorithm was a large constant approximation introduced by Cohen-Addad et al.~\cite{DBLP:conf/focs/Cohen-AddadFLM22}\footnote{the authors presented an 8,000,000 factor for their approach with unoptimized parameters, and they wrote ``it is very likely that the current approach could easily lead to a 1000-approximation but at the expense of a more tedious proof''.}.
%\Moses{Should we move the footnote in the abstract (about the large constant factor) to the intro instead?}
%\Ruiquan{I think it is better. Moved.}
\begin{theorem}[see Theorem~\ref{thm:5-approx-complete}]
\label{result:5-approx-complete}
There exists a $5$-approximation algorithm for \UMVD.
\end{theorem}
% \Ruiquan{Should we mention it here or in the related work?}

We also study a weighted variant of \UMVD where the weights satisfy the triangle inequality.
% \highlight{
%In the weighted variant of \UMVD, 
Here,
we are given an additional weight $w\in \mathbb{R}_{>0}^{\binom{[n]}{2}}$,
and the objective function is the sum of $w_{i,j}$ for all edges $(i,j)$ modified by the algorithm. 
% }
We give an algorithm that (slightly) improves the previous best $O(\log n \log\log n)$-approximation~\cite{DBLP:conf/focs/Cohen-AddadFLM22} that also works for general weights.
Interestingly, the same approximation ratio is obtained by~\cite{DBLP:conf/focs/Cohen-AddadFLM22} for the combinatorial pivot-based algorithm on (unweighted) \UMVD.
\begin{theorem}[see Theorem~\ref{thm:logn-app-s-weighted}]
\label{result:logn-app-s-weighted}
There exists an $O(\min\{L,\log n\})$-approximation algorithm for weighted \UMVD where the weights satisfy the triangle inequality and $L$ is the number of distinct values in the input.
\end{theorem}
An interesting corollary of this result is a constant factor approximation for a variant of weighted \CClust with triangle inequality constraints\footnote{We note that this name is also used in previous works but for different settings. Previously, weighted \CClust is defined by two weights $w_{i,j}^+$ and $w_{i,j}^-$ for each edge $(i,j)\in \binom{[n]}{2}$. The objective function of the problem is the sum of $w_{i,j}^+$ for all $-$ edges in the output and $w_{i,j}^-$ for all $+$ edges in the output. \cite{DBLP:journals/jacm/AilonCN08} defines the triangle inequality constraints to be $w_{i,j}^{-}+w_{j,k}^{-}\geq w_{i,k}^{-}$ for any tuple $(i,j,k)$. \cite{DBLP:journals/mor/ZuylenW09}'s definition adds an additional constraint $w_{i,j}^{+}+w_{j,k}^{-}\geq w_{i,k}^{+}$ for each tuple. Both settings are incomparable with ours.}
%\Ruiquan{other related settings are discussed in the footnote.}
%\Moses{Should we clarify what our weighted version is?}
%\Ruiquan{I added one explain for \UMVD above, is it enough here?}
because \CClust problems are special cases of \UMVD problems with $L=2$.
To the best of our knowledge, this is the first $O(1)$-approximation for this variant of \CClust.
Assuming the Unique Games Conjecture, this result also separates weighted \CClust with and without triangle inequality constraints.
\begin{corollary}
\label{result:logn-app-s-weighted-cclust}
There exists an $O(1)$-approximation algorithm for weighted \CClust where the weights satisfy the triangle inequality.
\end{corollary}

We consider another interesting variant of (unweighted) \UMVD with incomplete information, where the input distances are specified only on a subset of the edges.
In the absence of any structure on the specified edges,
%If the specified edges form a general graph, 
%it is hopeless to obtain an $O(1)$-approximation for this problem because its 
a special case is unweighted \CClust on general graphs, which does not have an $O(1)$-approximation assuming the Unique Games Conjecture~\cite{demaine2006correlation,DBLP:journals/cc/ChawlaKKRS06}.
%Our only hope is to consider this variant on some specific graphs.
We give a 16-approximation if the specified edges form a complete $k$-partite graph.
A similar setting has been studied in the literature on \CClust \cite{amit2004bicluster,DBLP:journals/siamcomp/AilonALZ12,DBLP:conf/stoc/ChawlaMSY15}, culminating in a ratio of $3$ on complete $k$-partite graphs.

\begin{theorem}[see Theorem~\ref{thm:16-approx-k-part}]
\label{result:16-approx-k-part}
There exists a $16$-approximation algorithm for \UMVD on complete $k$-partite graphs.
\end{theorem}

\subsection{Techniques}

\paragraph{Pivot-based LP rounding algorithm.} 
As is the case for \CClust~(e.g., \cite{DBLP:journals/jcss/CharikarGW05,DBLP:journals/jacm/AilonCN08,DBLP:conf/stoc/ChawlaMSY15}) and for fitting ultrametric with $\ell_1$ objective~\cite{DBLP:journals/siamcomp/AilonC11,DBLP:conf/focs/Cohen-Addad0KPT21}, a standard LP relaxation was introduced by Cohen-Addad et al.~\cite{DBLP:conf/focs/Cohen-AddadFLM22} to capture the objective of (weighted) \UMVD.
This LP can be viewed as a hierarchical generalization of the LP for \CClust, in which variables $y_\ell(i,j)$ are defined for every edge $(i,j)\in \binom{[n]}{2}$ at every level $\ell\in [L]$.
Variables at the same level are subject to the constraints of \CClust (constraint \eqref{eqn:umvd-lp-triangle} below), and variables for the same edge are monotonic with respect to the level (constraint \eqref{eqn:umvd-lp-increasing} below). 
% \highlight{
Formally, the LP relaxation is as follows, where we define $y_0(i,j)=0$ for each edge $(i,j)\in \binom{[n]}{2}$ and define $\tilde{\ell}(i,j)$ as the input distance level of $(i,j)$.
We refer the reader to Section~\ref{subsec:lp-def} for more discussion of this LP relaxation.
\begin{align}
\text{minimize} \quad 
& 
\sum_{i\neq j \in [n]} w(i,j)\cdot(1- y_{\tilde{\ell}({i,j})}(i,j) + y_{\tilde{\ell}({i,j})-1}(i,j))
\label{eqn:umvd-lp}
\tag{\textsf{UMVD}\;\textsc{LP}}
\\
\text{subject to} \quad
&
y_\ell(i,j) \leq y_\ell(i,k) + y_\ell(k,j)
& 
\forall \ell\in [L], i,j,k\in [n]
\label{eqn:umvd-lp-triangle}
\\
&
y_{\ell-1}(i,j) \le y_{\ell}(i,j)
&
\forall \ell\in [L], i,j\in [n]
\label{eqn:umvd-lp-increasing}
\\
&
y_\ell({i,j}) \in [0,1]
&
\forall \ell\in [L], i,j\in [n]
\label{eqn:umvd-lp-[0,1]}
\end{align}
% }
Previous works~\cite{DBLP:journals/siamcomp/AilonC11,DBLP:conf/focs/Cohen-Addad0KPT21,DBLP:conf/focs/Cohen-AddadFLM22} use this (or a very similar) LP relaxation as a hierarchy of LP solutions for the~\CClust problem.
In this work, we build our algorithm based on an alternative interpretation of the relaxation, 
%where variables for the same edge characterize the complementary cumulative distribution function (CCDF) for the edge's distance distribution.
%In other words, the LP variables can be interpreted as specifying the probability that edge distances exceed a value.
using the LP variables for an edge as a specification of a probability distribution on the distance value for that edge because we can always find some optimal solution with $y_{L}(i,j)=1$.
This interpretation enables us to naturally use pivot-based algorithms to round the variables. 
More specifically, the algorithm starts with a recursive call involving all vertices.
In each recursive call, we choose a random pivot from the vertex set of the call and decide the distance of edges that connect to the pivot according to the LP solution. 
% \highlight{
To determine these distances, we consider two different known rounding schemes:
\begin{description}
    \item[Randomized rounding:] we interpret the LP variables $\{y_{\ell}(i,j)\}_{\ell\in [L]}$ for each edge $(i,j)$ as a complementary cumulative distribution function (CCDF) for the distance of $(i,j)$ and sample %the distance 
    from the distribution; a similar method was introduced for \CClust~\cite{DBLP:journals/jacm/AilonCN08} .
    \item[Deterministic rounding:] we select the most probable outcome of the above distribution, breaking ties arbitrarily; a similar method can be found in~\cite{DBLP:journals/jcss/CharikarGW05} for \CClust.
\end{description}
Our algorithm uses a natural combination of these two schemes: if an edge's distribution produces a specific distance value with a high probability %exceeding a threshold 
(close to $1$), we call it a \emph{deterministic edge} and use the deterministic rounding;
otherwise, we call it a \emph{random edge} and use the randomized rounding instead.
%Beyond simply 
% \Ruiquan{To be deleted: 
% Instead of simply
% sampling from the distributions, we propose a novel approach: 
% %to treat the distributions: 
% if an edge's distribution produces a specific distance value with high probability %exceeding a threshold 
% (close to $1$), we call it a \emph{deterministic edge} and fix the edge's distance to that value;
% otherwise, we call it a \emph{random edge} and we sample the distance value from the distribution.
% }
% }
After determining the distances of the edges connected to the pivot, we modify the distances of the other edges minimally so that the ultrametric inequalities involving the pivot are satisfied.
Finally, we partition the non-pivot vertices into sets by grouping vertices with the same distance from the pivot into the same set and recursively solve the smaller instances with these vertex sets.
These latter two steps can also be found in the pivot-based algorithm of~\cite{DBLP:conf/focs/Cohen-AddadFLM22} that gives an $O(\min\{L,\log n\})$-approximation for \UMVD.

% \highlight{
One key benefit of distinguishing deterministic edges and random edges is that the ultrametric inequality can be automatically satisfied in a triangle with three deterministic edges by 
% setting the distances to the dominant distance level 
the deterministic rounding
and 
appropriate choices of the parameters.
This property and its generalizations are frequently used in our analysis.
% They are the key to proving the approximation ratios for weighted cases where the weights satisfy the triangle inequality and for $k$-partite cases.
% \Ruiquan{Moved from Section~\ref{sec:algo}}
Next, we present an example demonstrating how this property is established. 
% Refer to Example~\ref{exp:no-violation} for an example.
% }
\begin{tcolorbox}[colback=lightgray!40,arc=0pt,outer arc=0pt,width=\textwidth,boxrule=.5pt]
    \textbf{Example.~}
Consider an \UMVD instance with $L=2$ and any three distinct vertices $i,j,k\in [n]$. 
Because we can w.l.o.g. assume $y_{L}(u,v)=1$ for any distinct $u,v\in [n]$, for simplicity, we can use $a,b,c$ to denote the only variables $y_1(i,j), y_1(i,k)$ and $y_1(j,k)$. 
Consider that we deterministically round the distance of $(i,j)$ to the larger input distance if $a>\frac{2}{3}$ and deterministically round the distance of $(i,j)$ to the smaller input distance if $a<\frac{1}{3}$.
Suppose that the same rounding is done for $(i,k)$ and $(j,k)$.
If all three edges are deterministically rounded, the only possible violation of the ultrametric inequality occurs when exactly one of the edges is rounded to the larger input distance. 
However, this implies that one of $a,b,c$ is strictly greater than $\frac{2}{3}$, while the other two are strictly less than $\frac{1}{3}$, 
violating the triangle inequality constraint in the LP for $a,b,c$.
\end{tcolorbox}
% We introduce an interpretation of the relaxation in which variables for the same edge characterize the complementary cumulative distribution function (CCDF) for the distribution of the distance of the edge. 
% Beyond sampling from the distribution, we introduce a novel way to treat the distribution.
% If the distribution of an edge tends to produce a particular distance with a very high probability, say, with a probability higher than a threshold $1-\alpha$, we consider to fix the distance of the edge to this particular distance.
% Otherwise, we simply sample the distance of the edge from the distribution.
% With this treatment of the distribution, a natural approach to round the variables is to use pivot-based algorithms.
% We start with a recursive call involving all the vertices.
% In a recursive call, we randomly choose a pivot from the vertex set and use the above treatment of the distributions to decide the distances connecting to the pivot. 
% In addition, our novel treatment of the distribution is necessary for us to obtain the constant factor approximation in the weighted~\CClust with triangle inequality constraints and~\UMVD on complete $k$-partite graphs. 

\paragraph{Overcoming the $\Omega(\log n)$ barrier.} 
% \highlight{
Our main technical contribution is an improved triple-based analysis framework to overcome the previous $\Omega(\log n)$ barrier. 
We start with a simple triple-based 
%approach for \UMVD can be summarized as follows:
analysis and explain its pitfalls.
\begin{description}
    \item[A simple triple-based approach:] we upper bound the number of modified edges by the total number of times the distances on the edges are modified by the algorithm.
    The triple-based approach bounds the cost -- modifications incurred for an edge $(j,k)$ when $i$ is chosen as a pivot -- by charging the triangle $(i,j,k)$, 
    and the LP contributions of all three edges of a triangle are used to pay for the expected cost charged to the triangle.
\end{description}
This approach fails for some ``bad'' triangles, resulting in the LP contributions of some edges being charged $\Omega(\log n)$ times.
%\Moses{replaced ``bad-behaved triangles'' by ``bad triangles''.}
Consider a triangle $(i,j,k)$.
Suppose in a recursive call where $i$ is selected as the pivot, the two pivot edges $(i,j)$ and $(i,k)$ are both deterministic and are going to be fixed to the same distance. 
If the current distance of $(j,k)$ is strictly greater than the distance $(i,j)$ will be fixed to, then the third edge $(j,k)$ will be modified by the algorithm in the recursive call\footnote{even if we use the purely randomized rounding scheme, this could still happen with probability $1-o(1)$}.
If the LP contributions of $(i,j)$ and $(i,k)$ are very small (e.g., when the input distance of $(i,j)$ and $(j,k)$ equals the distance they are fixed in this recursive call), %our only option is 
we are forced
to charge an $\Omega(1)$ multiple of the LP contribution of $(j,k)$ to 
% upper bound the expected number of times these three edges are modified 
pay for the modification of $(j,k)$
in this triangle.
At the same time, the pair $j,k$ will be grouped together and get involved in a subsequent recursive call\footnote{this also possibly happens with probability $1-o(1)$ if we use the purely randomized rounding scheme.}.
Again, in subsequent recursive calls, the edge $(j,k)$ can be similarly charged by these bad triangles.
As analyzed in~\cite{DBLP:conf/focs/Cohen-AddadFLM22}, it can cause edge $(j,k)$ to be charged by $\Omega(\log n)$ such bad triangles in expectation, resulting in an $\Omega(\log n)$ bound on the approximation ratio for the simple triple-based approach.
To address this issue, we make one key observation about these triangles: 
\begin{center}
    \emph{all bad triangles have at least one edge with a high LP contribution.}
\end{center}
% To address this issue, we make two key observations about these 
%triangles causing the issue: 
% problematic triangles:
% \begin{itemize}
%     \item If the LP contribution of $(j,k)$ is low (i.e., lower than a certain constant threshold), either no edge in the triangle is modified, or the total LP contribution of $(i,j)$ and $(i,k)$ is as large as a constant, whose constant multiples are enough to upper bound the expected number of times all three edges are modified.
%     \item If the LP contribution of $(j,k)$ is high (i.e., exceeding a certain constant threshold), the expected number of times $(i,j)$ and $(i,k)$ are modified in this triangle can be upper bounded solely by the total LP contribution of $(i,j)$ and $(i,k)$; however, it is not the case for $(j,k)$.
% \end{itemize}
% The second observation inspires us to consider a different approach to upper bound the algorithm's cost: for those edges with high LP contributions, we simply upper bound its cost (i.e., whether its output disagrees with its input) by $1$ instead of the number of times they are modified by the algorithm.
For example, in the above example, the LP contribution of $(j,k)$ must be $\Omega(1)$. 
This %inspires us to use 
suggests
a different way to upper bound the number of modified edges to handle bad triangles:
%having high LP contribution modified by the algorithm, %using the definition of the \UMVD objective:
%\Moses{rewrote this previous sentence}
%\begin{center}
\emph{we simply upper bound %whether an edge is modified by the algorithm simply by $1$.
    the cost for edges with high LP contribution by 1.}
    %\Moses{rewrote this.}
%\end{center}  

Thus, we focus on upper bounding the total number of modifications on edges with low LP contributions in the triple-based analysis; %by their LP contributions.
%It turns out that 
the simple triple-based analysis is sufficient for these edges.
%\Moses{incomplete sentence commented out}
%\Ruiquan{The incomplete sentence is fixed.}
%Technically, 
By appropriately choosing the threshold for low- vs. high-cost, we can prove that edges with low LP contributions are always deterministic in the algorithm, simplifying the proof.
In Section~\ref{sec:framework} and~\ref{sec:ratio-complete}, we 
%provide a further optimized detailed analysis 
refine the analysis further
to obtain an approximation ratio of 5.
The bad triangles belong to the ``$(d,d,r)$-same-triangles'' described there.
% we also explain why the triple-based analysis is sufficient for edges with low LP contributions.
%The reason why triple-based analysis is sufficient for edges with low LP contributions can be found in the discussion of these triangles.
%}
\paragraph{Comparison between our rounding scheme vs purely randomized rounding.}
% \highlight{
%Moreover, with this improved triple-based analysis framework, 
We believe that our improved triple-based analysis framework
can also be used to establish a constant approximation ratio for the purely randomized rounding scheme.
It appears that this would not give a bound better than 5,
%we did not see that the scheme can be analyzed to a ratio less than $5$ using our framework, 
%and the analysis seems to be 
and the analysis would be more tedious and less straightforward.
% }
%
% \highlight{
In addition, our improved analysis framework can be directly applied to analyze our algorithm on the weighted variant of \UMVD where the weights satisfy the triangle inequality constraints; we obtain an $O(\min\{L,\log n\})$ approximation.
However, the purely randomized rounding scheme cannot be analyzed similarly to obtain this bound because its analysis for the unweighted cases heavily rely on the fact that all edges contribute equally to the objective.
Example~\ref{exp:counter-for-random} %is shown to 
suggests that, even when $L=2$ (the \CClust instances),
the purely randomized rounding scheme may not yield a comparable approximation in the (improved) triple-based analysis framework.
%even when $L=2$ (the \CClust instances).
%\Moses{I don't understand this last comment}.
%\Ruiquan{A ``not'' was missing. Fixed.}
% }
% \Ruiquan{The last two paragraphs are to answer the second reviewer}

% \Ruiquan{TBD: highlight the challenges we overcome for weighted correlation clustering with triangle inequality constraints.}

\paragraph{One step further: an improved algorithm for $k$-partite cases.}
% For simplicity, we assume specified edges form a set $E$ while the unspecified ones form $E_{\varnothing}$.
%However, 
% \highlight{
Our improved triple-based analysis framework still runs into the $\Omega(\log n)$ issue (when $L=\Omega(\log n)$) for the $k$-partite case (a special case of the weighted case with triangle inequality) and new ideas are needed.
% when one of the deterministic edges $(i,j)$ and $(i,k)$ is unspecified in the input. 
For example, consider a triangle $(i,j,u)$, where $(i,j)$ and $(j,u)$ are specified, but $(i,u)$ is not specified.
% \Moses{$k$ is overloaded -- used to denote a vertex as well as the number of partitions. Not sure what to do about that.}
The LP contribution of $(i,u)$ is therefore $0$. 
In the previous analysis on complete cases, if $(j,u)$ is a random edge, we need to use the LP contributions of both $(i,j)$ and $(i,u)$ to pay for the cost on $(i,j)$ when $(i,j,u)$ is a bad triangle and $u$ is chosen as the pivot. 
However, in the $k$-partite case, we can no longer rely on the LP contribution of $(i,u)$ to pay for this type of cost. 
To resolve this issue, we use a (slightly) different rounding approach that treats specified edges and unspecified edges differently.
First, we use a larger threshold for unspecified edges.
For example, we may define a specified edge as deterministic if its distance distribution (given by the LP) puts a probability mass greater than $5/8$ on some distance value, while we define an unspecified edge as deterministic only if the probability is greater than $3/4$.
Then, we employ rejection sampling on the specified random edges: with the same parameters mentioned above, %\Moses{which example?} \Ruiquan{Can we write: with the same parameters mentioned above?} \Moses{sure} 
if the sampled distance is in the largest quantile of the distribution, we reject the outcome and repeat the process until we accept. 
Using the fact that the LP variables at each level satisfy the triangle inequality constraint, the expected number of modifications on $(i,j)$ in this triangle can be upper bounded by $O(1)$ times its LP contribution.
% we are able to reduce the cost on $(i,j)$ within a $O(1)$ times solely of its LP contribution.
For a formal definition of this algorithm, see Section~\ref{sec:algo}, and for a detailed analysis, see Section~\ref{sec:ratio-k-part} and Appendix~\ref{app:edge-charging-for-k-partite-B}.
% }

% \highlight{
%\color{red}
We note that the purely randomized rounding scheme 
does not give a constant approximation ratio for 
$k$-partite instances via triple-based analysis, 
%can encounter issues when analyzing a constant factor approximation with triple-based analysis, 
even in the special case of \CClust on complete $k$-partite graphs. 
%\Moses{ What do you mean by degenerated instances?}
%\Ruiquan{Modified.}
%The aforementioned 
Example~\ref{exp:counter-for-random} shows the failure of fully randomized rounding for the $k$-partite case.
This further highlights the advantage of our new rounding scheme.
% }
%\Ruiquan{Answer to the second reviewer.}

\subsection{Further Related Work}
\label{sec:related}

\paragraph{Fitting ultrametrics with other objectives.}
The problem of fitting ultrametrics for $\ell_{\infty}$ and $\ell_1$ norms is well studied in the literature.
Farach et al.~\cite{farach1993robust} give a polynomial-time exact algorithm for the $\ell_{\infty}$ objective.
However, for the $\ell_p$ objective ($1\leq p<\infty$), its APX-hardness can be derived from the APX-hardness of \CClust~\cite{DBLP:journals/jcss/CharikarGW05}.
There is a line of works studying its approximation algorithms~\cite{harb2005approximating,DBLP:journals/siamcomp/AilonC11,DBLP:conf/focs/Cohen-Addad0KPT21}.
The state-of-the-art algorithms are an $O(1)$-approximation algorithm for the $\ell_1$ norm and $O((\log n\log \log n)^{1/p})$-approximations for $\ell_p$ norms ($1< p<\infty$).
Many other objectives and assumptions are also considered in the literature of fitting ultrametrics, including the maximization version~\cite{duggal2013resolving,DBLP:conf/focs/Cohen-AddadFLM22}, the outlier deletion version~\cite{sidiropoulos2017metric}, multiplicative distortion version~\cite{di2015finding} and assuming the input is in a high-dimensional Euclidean space~\cite{cohen2020efficient,cohen2021improving}.

\paragraph{Fitting metrics and tree metrics.} 
% There is also a line of literature on fitting a metric or tree metric to given data. 
Brickell et al.~\cite{brickell2008metric} formulated the problem of fitting metrics to a given data with $\ell_p$ objective for $1\leq p\leq \infty$, which can be solved exactly in polynomial-time via linear or convex programs.
The $\ell_0$ version was recently introduced and studied~\cite{gilbert2017if,fan2018metric,fan2020generalized,DBLP:conf/focs/Cohen-AddadFLM22}, culminating in an $O(\log n)$-approximation algorithm that runs in $O(n^3)$ time.

% The literature studying fitting tree metrics is much richer. 
On the other hand, Cavalli-Sforza and Edwards~\cite{cavalli1967phylogenetic} introduced the tree fitting problem.
If there is a tree metric that exactly fits the data, the structure can be found in polynomial time~\cite{waterman1977additive}.
However, if there is no such tree metric, the problem for any $\ell_p$ objective ($0\leq p \leq \infty$) is APX-hard~\cite{agarwala1998approximability,DBLP:journals/jcss/CharikarGW05,DBLP:conf/esa/Kipouridis23}.
Recent work of Kipouridis~\cite{DBLP:conf/esa/Kipouridis23} shows that any $\rho$-approximation of Ultrametric Violation Distance can be converted to a $6\rho$-approximation for tree metric fiiting with $\ell_0$ objective.
Combined with our result, there is automatically a $30$-approximation.
Furthermore, Agarwala et al.~\cite{agarwala1998approximability} and Cohen-Addad et al.~\cite{DBLP:conf/focs/Cohen-Addad0KPT21} give reductions from tree metrics to ultrametrics with $\ell_p$ objectives ($1\leq p\leq \infty$), only losing a constant factor.
The state-of-the-art algorithms are then a 3-approximation for $\ell_{\infty}$ norm, an $O(1)$-approximation algorithm for $\ell_1$ norm, and $O((\log n\log \log n)^{1/p})$-approximations for $\ell_p$ norms ($1< p<\infty$).
% Combined with the recent work~\cite{DBLP:conf/esa/Kipouridis23}, we give a 30-approximation for tree metric fitting with $\ell_0$ objective.
% However, we are not aware of any work considering fitting tree metrics for the $\ell_0$ norm.
% Other related works involve fitting line metrics~\cite{},  

\paragraph{Correlation clustering.} 
As mentioned before, \CClust has been extensively studied.
Bansal, Blum and Chawla~\cite{bansal2004correlation} introduced the problem and gave a constant factor approximation.
Later, it was improved by~\cite{DBLP:journals/jcss/CharikarGW05,DBLP:journals/jacm/AilonCN08,DBLP:conf/stoc/ChawlaMSY15,DBLP:conf/focs/Cohen-AddadLN22,cohen2023handling}.
The current best algorithm is a $1.73$-approximation based on the Sherali-Adams hierarchy.
Furthermore, \CClust has been proven to be APX-hard~\cite{DBLP:journals/jcss/CharikarGW05}.
Various variants of \CClust have also been studied. 
Among those, works regarding weighted cases include the general version~\cite{demaine2006correlation,DBLP:journals/jcss/CharikarGW05}, the version with probability constraints~\cite{gionis2007clustering,DBLP:journals/jacm/AilonCN08,DBLP:conf/stoc/ChawlaMSY15} and the version for asymmetric error costs~\cite{jafarov2020correlation}.

\subsection{Organization}
In Section~\ref{sec:prelim}, we present the problem formulations and discuss how the linear program captures the objective of (weighted) \UMVD.
In Section~\ref{sec:algo}, we formally describe a unified algorithm for all three cases that we study.
In Section~\ref{sec:framework}, we present our analytical framework.
Specifically, in Section~\ref{subsec:improved-triple}, we present our improved triple-based analysis for \UMVD and establish two key lemmas that are used for the three cases.
In Sections~\ref{sec:ratio-complete},~\ref{sec:ratio-s-weighted} and \ref{sec:ratio-k-part}, we respectively prove the approximation ratios for the three cases.

\section{Preliminaries}
\label{sec:prelim}
% \Ruiquan{Short version:
% \begin{itemize}
%     \item $L$ distinct input distance levels
%     \item $d_1>d_2>\cdots>d_L$: distance levels
%     \item $\ell(i,j)$: the input distance level of edge $(i,j)$
%     % \item $E=E^{(1)}\supset E^{(2)}\supset E^{(3)}\supset \cdots \supset E^{(L+1)}=\emptyset$: edges on levels higher than or equal to levels $1,2,3,\cdots, L$
%     \item $\Delta y_{\ell}(i,j)=y_{\ell}(i,j)-y_{\ell-1}(i,j)$
% \end{itemize}
% }

We use $[n]$ to denote the set $\{1,2,\cdots,n\}$. 
We use $\mathbb{R}_{\geq 0}$ and $\mathbb{R}_{>0}$ to denote non-negative and positive real numbers. 
For any set $S\subseteq [n]$, we use $\binom{S}{2}$ to denote the set of pairs of distinct elements in $S$, i.e., $\binom{S}{2}=\{(i,j):i,j\in S, i<j\}$. 
% For simplicity, we use $E$ to denote $\binom{[n]}{2}$.
% For any vector $x$ and set $S$, we use $x_{S}$ to denote the 
% We use $\times$ to denote the set product $X\times Y=\{(x,y): x\in X, y\in Y\}$. 
We use $x^+$ to denote the positive part $\max(x,0)$.
% We use $\|\!\cdot\!\|_0$ to denote the $\ell_0$ norm, which equals the number of non-zero coordinates in the vectors.
We use a tuple $t=(i,j,k)$ to denote a triangle consisting of three distinct vertices $i,j,k$ and all three edges between them.
% We use $x|_{S}$ to denote the subset of variables of $x$ restricted to the indices in $S$: $x|_S\defeq \{x_e:e\in S\}$.

\subsection{The Ultrametric Violation Distance Problem}
The \UMVD problem takes as input two disjoint sets of edges $E,E_{\varnothing}$ such that $E\cup E_{\varnothing}=\binom{[n]}{2}$ and distances $x_{\text{in}}\in \mathbb{R}_{>0}^{E}$ that are specified on the edges in $E$.
%\Moses{Can we avoid the subscript ``in''?}
The goal is to find some ultrametric $x\in \mathbb{R}_{>0}^{\binom{[n]}{2}}$ such that its disagreement with the input $\sum_{(i,j)\in E} \mathbf{1}(x(i,j)\neq x_{\text{in}}(i,j))$ is minimized. 
%$E$ forms a complete graph if $E=\binom{[n]}{2}$.
%$E$ forms a complete k-partite graph if we can find $k$ disjoint sets of vertices $V_{1},V_{2},\cdots,V_{k}$ such that $V_{1}\cup V_{2}\cup \cdots \cup V_{k}=[n]$ and $E=\binom{[n]}{2}\setminus\big(\binom{V_1}{2}\cup \cdots \cup \binom{V_k}{2}\big)$.
%\Moses{Perhaps we can remove the definition of complete $k$-partite graphs?}
We use $L$ to denote the number of distinct distances in $x_{\text{in}}$ and use $d_1>d_2>\cdots>d_L$ to denote the distinct distances in the input. 
For each edge $(i,j)\in E$, we use $\tilde{\ell}(i,j)$ to denote the input distance \emph{level} of $x_{\text{in}}(i,j)$, which satisfies $x_{\text{in}}(i,j)=d_{\tilde{\ell}(i,j)}$. 
%\Moses{Did you use $\tilde{\ell}(i,j)$ here because you use $\ell(i,j)$ elsewhere?}
%\Ruiquan{Can we use $\ell$ and $\ell(i,j)$ simultaneously?}
%\Moses{I see. I think it's ok if $\ell$ is overloaded as long as it is clear from context, but yes, formally, it is better to have separate notation as you currently have. So keep it as is.}

The weighted version of \UMVD~takes input additional weights $w\in \mathbb{R}_{\geq 0}^{\binom{[n]}{2}}$.
In addition, in the weighted version, because of the existence of zero weights, we can w.l.o.g. assume the inputs are specified on all the edges, i.e., $E=\binom{[n]}{2}$.
The goal is to minimize $\sum_{(i,j)\in E} w(i,j)\cdot \mathbf{1}(x(i,j)\neq x_{\text{in}}(i,j))$.
% Because we can move those $(i,j)\in E_{\varnothing}$ into $E$ and set $w(i,j)=0$, we can w.l.o.g. assume that $E=\binom{[n]}{2}$ for the weighted version.
%\Moses{Is the point of this comment to say that the $k$-partite version is a special case of the weighted version with triangle inequality}
%\Ruiquan{I am trying to say $E$ is an unimportant definition in the general weighted version.}
%\Moses{I agree that you might as well have the complete graph in the general weighted version. But with triangle inequality on weights, having weights specified on a subset of edges is not equivalent to having weights on a complete graph, correct? This is because the zero weight edges introduced could violate triangle inequality.}
%\Ruiquan{So we should define that $E=\binom{[n]}{2}$ in the weighted version to avoid this issue?}
In the weighted version with triangle inequality constraints, we assume $\forall i,j,k\in [n], w(i,j)+w(j,k)\geq w(i,j)$.

Moreover, the unweighted \UMVD on complete $k$-partite graphs can be viewed as a special case of the weighted version with triangle inequality constraints:
we can set $w(i,j)=1$ for $(i,j)\in E$ and set $w(i,j)=0$ for $(i,j)\in E_{\varnothing}$ in the weighted version, and the complete $k$-partite graph implies that $w$ satisfies the triangle inequality constraint.

\subsection{LP Definitions for Ultrametric Violation Distance}
\label{subsec:lp-def}
We can formulate a lower bound for (weighted) \UMVD~by an integer linear program.
For every edge $(i,j)\in \binom{[n]}{2}$ and every level $\ell\in [L]$, the integer program has a variable $y_{\ell}(i,j)\in \{0,1\}$ characterizing whether the distance of $(i,j)$ is $\geq d_{\ell}$: $y_{\ell}(i,j)=1$ if the output $x$ satisfies $x(i,j)\geq d_{\ell}$; $y_{\ell}(i,j)=0$ if the output $x$ satisfies $x(i,j)<d_{\ell}$.
We shall use $y_{\ell}(i,j) \text{ and } y_{\ell}(j,i)$ to denote the same variable for each $\ell\in [L], (i,j)\in \binom{[n]}{2}$.
For convenience, we define $y_\ell(i,j)\defeq 0$ for $\ell=0$ or $i=j$, and use $\Delta y_{\ell}(i,j)\defeq y_{\ell}(i,j)-y_{\ell-1}(i,j)$.
To guarantee that we can recover some $x\in \mathbb{R}_{>0}^{\binom{[n]}{2}}$ from the variables, we can use constraints requiring that $y_{\ell}(i,j)$ is increasing with respect to $\ell$ for every $(i,j)$ (constraint~\eqref{eqn:umvd-lp-increasing}).
To further guarantee that the recovered $x$ forms an ultrametric, we can use triangle inequality constraints on every level (constraint~\eqref{eqn:umvd-lp-triangle}) to ensure that the output forms a correlation clustering solution on every level.
The objective of (weighted) \UMVD~is to minimize the sum of weights of disagreement between the input and the output ultrametric.
Note that for each $(i,j)\in \binom{[n]}{2}$ and $\ell\in [L]$, $x(i,j)=d_{\ell}$ if and only if $y_{\ell}(i,j)=1$ and $y_{\ell-1}(i,j)=0$.
%\Moses{Is it clear to the reader that wlog, the optimal solution only uses distances in the set of pairwise distances specified in the input?}
%\Ruiquan{By modifying the distances to the maximum input values no greater than it, we can maintain the ultrametric property and obtain a solution with no worse cost. I guess it should be trivial to most readers.}
Recall that we define $\tilde{\ell}(i,j)$ as the distance level of $x_{\text{in}}(i,j)$. Therefore, we can write the objective as $\sum_{(i,j)\in E} w(i,j)\cdot (1-\Delta y_{\tilde{\ell}(i,j)}(i,j))$.
%\Moses{Check discrete derivative terminology}
%\Ruiquan{I simply give up using the terminology.}

By relaxing the integer program, we can obtain the LP relaxation~\eqref{eqn:umvd-lp}, which is equivalent to the standard LP definition introduced by~\cite{DBLP:conf/focs/Cohen-AddadFLM22} for the weighted \UMVD.
To ensure convenient access, we are restating the LP as follows:
\begin{align*}
\text{minimize} \quad 
& 
\sum_{i\neq j \in [n]} w(i,j)\cdot(1- y_{\tilde{\ell}({i,j})}(i,j) + y_{\tilde{\ell}({i,j})-1}(i,j))
% \tag{\textsf{UMVD}\;\textsc{LP}}
\\
\text{subject to} \quad
&
y_\ell(i,j) \leq y_\ell(i,k) + y_\ell(k,j)
& 
\forall \ell\in [L], i,j,k\in [n]
\\
&
y_{\ell-1}(i,j) \le y_{\ell}(i,j)
&
\forall \ell\in [L], i,j\in [n]
\\
&
y_\ell({i,j}) \in [0,1]
&
\forall \ell\in [L], i,j\in [n]
\end{align*}
% whose variables are $\{y_{\ell}(i,j)\}_{\ell\in [L], (i,j)\in \binom{[n]}{2}}$. 
% If the solution $x'$ satisfies $x'(i,j)\geq d_{\ell}$, $y_{\ell}(i,j)=1$ and otherwise $y_{\ell}(i,j)=0$. 
 % to denote the increment of $y$ on any edge $(i,j)$ and any level $\ell$. 
%to denote the discrete derivatives.
% The optimal value of this LP gives a lower bound 
%\Moses{It would be good to explain the LP variables and the constraints.}
It can be observed that any optimal solution of~\eqref{eqn:umvd-lp} can be easily modified to an optimal solution satisfying $y_L(i,j)=1$ for any $(i,j)\in E$ because these variables have non-positive coefficients in the objective of~\eqref{eqn:umvd-lp}. 
Hence, we can w.l.o.g. assume $y_L(i,j)=1$ for any $(i,j)\in E$ in the rest of the paper.

% \paragraph{CCDF viewpoint of the optimal LP solutions.} 
% Under the above assumption, for any such optimal LP solution $y$ and each edge $(i,j)\in \binom{[n]}{2}$, we can consider a random variable $D(y,i,j)$ whose supporting set is $\{d_1, \cdots, d_{L}\}$ and whose complementary cumulative distribution function (CCDF) follows $\big(y_\ell({i,j})\big)_{\ell\in [L]}$, i.e.,
% \begin{align}
%     \label{eqn:weight-ccdf}
%     \forall i\in [L],\quad \Pr[D(y,i,j)\geq d_{\ell}] = y_\ell(i,j)~.
% \end{align}
% Upon this viewpoint, the LP contribution $c(y,i,j)$ of each edge equals the probability $D(y,i,j)$ does not equal $d_{\ell({i,j})}$. 

\section{Pivot-based Algorithm}
\label{sec:algo}

In this section, we present our pivot-based LP rounding algorithm for (weighted) \UMVD~(Algorithm~\ref{alg:pivot}). 
The algorithmic framework follows the pivot-based algorithms of~\cite{DBLP:journals/siamcomp/AilonC11,DBLP:conf/focs/Cohen-AddadFLM22} and utilizes the optimal LP solutions of~\eqref{eqn:umvd-lp}.
More specifically, the algorithm is recursive, parameterized by two constants $\alpha\in [0,0.5]$ and $\beta\in [0,1]$, and has access to all distinct distances $d_1>d_2>\cdots>d_L$ of the input and an optimal solution $y^*$ satisfying $y^*_{L}(i,j)=1$ for each $(i,j)\in \binom{[n]}{2}$.
Each recursive call takes as input a subset of vertices $V\subseteq [n]$, the distances $x$ between vertices in $V$, and an upper bound level $u\in [L]$ representing that the output distances of the call should be upper bounded by $d_u$.
%\Ruiquan{I modify  ``index'' to ``upper bound level'' here.}
To run the algorithm, we call $\algoname([n], x_{\text{in}}, 1)$ and use the distances it returns as the output ultrametric. 
% Besides the information of the distances $x$, we shall also input an optimal solution $y^*$ of \eqref{eqn:umvd-lp} satisfying $y^*_{L}(i,j)=1$ for each $(i,j)\in E$. 

\begin{algorithm2e}[!t]
    \caption{$\algoname(V, x, u)$}
    \label{alg:pivot}
    \DontPrintSemicolon
    \SetKwInOut{Input}{Input}
    \SetKwInOut{Output}{Output}
    \SetKwInOut{Parameter}{Parameter}
        
    \Input{ $V\subseteq [n], ~u\in [L], ~x\in \mathbb{R}_{\geq 0}^{\binom{V}{2}}$} %,~ d\in \mathbb{R}_{\geq 0}^{[L]}, ~y^*\in \mathbb{R}_{\geq 0}^{\binom{V}{2}\times [L]}$}
    
    \Output{ $x'\in \mathbb{R}_{\geq 0}^{\binom{V}{2}}$}
    
    \Parameter{$\alpha\in [\frac{1}{2},1]$}
    
    \If{$|V|\leq 2$ \bf{or} $u=L$}{
    
        \Return{$x$}
    }
    
    $i\gets$ a vertex in $V$ uniformly at random
    
    $y\gets y^*$

    \For{$i,j\in \binom{V}{2}$}{
        \For{$\ell\in [u-1]$}{
            $y_{\ell}(i,j)\gets 0$
        }
    }
    
    \For{$j\in V\setminus \{i\}$}{

        $\ell^*\gets \arg\max_{\ell\in [L]} ~\Delta y_\ell(i,j)$ \tcp*{break ties arbitrarily}
    
        \eIf{$(i,j)\in E$}{
            \eIf{$\Delta y_{\ell^*}({i,j}) > 1-\alpha$}{
                $x'(i,j)\gets d_{\ell^*}$ \tcp*{deterministic edges in $E$}
            }{
                $x'(i,j)\gets$ sampled according to Eqn.~\eqref{eqn:distance-ccdf} \tcp*{random edges in $E$}
            }
        }
        {
            \eIf{$\Delta y_{\ell^*}({i,j}) > 1-\alpha\beta$}{
                $x'(i,j)\gets d_{\ell^*}$ \tcp*{deterministic edges in $E_{\varnothing}$}
            }{
                $x'(i,j)\gets$ sampled according to Eqn.~\eqref{eqn:empty-distance-ccdf} \tcp*{random edges in $E_{\varnothing}$}
            }
        }
    }
    
    \For{$(j,k)\in \binom{V\setminus \{i\}}{2}$}{
        \eIf{$x'({i,j})=x'({i,k})$}{
            $x'({j,k})\gets \min\{x({j,k}),x'({i,j})\}$
        }{
            $x'({j,k})\gets \max\{x'({i,j}),x'({i,k})\}$
        }
    }
    
    \For{$\ell\in \{u,u+1,\cdots,L\}$}{
        $V_\ell\gets \{j\in V: x'({i,j})=d_{\ell}\}$ \label{line:partition-non-pivots}
    
        $x'\big|_{\binom{V_\ell}{2}}\gets \algoname\Big(V_\ell, x'\big|_{\binom{V_\ell}{2}}, \ell\Big)$ \label{line:child-call}%, d, y^*\big|_{\binom{V_{\ell}}{2}\times [L]}\Big)$
        \tcp*{$x'\big|_{\binom{V_\ell}{2}}$ is $x'$ restricted to $\binom{V_{\ell}}{2}$}
    }
    
    \Return{$x'$}
\end{algorithm2e}

In each recursive call, we randomly select a pivot vertex $i\in V$ and compute a truncated LP solution $y$ in which all $y_\ell(i,j)$s are set to $0$ for any $\ell<u$. 
% Notice that the algorithm obtains the truncated LP solution $y$ in each recursive call by modifying $y^*_{\ell}(i,j)$ to $0$ for all $\ell<u$ and $(i,j)\in \binom{V}{2}$. 
According to the constraints of~\eqref{eqn:umvd-lp}, the truncated LP solution satisfies the following lemma, which enables us to view $\Delta y_{\ell}(i,j)$s as a probability distribution of the distance of each edge $(i,j)$.
\begin{lemma}
    \label{lem:feasible-y-in-calls}
    In each recursive call, the truncated LP solution $y$ is a feasible solution of \eqref{eqn:umvd-lp} for the vertices in $V$ and satisfies $y_L(i,j)=1$ for any $(i,j)\in \binom{V}{2}$. 
    In particular, $\sum_{\ell\in [L]} \Delta y_{\ell}(i,j)=1$.
\end{lemma}
%\Moses{This is obvious from the pseudocode, but do we have to explain that $y$ is set to $y^*$ and modified?}
We call edges having the pivot $i$ as an endpoint \emph{pivot edges} and call the rest of the edges \emph{non-pivot edges}.
Then, we %(randomly) 
determine the distances on the pivot edges as follows.
For each non-pivot vertex $j\in V\setminus\{i\}$, we pick a \emph{dominant level} $\ell^*(y,i,j) \in \arg\max_{\ell\in [L]} ~\Delta y_\ell(i,j)$. 
If the dominant level is not unique, we break ties arbitrarily. 
With the dominant level, we divide the edges into the following four classes and set the distances for the pivot edges as follows.
Also, see Figure~\ref{fig:rounding-process} for an intuitive example of how the rounding process works on complete graphs under $\beta=0$.
%\Moses{Can we give some intuition behind the algorithm? Why do we have to classify edges into these groups? Why do we need the deterministic edges? What is the issue with only random edges?}
%\Moses{It would be great if we can have a figure to explain the rounding process, but this might be tricky to do.}
\begin{description}
    \item[Random edges in $E$.] $(i,j)\in E$ falls into this class if %the increment of $y$ on $(i,j)$ and its dominant level is no greater than $1-\alpha$, i.e., 
    $\Delta y_{\ell^*(y,i,j)}(i,j)\leq 1-\alpha$. 
    When it is a pivot edge, its distance $x'(i,j)$ is random and follows the distribution (technically, the CCDF):
    %whose CCDF in terms of $d_{\ell}$ is proportional to $(y_{\ell}(i,j)-\alpha\beta)^{+}$, i.e.,
    %\Moses{commented out the previous text because math is easier to read directly}
    \begin{align}
        \label{eqn:distance-ccdf}
        \forall \,\ell\in [L],\quad \Pr[x'(i,j)\geq d_{\ell}] = \frac{1}{1-\alpha\beta}\cdot (y_\ell(i,j)-\alpha\beta)^+~.
    \end{align}
%\Moses{Changed the order to present both cases for edges in $E$ first.}    
    \item[Deterministic edges in $E$.] $(i,j)\in E$ falls into this class %the increment of $y$ on $(i,j)$ and its dominant level is strictly greater than $1-\alpha$, i.e., 
    $\Delta y_{\ell^*(y,i,j)}(i,j)>1-\alpha$. 
    When it is a pivot edge, its distance $x'(i,j)$ is fixed to its dominant distance level $d_{\ell^*(y,i,j)}$. 
    
    \item[Random edges in $E_{\varnothing}$.] $(i,j)\in E_{\varnothing}$ falls into this class if %the increment of $y$ on $(i,j)$ and its  dominant level is no greater than $1-\alpha\beta$, i.e., 
    $\Delta y_{\ell^*(y,i,j)}(i,j)\leq 1-\alpha\beta$. 
    When it is a pivot edge, its distance $x'(i,j)$ is random and follows the distribution (i.e., the CCDF):
    %whose CCDF in terms of $d_{\ell}$ is $y_{\ell}(i,j)$, i.e.,
    \begin{align}
        \label{eqn:empty-distance-ccdf}
        \forall \,\ell\in [L],\quad \Pr[x'(i,j)\geq d_{\ell}] = y_\ell(i,j)~.
    \end{align}

    \item[Deterministic edges in $E_{\varnothing}$.] $(i,j)\in E_{\varnothing}$ falls into this class if %the increment of $y$ on $(i,j)$ and its dominant level is strictly greater than $1-\alpha\beta$, i.e., 
    $\Delta y_{\ell^*(y,i,j)}(i,j)>1-\alpha\beta$. 
    When it is a pivot edge, its distance $x'(i,j)$ is fixed to its dominant distance level $d_{\ell^*(y,i,j)}$. 
\end{description}

\begin{Remark}
When $\beta=0$, Eqn.~\eqref{eqn:distance-ccdf} becomes $\forall \,\ell\in [L],\; \Pr[x'(i,j)\geq d_{\ell}] = y_\ell(i,j)$, which is an analog of the purely randomized scheme introduced by~\cite{DBLP:journals/jacm/AilonCN08}. 
% Except for the $k$-partite cases, we set $\beta=0$.
\end{Remark}

\begin{figure}
    \tikzset{%
        every neuron/.style={
            circle,
            draw,
            minimum size=24pt,
            very thick
        },
        neuron/.style={
            circle,
            minimum size=0.2cm
        },
    }
    \centering
    \begin{tikzpicture}[x=0.8cm, y=0.8cm, >=latex, every text node part/.style={align=center}]
        \node [every neuron/.try, neuron 1/.try] (i) at (17pt,0pt) {$i$};
        \draw (-30pt,0pt) node (pivot) {\small \color{blue} \textbf{random}\\\color{blue} \textbf{pivot}};
        \node [every neuron/.try, neuron 1/.try] (j) at (52pt,30pt) {$j$};
        \draw (52pt,3pt) node {$\vdots$};
        \node [every neuron/.try, neuron 1/.try] (k) at (52pt,-30pt) {$k$};
        
        \path[very thick] (i) edge (j);
        \path[very thick, dashed] (i) edge (k);
        \path[->, blue] (pivot) edge (i);

        \fill[fill=orange!40] (95pt, 20pt) rectangle ++ (150pt, 20pt);
        \draw[draw=lightgray, line width=0.7pt] (80 pt, 20 pt) rectangle ++ (200pt, 20pt);
        \fill[fill=black] (94.5 pt, 20 pt) rectangle ++ (1pt, 20pt);
        \fill[fill=black] (244.5 pt, 20 pt) rectangle ++ (1pt, 20pt);
        \fill[fill=black] (269.5 pt, 20 pt) rectangle ++ (1pt, 20pt);
        \fill[fill=black] (279.5 pt, 20 pt) rectangle ++ (1pt, 20pt);

        \draw (120pt,3pt) node {$\vdots$};

        \fill[fill=orange!40] (160pt, -40pt) rectangle ++ (90pt, 20pt);
        \draw[draw=lightgray, line width=0.7pt] (80 pt, -40 pt) rectangle ++ (200pt, 20pt);
        \fill[fill=black] (104.5 pt, -40 pt) rectangle ++ (1pt, 20pt);
        \fill[fill=black] (159.5 pt, -40 pt) rectangle ++ (1pt, 20pt);
        \fill[fill=black] (249.5 pt, -40 pt) rectangle ++ (1pt, 20pt);
        \fill[fill=black] (279.5 pt, -40 pt) rectangle ++ (1pt, 20pt);

        \draw (180pt,50pt) node {\small $\{y_{\ell}(i,j)\}_{\ell\in [L]}$, with $\ell^*=2$};
        \draw (180pt,-50pt) node {\small $\{y_{\ell}(i,k)\}_{\ell\in [L]}$, with $\ell^*=3$};

        \draw (230pt, 0pt) node (domlev) {\small \color{orange} $\boldsymbol{\Delta y}$ \textbf{for dominant levels}};
        \draw (205pt, -20pt) node (domik) {};
        \draw (205pt, -30pt) node (deltaik) {\small $\boldsymbol{\Delta y_3(i,k)=0.45}$};
        \draw (170pt, 20pt) node (domij) {};
        \draw (170pt, 30pt) node (deltaij) {\small $\boldsymbol{\Delta y_2(i,j)=0.75}$};
        \path[->, thick, orange] (domlev) edge (domij);
        \path[->, thick, orange] (domlev) edge (domik);

        \draw (340pt, 60pt) node {\color{blue} \large \textbf{rounding}};
        \draw (340pt, 30pt) node {\color{blue} fixed to $d_{2}$};
        \draw (340pt, -30pt) node {\color{blue} sample according\\\color{blue}to $\Delta y_{\ell}(i,k)$};

        % \node [every neuron/.try, neuron 1/.try] (4) at (9,0) {4};
        % \draw (9,-0.8) node {$\eset^4=\{a,c\}$};
        % \node [every neuron/.try, neuron 1/.try] (5) at (12,0) {5};
        % \draw (12,-0.8) node {$\eset^5=\{b,c\}$};
        % \node [every neuron/.try, neuron 1/.try] (6) at (15,0) {6};
        % \draw (15,-0.8) node {$\eset^6=\{b,c\}$};
        % \path[->, very thick, dashed] (5) edge (6);
        % \path[->, very thick, dashed] (1) edge[bend left] (4);
        % \path[->, very thick, dashed] (1) edge[bend left] (3);
        % \path[->, very thick, dashed] (3) edge[bend left] (5);
        % \path[->, very thick, dashed] (5) edge[bend left] (6);
    \end{tikzpicture}
    \caption{Example of our rounding process on complete graphs, where we choose $\alpha=0.4$ and $\beta=0$. }
    \label{fig:rounding-process}
\end{figure}
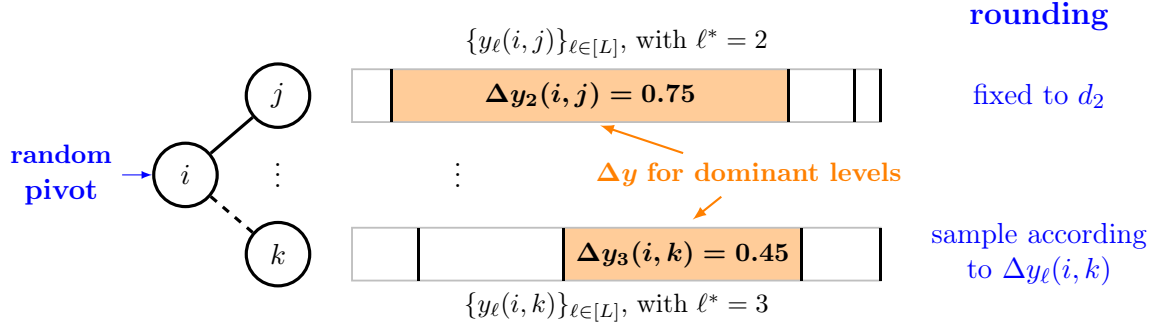

After determining the pivot edges, we partition the non-pivot vertices according to the distances between them and the pivot. 
Suppose $V_{\ell}$ denotes the set of non-pivot vertices $j\in V$ such that $x'(i,j)=d_{\ell}$. 
For any non-pivot edge $(j,k)$ whose endpoints $j,k$ are partitioned into the same set $V_{\ell}$, we minimally fix its distance by setting the distance to %the minimum of $x(j,k)$ and $d_{\ell}$. 
$\min(x(j,k), d_{\ell})$.
For any non-pivot edge $(j,k)$ whose endpoints $j,k$ are partitioned into different sets, say $j\in V_{\ell}, k\in V_{\ell'}$, we minimally fix its distance by setting the distance to %the maximum of $d_{\ell}$ and $d_{\ell'}$. 
$\max(d_{\ell},d_{\ell'})$.
Note that after this step, we can guarantee that the ultrametric inequality is satisfied for every triangle that involves $i$. 
Finally, in the recursive call, for each set $V_{\ell}$, we call our algorithm on $V_{\ell}$ with an upper bound of $d_{\ell}$ on the output distance and modify the distances between vertices in $V_{\ell}$ to the output of the call. 

%  which is an arbitrary most probable value sampled from $W(y,i,j)$, i.e. $\ell^*(y,i,j)\in $. 

Compared to the pivot-based algorithms of~\cite{DBLP:journals/siamcomp/AilonC11,DBLP:conf/focs/Cohen-AddadFLM22}, our algorithm differs only in how the distances on the edges incident to the pivot are determined.
Similar to their analysis, we can show that the algorithm outputs an ultrametric. 
For the sake of completeness, we present the proof in Appendix~\ref{app:output-ultrametric}.
\begin{lemma}
    \label{lem:output-ultrametric}
    Algorithm~\ref{alg:pivot} outputs an ultrametric in polynomial time.
\end{lemma}
% \begin{proof}
%     We shall prove it by induction on a stronger argument: each recursive call returns an ultrametric $x'$ for vertices in $V$ satisfying $x'(i,j)\leq d_u$ for any $(i,j)\in \binom{V}{2}$ in polynomial time.
%     The base case is when $|V|\leq 2$ or $u=L$.
%     The stronger argument holds because the algorithm returns the input $x$ and $x$ is trivially an ultrametric.

%     TBD
% \end{proof}

\section{Analytical Framework}
\label{sec:framework}
In this section, we present the analytical framework for Algorithm~\ref{alg:pivot} that will be used to obtain all the approximation ratios.

\subsection{Definitions and Basic Properties}
\label{subsec:basic-property}
%\Moses{replace terminologies with better word? Preliminaries?}
%\Ruiquan{Can we use ``Definitions''?}
%\Moses{Sounds good}
% In this subsection, we prove several basic properties of Algorithm~\ref{alg:pivot}.

Notice that the recursive calls in the algorithm form a tree structure.
We use the \emph{root call} to refer to the recursive call $\algoname([n], x_{\text{in}}, 1)$ at the beginning of the algorithm.
For any recursive call, we use \emph{child call}s to refer to the recursive calls on line~\ref{line:child-call} and use \emph{parent call} for vice versa.
If a recursive call does not have child calls, we say it is a \emph{leaf call}.

With an abuse of notation, for any LP solution $y$, %we shall ignore the $y$ term in $\ell^*$ if $\ell^*$ appears as a subscript of $y$, i.e. 
we use $y_{\ell^*(u,v)}(i,j)$ to denote $y_{\ell^*(y,u,v)}(i,j)$. 
Let $c^*(i,j)$ denote the LP cost of edge $(i,j)$ in $y^*$, i.e., $c^*(i,j)\defeq 1-\Delta y^*_{\tilde{\ell}({i,j})}(i,j)$. 
With this definition, we can rewrite the LP objective~\eqref{eqn:umvd-lp} and thus lower bound the total number of distances modified in the optimal solution OPT as follows:
\begin{align}
    \label{eqn:optimal-lower-bound}
    \text{OPT} \geq \sum_{(i,j)\in E} w(i,j)\cdot c^*({i,j})~.
\end{align}
% In particular, we shall use $c^*(i,j)$ for the LP cost $c(y^*,i,j)$.
% Let $\ell^*_{i,j}$ be any level $\ell$ that maximizes $\Delta y_{i,j}^{(\ell)}$. 
%\Moses{Should we clarify that the LP objective can be written in terms of the $c^*(i,j)$?}
Based on the LP costs, we shall classify the edges in $E$ into the following two classes: 
\begin{description}
    \item[low-cost edges] include all edges $(i,j)\in E$ satisfying $c^*(i,j)<\alpha$, and
    \item[high-cost edges] include all edges $(i,j)\in E$ satisfying $c^*(i,j)\geq \alpha$.
    % ; note that it implies $\ell^*(y^*,i,j)\neq \ell(i,j)$ and thus $c^*({i,j})\geq \Delta y^*_{\ell^*(i,j)}(i,j) > \frac23$;
    % \item high-error and initially random ($\highrand$): $c^*({i,j})\geq \frac13$ and $\Delta y^*_{\ell^*(i,j)}(i,j) \leq \frac23$.
\end{description}
We shall use $E_{\lowerr}$ and $E_{\higherr}$ to denote the set of low-cost and high-cost edges.
% According to Lemma~\ref{lem:det-in-calls}, low-error edges are deterministic in any recursive call.
% Recall that $\alpha \in [0,0.5]$. 
In the optimal solution of the LP, low-cost edges satisfy
\begin{align*}
    \Delta y^*_{\ell^*(i,j)}(i,j) \geq \Delta y^*_{\tilde{\ell}(i,j)}(i,j) = 1-c^*(i,j) > 1-\alpha~,
\end{align*}
and thus are deterministic in the root call. 
However, this may not be the case for high-cost edges.
We shall subdivide the high-cost edges into the following two classes based on whether they are deterministic in the root call:
\begin{description}
    \item[initially deterministic edges] include all edges $(i,j)\in E_{\higherr}$ such that $\Delta y^*_{\ell^*(i,j)}(i,j)>1-\alpha$,
    \item[initially random edges] include all edges $(i,j)\in E_{\higherr}$ such that $\Delta y^*_{\ell^*(i,j)}(i,j) \leq 1-\alpha$.
\end{description}
We shall use $E_{\highdet}$ and $E_{\highrand}$ to denote the set of high-cost edges that are initially deterministic and that are initially random. 
For the high-cost edges that are initially deterministic, we have the following tighter lower bound for their LP costs.
We defer the proof to Appendix~\ref{app:highdet-error-bound}.
\begin{lemma}
    \label{lem:highdet-error-bound}
    For any edge $(i,j)\in E_{\highdet}$, $c^*(i,j)>1-\alpha$.
\end{lemma}

%\Moses{One thing that was confusing to me in reading this was that $y^*$ is used for the optimal LP solution, and $\ell^*$ is used for the dominant level. In some sense, the meaning of the $*$ superscript is different in these two cases. But it's probably not worth changing notation to fix this.}
\iffalse
\begin{proof}
    Because $(i,j)$ is initially deterministic, if the dominant level $\tilde{\ell}(i,j)=\ell^*(y^*,i,j)$, $c^*(i,j)=1-\Delta y^*_{\ell^*(i,j)}(i,j)<\alpha$, violating that $(i,j)$ is a high-cost edge. Therefore, $\tilde{\ell}(i,j)\neq \ell^*(y^*,i,j)$. Since $y^*_0(i,j)=0$ and $y^*_L(i,j)=1$, $\sum_{\ell=1}^L \Delta y^*_{\ell}(i,j)=1$. Therefore, $c^*(i,j)=1-\Delta y^*_{\tilde{\ell}(i,j)}(i,j) \geq \Delta y^*_{\ell^*(i,j)}(i,j)>1-\alpha$.
\end{proof}
\fi

Next, we present two basic properties of Algorithm~\ref{alg:pivot}. 
The first property claims that the increments of $y$ on the dominant levels $\Delta y_{\ell^*(i,j)}$ are increasing from the root call to the leaf calls.
We defer the proof to Appendix~\ref{app:delta-y-inc}.
%\Moses{One possible question in the reader's mind at this point: When we change the LP solution, can the dominant level change? In other words, is $\ell^*(y,i,j) = \ell^*(y',i,j)$?}
%\Ruiquan{Yes, the dominant level can change.}
%\Moses{So, in the notation $\Delta y'_{\ell^*(i,j)}(i,j)$, we have fixed $\ell^*(i,j) = \ell^*(y^*,i,j)$?}
%\Ruiquan{Oh, we use $y_{\ell^*(u,v)}(i,j)$ to denote $y_{\ell^*(y,u,v)}(i,j)$. So it is not fixed in the subscript of $\Delta y$.}
%\Moses{Just to clarify then, when we say $\Delta y_{\ell^*(i,j)}(i,j)\leq \Delta y'_{\ell^*(i,j)}(i,j)$, we may be referring to different levels for $y$ and $y'$?}
%\Ruiquan{Yes.}
%\Moses{Got it. thanks. Let me think about whether we should clarify this somewhere.}
\begin{lemma}
    \label{lem:delta-y-inc}
    Suppose $y$ is the truncated LP solution in a recursive call with vertex set $V$. Suppose $y'$ is the truncated LP solution in one of its child calls with vertex set $V'\subseteq V$. 
    For each edge $(i,j)\in \binom{V'}{2}$, $\Delta y_{\ell^*(i,j)}(i,j)\leq \Delta y'_{\ell^*(i,j)}(i,j)$. 
    In particular, for any recursive call with vertex set $V$ and truncated LP solution $y$, $\forall (i,j)\in \binom{V}{2}, \Delta y_{\ell^*(i,j)}(i,j)\geq \Delta y^*_{\ell^*(i,j)}(i,j)$.
\end{lemma}
\iffalse
\begin{proof}
Suppose $u$ is the level in the recursive call, and $u'$ is the level in the child call. 
%\Moses{Should we say ``index'' or ``level''?}
%\Ruiquan{I think ``level'' is better. I name $u$ as the input ``upper bound level'' of the recursive calls.}
According to the algorithm, we have $u'\geq u$. 
Fix any edge $(i,j)\in \binom{V'}{2}$. 

%\Moses{Can explain the intuition behind the proof, but maybe not worth it for this simple Lemma.}
If $\ell^*(y,i,j)>u'$, since the algorithm sets both $y_{\ell}(i,j)$ and $y'_{\ell}(i,j)$ to $y^*_{\ell}(i,j)$ for any $\ell\in\{u',u'+1,\cdots,L\}$, we have $\Delta y'_{\ell^*(y,i,j)}(i,j)=\Delta y^*_{\ell^*(y,i,j)}(i,j)=\Delta y_{\ell^*(y,i,j)}(i,j)$.
Therefore, according to the definition of the dominant level, $\Delta y'_{\ell^*(i,j)}(i,j)\geq \Delta y'_{\ell^*(y,i,j)}(i,j) = \Delta y_{\ell^*(i,j)}(i,j)$. 

Notice that $\Delta y_{\ell}(i,j)\leq y^*_{\ell}(i,j)$ for any $\ell\in [L]$ and $\Delta y'_{u'}(i,j)=y^*_{u'}(i,j)$. 
On the other hand, if $\ell^*(y,i,j)\leq u'$, because of Lemma~\ref{lem:feasible-y-in-calls} and the LP constraint~\eqref{eqn:umvd-lp-increasing}, we have $\Delta y'_{u'}(i,j)=y^*_{u'}(i,j)\geq y^*_{\ell^*(y,i,j)}(i,j) \geq \Delta y_{\ell^*(i,j)}(i,j)$ and thus $\Delta y'_{\ell^*(i,j)}(i,j)\geq \Delta y_{\ell^*(i,j)}(i,j)$.

In particular, in the root call, because $u=1$, $y^*$ equals the truncated LP solution. Therefore, for any $(i,j)\in \binom{[n]}{2}$, $\Delta y_{\ell^*(i,j)}(i,j)\geq \Delta y^*_{\ell^*(i,j)}(i,j)$.
\end{proof}
\fi
\noindent From the above lemma and the fact that any edge $(i,j)\in E_{\lowerr}\cup E_{\highdet}$ satisfies $\Delta y^*_{\ell^*(i,j)}(i,j)>1-\alpha$, we obtain the following corollary.
\begin{corollary}
\label{cor:low-err-are-det}
In any recursive call, any low-cost (or initially deterministic high-cost) edge is deterministic. 
\end{corollary}

Because we choose $\alpha\leq \frac{1}{2}$, we can give the second property that low-cost edges are modified only when they are non-pivot edges.
% The proof of this lemma follows $\alpha\leq 0.5$.
We defer the proof to Appendix~\ref{app:low-error-modify}.
\begin{lemma}
\label{lem:low-error-modify}
In any recursive call with vertex set $V$ and truncated LP solution $y$, for any edge $(i,j)\in E_{\lowerr}\cap \binom{V}{2}$, the input distance satisfies $x(i,j)=d_{\ell^*(y,i,j)}$.
% Suppose $\alpha\geq \frac12$.
In particular, when an edge $(i,j)\in E_{\lowerr}$ appears as a pivot edge, its distance $x(i,j)$ is not modified in the recursive call. 
\end{lemma}

\subsection{Towards the Approximation Ratios: Improved Triple-based Analysis for Ultrametric Violation Distance}
\label{subsec:improved-triple}
In this subsection, we develop two key lemmas,~\ref{lem:approx-ratio} and~\ref{lem:bound-of-sum-B}, to prove the approximation ratios for the three cases and any constant $\alpha\in (0,\frac{1}{2}]$.
By contrast to the triple-based analysis of the previous works (e.g.,~\cite{DBLP:journals/jacm/AilonCN08,DBLP:conf/stoc/ChawlaMSY15}), we separately upper bound the cost of the algorithm on different classes of edges: for the high-cost edges, we %can simply 
upper bound the algorithm's cost by $1$ and further by a constant multiple of its LP cost; for low-cost edges, we shall utilize the previous triple-based analysis.
In the triple-based analysis, Lemma~\ref{lem:approx-ratio} establishes the approximation ratio given a scheme charging the non-pivot edges so that: (1) in each triangle $t=(i,j,k)$ the expected number of modifications on low-cost edges in $t$ can be upper bounded by the expected total charges in $t$; and (2) the expected total charges on each edge can be upper bounded.
In addition, Lemma~\ref{lem:bound-of-sum-B} establishes an upper bound for the expected total charges on each edge by the probability that the edge can no longer be charged in the recursive algorithm (conditioned on any possible charging value).
%\Moses{Note for intro: What aspects of the analysis should we emphasize? What is novel and different from other papers?}

%\Moses{Say something about the fact that we establish two key lemmas here and how they are used in the analysis.}

Let ALG be the total number of distances modified by the algorithm. 
Because of the definition of high-cost edges and Lemma~\ref{lem:highdet-error-bound}, we can upper bound its expectation as follows:
% \begin{align*}
%     \expect[\text{ALG}]
%     &
%     = 
%     \sum_{(i,j)\in E} \Pr[\text{edge $(i,j)$ is modified}]
%     \\
%     &
%     \leq 
%     \sum_{(i,j)\in E_{\lowerr}} \Pr[\text{edge $(i,j)$ is modified}] + |E_{\higherr}|
%     \\
%     &
%     \leq
%     \sum_{(i,j)\in E_{\lowerr}} \Pr[\text{edge $(i,j)$ is modified}] + 3\cdot \sum_{(i,j)\in E_{\higherr}} c^*(i,j)
% \end{align*}
\begin{align}
    % \begin{split}
    \expect[\text{ALG}]
    &
    =
    \sum_{(i,j)\in E} w(i,j)\cdot \Pr[\text{$(i,j)$ is modified}]
    \notag
    \\ 
    &
    \leq 
    \sum_{(i,j)\in E_{\lowerr}} w(i,j)\cdot \Pr[\text{$(i,j)$ is modified}] + \sum_{(i,j)\in E_{\highdet}} w(i,j) + \sum_{(i,j)\in E_{\highrand}} w(i,j)
    \notag
    \\
    \begin{split}
    &
    \leq
    \sum_{(i,j)\in E_{\lowerr}} w(i,j)\cdot \expect[\text{\#times $(i,j)$ is modified}] + \frac{1}{1-\alpha} \cdot \sum_{(i,j)\in E_{\highdet}} w(i,j)\cdot c^*(i,j)
    \\
    & \qquad\qquad\qquad\qquad\qquad\qquad\qquad\qquad\qquad\quad\;\;\,
    + \frac{1}{\alpha} \cdot \sum_{(i,j)\in E_{\highrand}} w(i,j)\cdot c^*(i,j)~.
    \end{split}
    \label{eqn:bound-E-alg}
    % \end{split}
\end{align}
% On the other hand, we can lower bound the total number of distances modified in the optimal solution OPT by the LP value of the optimal solution $y^*$ of~\eqref{eqn:umvd-lp}:
% \begin{align}
%     \label{eqn:optimal-lower-bound}
%     \text{OPT} \geq \sum_{(i,j)\in E} w(i,j)\cdot c^*({i,j})~.
% \end{align}
Recall the lower bound for the optimal solution~\eqref{eqn:optimal-lower-bound}.
%we have formulated in~\eqref{eqn:optimal-lower-bound}.
To prove the approximation ratio, we only need to upper bound the number of modifications the algorithm makes on those low-cost edges by the LP value of $y^*$. Because of Lemma~\ref{lem:low-error-modify}, low-cost edges are only modified when they appear as non-pivot edges in the recursive calls. Hence, we only need to consider the modifications on non-pivot edges. 

For each triangle $t=(i,j,k)$, let $\mathcal{A}_t$ denote the event that $i,j,k$ appear in the same recursive call where one of them is chosen as a pivot. 
For each triangle $t$ and each edge $(i,j)\in t$, we define $M_{i,j,t}=1$ if $\mathcal{A}_t$ happens, $(i,j)\in E_{\lowerr}$, and $(i,j)$ is modified in the event $\mathcal{A}_t$.
Otherwise, we define $M_{i,j,t}=0$.
With this definition, we can rewrite the total number of modifications on low-cost edges as
\begin{align}
    \label{eqn:ub-of-modify-times}
    \sum_{(i,j)\in E_{\lowerr}} w(i,j)\cdot (\text{\#times $(i,j)$ is modified})
    =
    \sum_{(i,j)\in E_{\lowerr}} w(i,j)\cdot \sum_{t:i,j\in t} M_{i,j,t}~.
\end{align}

\paragraph{The charging scheme.}
To upper bound the sum of $M_{i,j,t}$ by the LP value of the optimal solution $y^*$, we introduce intermediate random variables $B_{i,j,t}$, which specify how much we charge the non-pivot edge in $\mathcal{A}_t$.
More specifically, for any triangle $t\in \binom{[n]}{3}$ and any edge $(i,j)\in t\cap E$, we will charge $(i,j)$ in triangle $t$ if both of the following conditions are satisfied:
\begin{itemize}
    \item the event $\mathcal{A}_t$ happens (i.e., three vertices of $t$ appear in the same recursive call where one of them is chosen as a pivot), and
    \item the edge $(i,j)$ is the non-pivot edge in the event $\mathcal{A}_t$.
\end{itemize}
% if $(i,j)\in E$, the event $\mathcal{A}_t$ happens and $(i,j)$ is the non-pivot edge in the event $\mathcal{A}_t$, 
When both conditions are satisfied, we assign a non-negative value to $B_{i,j,t}$ and use $B_{i,j,t}\cdot w(i,j)\cdot c^*(i,j)$ to denote how much $(i,j)$ is charged in triangle $t$. 
Note that since the event $\mathcal{A}_t$ can only happen in at most one recursive call, $B_{i,j,t}$ cannot be assigned twice and is thus well defined.
Otherwise, we assign $0$ to $B_{i,j,t}$ for convenience.
In the analysis, we want two properties for this charging scheme:
\begin{enumerate}
    \item we can upper bound the expectation of expression~\eqref{eqn:ub-of-modify-times} by the expected sum of $B_{i,j,t}\cdot w(i,j)\cdot c^*(i,j)$, and
    \item fixing an edge $(i,j)$, we can upper bound the expectation of the sum of $B_{i,j,t}$ over all $t\owns i,j$.
\end{enumerate}
Next, we present the key Lemma~\ref{lem:approx-ratio} in this section: given a charging scheme (of $B_{i,j,t}$s) satisfying the above two properties, we can establish an upper bound on the approximation ratio of Algorithm~\ref{alg:pivot}.

\begin{lemma}
    \label{lem:approx-ratio}
    Suppose there exists a charging scheme such that 
    \begin{enumerate}
        \item for any triangle $t\in \binom{[n]}{3}$, 
        \begin{align}
            \label{eqn:improved-M-bound-by-B}
            \sum_{(i,j)\in t} w(i,j) \cdot \expect[M_{i,j,t}|\mathcal{A}_t]\leq \sum_{(i,j)\in t\cap E} \expect[B_{i,j,t}|\mathcal{A}_t]\cdot w(i,j)\cdot c^*(i,j)~;
        \end{align}
        \item for any edge $(i,j)\in E$, the total charges on $(i,j)$ satisfies
        \begin{align}
            \label{eqn:total-edge-charge}
            \sum_{t:i,j\in t} \expect[B_{i,j,t}] \leq \begin{cases}
                \overline{B}_{\lowerr} & \text{if $(i,j)\in E_{\lowerr}$,}\\
                \overline{B}_{\highdet} & \text{if $(i,j)\in E_{\highdet}$,}\\
                \overline{B}_{\highrand} & \text{if $(i,j)\in E_{\highrand}$.}
            \end{cases}
        \end{align}
    \end{enumerate}
    Then, Algorithm~\ref{alg:pivot} is a $\max\{\overline{B}_{\lowerr},~\overline{B}_{\highdet}+\frac{1}{1-\alpha},~\overline{B}_{\highrand}+\frac{1}{\alpha}\}$-approximation.
\end{lemma}
% With the two properties, we can now prove this section's main theorem.
\begin{proof}
    According to the definition, when $\mathcal{A}_t$ does not happen, $M_{i,j,t}=B_{i,j,t}=0$ for any $(i,j)\in t$. 
    Therefore, $\expect[M_{i,j,t}]=\expect[M_{i,j,t}|\mathcal{A}_t]\cdot \Pr[\mathcal{A}_t]$ and $\expect[B_{i,j,t}]=\expect[B_{i,j,t}|\mathcal{A}_t]\cdot \Pr[\mathcal{A}_t]$.
    By multiplying $\Pr[\mathcal{A}_t]$ on both sides of Eqn.~\eqref{eqn:improved-M-bound-by-B} and summing over all possible $t\in \binom{[n]}{3}$, we have
    \begin{align*}
        \sum_{(i,j)\in E_{\lowerr}} w(i,j) \cdot \sum_{t:i,j\in t} \expect[M_{i,j,t}] 
        &
        \leq 
        \sum_{(i,j)\in E} \bigg(\sum_{t:i,j\in t} \expect[B_{i,j,t}]\bigg) \cdot w(i,j)\cdot c^*(i,j)
        ~.
        % \\
        % &
        % =
        % \sum_{(i,j)\in E} \bigg(\sum_{t:i,j\in t} \expect[B_{i,j,t}]\bigg) \cdot w(i,j) \cdot c^*(i,j)
        % ~.
        % \tag{$\forall (i,j)\notin E, B_{i,j,t}\defeq 0$}
    \end{align*}
    % According to the definition of $M_{i,j,t}$, whenever $(i,j)$ is low-error, $\sum_{t:(i,j)\in t} \expect[M_{i,j,t}]$ is the number of times the algorithm modifies it. The left-hand side of the inequality is an upper bound of 
    % \begin{align*}
    % \sum_{(i,j)\in E_{\lowerr}} \Pr[\text{edge $(i,j)$ is modified}]~.
    % \end{align*}
    % Hence, the right-hand side of the inequality is its upper bound. 
    According to Eqn.~\eqref{eqn:optimal-lower-bound},~\eqref{eqn:bound-E-alg},~\eqref{eqn:ub-of-modify-times}, and~\eqref{eqn:total-edge-charge}, 
    \begin{align*}
        \expect[ALG] 
        &
        \leq 
        \overline{B}_{\lowerr} \sum_{(i,j)\in E_{\lowerr}} w(i,j) \cdot c^*(i,j)
        + 
        \overline{B}_{\highdet} \sum_{(i,j)\in E_{\highdet}} w(i,j) \cdot c^*(i,j)
        + 
        \overline{B}_{\highrand} \sum_{(i,j)\in E_{\highrand}} w(i,j) \cdot c^*(i,j)
        \\
        &
        \qquad\qquad\qquad\qquad\qquad\qquad\, + 
        \frac{1}{1-\alpha} \sum_{(i,j)\in E_{\highdet}} w(i,j)\cdot c^*(i,j) 
        + 
        \frac{1}{\alpha} \sum_{(i,j)\in E_{\highrand}} w(i,j)\cdot c^*(i,j)
        \\
        &
        \leq
        \max\Big\{\overline{B}_{\lowerr},~\overline{B}_{\highdet}+\frac{1}{1-\alpha},~ \overline{B}_{\highrand}+\frac{1}{\alpha}\Big\} \sum_{(i,j)\in E} w(i,j) \cdot c^*(i,j) 
        \\
        &
        \leq 
        \max\Big\{\overline{B}_{\lowerr},~\overline{B}_{\highdet}+\frac{1}{1-\alpha},~ \overline{B}_{\highrand}+\frac{1}{\alpha}\Big\} \cdot \text{OPT}
    \end{align*}
    Therefore, Algorithm~\ref{alg:pivot} is a $\max\{\overline{B}_{\lowerr},~\overline{B}_{\highdet}+\frac{1}{1-\alpha},~\overline{B}_{\highrand}+\frac{1}{\alpha}\}$-approximation.
\end{proof}
    
\paragraph{Bounding the total charges.} Next, we present a method to upper bound the total charge $\sum_{t:(i,j)\in t} \expect[B_{i,j,t}]$ on each edge by lower bounding the probability that an edge can no longer be charged after each time it is charged. 
Similar analysis can be found in previous works~\cite{DBLP:journals/jacm/AilonCN08,DBLP:conf/stoc/ChawlaMSY15} on \CClust.
In Algorithm~\ref{alg:pivot}, when the endpoints of an edge are partitioned into different sets on line~\ref{line:partition-non-pivots} in a recursive call, the edge will no longer be charged in its child calls.
Then, the key lemma for the upper bound can be formulated as follows.
% \Ruiquan{is it better to present something in common in Section~\ref{sec:ratio-complete} and~\ref{sec:ratio-k-part}?}
% Fix an edge (i,j)\in E and any q>0. 
% Let R be a recursive call, b be a positive number, and t be a triangle containing i,j, such that event \mathcal{A}_t happens and (i,j) is a non-pivot edge charged by B_{i,j,t}=b. If for any valid choice of R,b,t the probability that ...

\begin{lemma}%[\bf Bounding total charge on an edge]
    \label{lem:bound-of-sum-B}
    Fix an edge $(i,j)\in E$ and any $q>0$.
    Let $\mathcal{R}$ be a recursive call, $b>0$ be a positive number, and $t$ be a triangle containing $i,j$ such that $\mathcal{A}_t$ happens in $\mathcal{R}$ and $(i,j)$ is a non-pivot edge charged by $B_{i,j,t}=b$.
    % If for any $b>0$ and triangle $t\owns i,j$, in any recursive call, conditioned on $\mathcal{A}_t$ and that $(i,j)$ is a non-pivot edge charged by $B_{i,j,t}=b$, 
    If for any valid choice of $\mathcal{R}, b, t$ the probability that $i$ and $j$ are partitioned into different sets on line~\ref{line:partition-non-pivots} of Algorithm~\ref{alg:pivot} is at least $q\cdot b$,
    then the expected total charges on edge $(i,j)$ satisfy
    \[
        \sum_{t:(i,j)\in t} \expect[B_{i,j,t}] \leq q^{-1}~.
    \]
\end{lemma}
\begin{proof}
    We consider a stronger argument on the calls: \emph{for any call $\algoname(V,x,u)$ such that $i,j\in V$, $\sum_{t\in\binom{V}{3}:(i,j)\in t} \expect[B_{i,j,t}]\leq q^{-1}$.}
    We shall prove it by induction in a bottom-up way. 
    The base case is when $|V|=2$ or $u=L$. 
    Because the algorithm returns the input distance at the beginning of the call, according to the definition, $B_{i,j,t}=0$ for any $i,j \in t\in \binom{V}{3}$.
    
    Consider any call $\algoname(V,x,u)$ such that $|V|\geq 3$ and $u<L$. 
    Consider any $(i,j)\in \binom{V}{2}$. 
    Suppose that the stronger argument holds for any $\algoname(V',x',u')$ satisfying $|V'|<|V|$. 
    If $i$ or $j$ is chosen as the pivot vertex in the call, $B_{i,j,t}=0$ for any $i,j\in t \in \binom{V}{3}$ because $(i,j)$ is not a non-pivot edge and $i, j$ do not simultaneously appear in any of its children calls.  
    We prove the stronger argument for this case.
    Otherwise, suppose that the pivot vertex is $k\neq i,j$.
    Let $t=(i,j,k)$.
    In this case, the expected total charge of $(i,j)$ equals the sum of its expected charge in triangle $t$ and its expected total charge in the child call involving both $i,j$ (if exists):
    \begin{align*}
        \sum_{t'\in \binom{V}{3}:i,j\in t'} \expect[B_{i,j,t'}] = \expect[B_{i,j,t}] + \sum_{\ell\geq u} \Pr[i,j\in V_{\ell}] \cdot \expect\Big[\sum\nolimits_{t'\in \binom{V_{\ell}}{3}:i,j\in t'} B_{i,j,t'} \;\Big|\; i,j\in V_{\ell}\Big]
    \end{align*}
    For any $b>0$, conditioning on $B_{i,j,t}=b$, we have $\sum_{\ell\geq u} \Pr[i,j\in V_{\ell}]\leq 1-q\cdot b$ according to the assumption of the lemma. 
    According to the induction hypothesis, (conditioning on $B_{i,j,t}=b$) $\expect\big[\sum\nolimits_{t'\in \binom{V_{\ell}}{3}:i,j\in t'} B_{i,j,t'} \big| i,j\in V_{\ell}\big]\leq q^{-1}$ for any possible $V_{\ell}\subset V$.
    Therefore, the stronger argument holds in this case because
    \begin{align*}
        \sum_{t'\in \binom{V}{3}:i,j\in t'} \expect[B_{i,j,t'}] 
        &
        \leq 
        \sum_{b\geq 0} \Pr[B_{i,j,t}=b]\cdot \big(b+(1-q\cdot b)\cdot q^{-1}\big)
        \\
        &
        =
        \sum_{b\geq 0} \Pr[B_{i,j,t}=b]\cdot q^{-1}
        =
        q^{-1}~.
    \end{align*}
\end{proof}

% \begin{remark}
%     In Lemma~\ref{lem:bound-of-sum-B}, if $q \cdot b>1$, it is impossible to have $B_{i,j,t}=b$.
% \end{remark}

\section{5-Approximation Analysis for Complete Graphs}
\label{sec:ratio-complete}
% \Ruiquan{In this section, there are many $\alpha$s to be modified to $1-\alpha$.}
In this section, we prove that Algorithm~\ref{alg:pivot} is a 5-approximation algorithm for unweighted \UMVD~on complete graphs by setting $\beta=0$ and %tuning the parameter 
an appropriate choice of $\alpha$. 
More specifically, we will prove the following theorem:
\begin{theorem}
    \label{thm:5-approx-complete}
    If $\alpha\in \big[\frac{3-\sqrt{5}}{2},0.5\big], \beta=0$, Algorithm~\ref{alg:pivot} is a randomized polynomial-time $\max\{\frac{3}{1-\alpha}, \frac{2}{\alpha}\}$-approximation algorithm for \UMVD~on complete graphs. 
    In particular, with $\alpha=0.4$, it is a randomized polynomial-time $5$-approximation for \UMVD~on complete graphs. 
\end{theorem}

Recall that we have $E=\binom{[n]}{2}, E_{\varnothing}=\emptyset$ and $\forall (i,j)\in E, w(i,j)=1$ in this case. 
For convenience, in the rest of this section, we assume $\beta=0$.
Upon this assumption, the CCDF function~\eqref{eqn:distance-ccdf} for random edges in $E$ becomes:
\begin{align}
    \label{eqn:complete-distance-ccdf}
    \forall \,\ell\in [L],
    \quad 
    \Pr[x'(i,j)\geq d_{\ell}] = y_\ell(i,j)~.
\end{align}

To prove Theorem~\ref{thm:5-approx-complete}, we shall use Lemma~\ref{lem:approx-ratio}, ~\ref{lem:bound-of-sum-B} and the following Lemma~\ref{lem:edge-charging-for-B}, which presents a charging scheme that meets the two conditions we need for the approximation ratio. 
\begin{lemma}
    \label{lem:edge-charging-for-B}
    If $\alpha\in \big[\frac{3-\sqrt{5}}{2},0.5\big]$, there exists a charging scheme on complete graphs such that 
    \begin{enumerate}
        \item for any triangle $t\in \binom{[n]}{3}$, 
        \begin{align*}
            % \label{eqn:improved-M-bound-by-B}
            \sum_{(i,j)\in t} \expect[M_{i,j,t}|\mathcal{A}_t]\leq \sum_{(i,j)\in t} \expect[B_{i,j,t}|\mathcal{A}_t]\cdot c^*(i,j)~.
        \end{align*}
        \item for any $b>0$, $(i,j)\in E$ and any $t\owns i,j$, in any recursive call, conditioning on $\mathcal{A}_t$ and that $(i,j)$ is a non-pivot edge charged by $B_{i,j,t}=b$, the probability that $i$ and $j$ are partitioned into different sets on line~\ref{line:partition-non-pivots} of Algorithm~\ref{alg:pivot} is at least $q_{i,j}\cdot b$, where 
            \begin{align*}
                q_{i,j} 
                =
                \begin{cases}
                   \min\{\frac{1-\alpha}{3},\frac{\alpha}{2}\}  & \text{if $(i,j)\in E_{\lowerr}$}\\
                   \frac{1-\alpha}{2}  & \text{if $(i,j)\in E_{\highdet}$}\\
                   \alpha  & \text{if $(i,j)\in E_{\highrand}$}
                \end{cases}
            \end{align*}
    \end{enumerate}
\end{lemma}
Assuming the correctness of the above lemma, we can obtain the proof of Theorem~\ref{thm:5-approx-complete}.
\begin{proof}[Proof of Theorem~\ref{thm:5-approx-complete}]
Because of Lemma~\ref{lem:bound-of-sum-B} and the second bullet of Lemma~\ref{lem:edge-charging-for-B}, we can present the second condition of Lemma~\ref{lem:approx-ratio} as follows:
\begin{align*}
\forall (i,j)\in E, \quad 
\sum_{t:i,j\in t} \expect[B_{i,j,t}]
\leq
\begin{cases}
    \max\{\frac{3}{1-\alpha}, \frac{2}{\alpha}\} & \text{if $(i,j)\in E_{\lowerr}$}~,\\
    \frac{2}{1-\alpha}  & \text{if $(i,j)\in E_{\highdet}$}~,\\
    \frac{1}{\alpha}  & \text{if $(i,j)\in E_{\highrand}$}~.
\end{cases}
\end{align*}
Because the first bullet of Lemma~\ref{lem:edge-charging-for-B} meets that of Lemma~\ref{lem:approx-ratio}, if $\alpha\in \big[\frac{3-\sqrt{5}}{2},0.5\big]$, Algorithm~\ref{alg:pivot} is a $\max\{\frac{3}{1-\alpha}, \frac{2}{\alpha}\}$-approximation for \UMVD.
\end{proof}

% \subsection{The edge-charging scheme - Proof of Lemma~\ref{lem:edge-charging-for-B}}
In the rest of this section, we present the charging scheme and prove Lemma~\ref{lem:edge-charging-for-B}. 
The charging scheme defines $B_{i,j,t}$ differently for different classes of triangles. 
Recall that given the truncated solution $y$, we say an edge is deterministic if $\Delta y_{\ell^*(i,j)}(i,j)>1-\alpha$, and otherwise, we say it is random.
The triangles are classified according to the number of deterministic edges and the dominant levels of the deterministic edges in the triangles:
\begin{description}
    \item[$(d,d,d)$-triangles] have three deterministic edges.
    \item[$(d,d,r)$-same-triangles] have two deterministic edges and one random edge, and the dominant levels of the deterministic edges are the same.
    \item[$(d,d,r)$-diff-triangles] have two deterministic edges and one random edge, and the dominant levels of the deterministic edges are different.
    \item[$(d,r,r)$-triangles] have one deterministic edge and two random edges.
    \item[$(r,r,r)$-triangles] have three random edges.
\end{description}
%\Moses{Perhaps the proofs of most of these cases should be in the Appendix?}
%\Ruiquan{I put the discussion of two cases into Appendix~\ref{app:edge-charging-for-B} and say something about the reason (after Lemma~\ref{lem:prob-disappear}).}

Next, we present three useful lemmas. 
In the discussion of different classes, we have the following observation: the probability that a low-cost edge is modified is highly related to the values of $y_{\ell^*(i,j)-1}(i,j)$ and $1-\Delta y_{\ell^*(i,j)}(i,j)$ on each edge $(i,j)\in t$ of the triangle. 
The first two lemmas give upper bounds for these terms based on the LP contributions $c^*(i,j)$ of the edges and will be extensively used to prove the first bullet of Lemma~\ref{lem:edge-charging-for-B}.
We note that the proofs of these two lemmas are independent of $E$ and $w(i,j)$ and can work for all three cases.
\begin{lemma}
\label{lem:stronger-y-bounded-by-c*}
In any recursive call with vertex set $V$ and truncated LP solution $y$, for any edge $(i,j)\in \binom{V}{2}$, $1-\Delta y_{\ell^*(i,j)}(i,j)$ can be upper bounded by
\begin{align*}
    \begin{cases}
        c^*(i,j) & \text{if $(i,j)\in E_{\lowerr} \cup E_{\highrand}$}~,\\
        \frac{\alpha}{1-\alpha}\cdot c^*(i,j) & \text{if $(i,j)\in E_{\highdet}$}~.
    \end{cases}
\end{align*} 
In particular, $y_{\ell^*(i,j)-1}(i,j)$ and $1-y_{\ell^*(i,j)}(i,j)$ satisfy the same upper bound. 
\end{lemma}
\begin{proof}
For any $(i,j)\in \binom{[n]}{2}$, because of Lemma~\ref{lem:delta-y-inc}, we have
\begin{align*}
    1-\Delta y_{\ell^*(i,j)}(i,j) 
    \leq 
    1-\Delta y^*_{\ell^*(i,j)}(i,j)
    =
    c^*(i,j)~.
\end{align*}
In particular, for any $(i,j)\in E_{\highdet}$, we have
\begin{align*}
    1-\Delta y_{\ell^*(i,j)}(i,j) 
    &
    \leq 
    1-\Delta y^*_{\ell^*(i,j)}(i,j)
    \tag{Lemma~\ref{lem:delta-y-inc}}
    \\
    &
    < 
    \alpha
    \tag{$(i,j)$ is initially deterministic}
    \\
    &
    <
    \frac{\alpha}{1-\alpha}\cdot c^*(i,j)~.
    \tag{Lemma~\ref{lem:highdet-error-bound}}
\end{align*}
In particular, the second lemma gives an improved bound for a certain type of edge under some conditions, which helps us to refine the analysis.
\begin{lemma}
    \label{lem:y-bounded-by-0}
    When an edge $(i,j)\in E_{\highrand}$ appears as deterministic in some recursive call with truncated LP solution $y$, $y_{\ell^*(i,j)-1}(i,j)=0$.
\end{lemma}
\begin{proof}
Because $(i,j)\in E_{\highrand}$, $\forall \ell\in [L], \Delta y^*_{\ell}(i,j)\leq 1-\alpha$. 
Consider any call $\algoname(V, x, u)$ with truncated LP solution $y$. 
For any $\ell>u$, $\Delta y_{\ell}(i,j)=\Delta y^*_{\ell}(i,j)\leq 1-\alpha$. 
Since $(i,j)$ is deterministic in the call and $y_{\ell}(i,j)=0$ for $\ell<u$, the dominant level $\ell^*(y,i,j)$ can only be $u$. 
Therefore, $y_{\ell^*(i,j)-1}(i,j)=0$.
\end{proof}

In particular, because $1-\Delta y_{\ell^*(i,j)}(i,j) = y_{\ell^*(i,j)-1}(i,j)+1-y_{\ell^*(i,j)}(i,j)$ and $\forall \ell\in [L], y_{\ell}(i,j)\in [0,1]$, $y_{\ell^*(i,j)-1}(i,j)$ and $1-y_{\ell^*(i,j)}(i,j)$ satisfy the same upper bound. 
\end{proof}
The third lemma gives lower bounds for the probabilities that the endpoints of a non-pivot edge are partitioned into different sets on line~\ref{line:partition-non-pivots} of the algorithm, which will be used for the second bullet of Lemma~\ref{lem:edge-charging-for-B}.
The proof of this lemma follows on the fact that the distance on a random edge can only equal a fixed value with probability at most $1-\alpha$, i.e., $\Delta y_{\ell}(i,j)\leq 1-\alpha$ for any random $(i,j)$ and $\ell\in [L]$.
\begin{lemma}
    \label{lem:prob-disappear}
    Consider any recursive call $\algoname(V, x, u)$ with $|V|>2$.
    Suppose $i$ is the pivot vertex of the call. 
    For any $(j,k)\in \binom{V\setminus\{i\}}{2}$, the probability that $j,k$ are partitioned into different sets on line~\ref{line:partition-non-pivots} can be lower bounded by
    \begin{itemize}
        \item $1$ if both $(i,j)$ and $(i,k)$ are deterministic in the call but they have different dominant levels, or
        \item $\alpha$ if at least one of $(i,j)$ and $(i,k)$ is random.
    \end{itemize}
\end{lemma}
\begin{proof}
    For the first bullet, according to the algorithm, the distances of $(i,j)$ and $(i,k)$ are deterministically different, and thus $j,k$ are partitioned into different sets with probability $1$.

    For the second bullet, suppose that the truncated LP solution in the call is $y$. 
    W.l.o.g., we assume that $(i,j)$ is random. 
    Note that this implies $\forall \ell\in [L], \Delta y_{\ell}(i,j)\leq \Delta y_{\ell^*(i,j)}(i,j)\leq 1-\alpha$.
    According to the CCDF~\eqref{eqn:complete-distance-ccdf}, the probability that $j$ and $k$ are partitioned into different sets on line~\ref{line:partition-non-pivots} is
    \begin{align*}
        1-\sum_{\ell\in [L]} \Pr[j,k\in V_{\ell}] 
        &
        = 
        1-\sum_{\ell\in [L]} \Pr[x'(i,j)=d_{\ell}]\cdot \Pr[x'(i,k)=d_{\ell}]
        \\
        &
        =
        1-\sum_{\ell\in [L]} \Delta y_{\ell}(i,j)\cdot \Pr[x'(i,k)=d_{\ell}]
        \\
        &
        \geq 
        1-\sum_{\ell\in [L]} (1-\alpha)\cdot \Pr[x'(i,k)=d_{\ell}]
        \\
        &
        =
        1-(1-\alpha)
        =
        \alpha
        ~.
    \end{align*}
\end{proof}

Next, we present the charging scheme and prove Lemma~\ref{lem:edge-charging-for-B} for each class of triangles.
Due to Corollary~\ref{cor:low-err-are-det}, there is no low-cost edge in $(r,r,r)$-triangles.
Therefore, we only need to discuss the remaining four classes of triangles.
Consider any triangle $t=(i,j,k)$ and any recursive call with a vertex set involving $i,j,k$ and truncated LP solution $y$.
For convenience, we will condition on $\mathcal{A}_t$ in the rest of the analysis in this section, and all the expectations will be automatically conditioned on $\mathcal{A}_t$. 
Note that the two pivot edges in $(d,r,r)$-triangles and $(d,d,r)$-diff-triangles always satisfy one of the conditions in Lemma~\ref{lem:prob-disappear}.
The endpoints of the non-pivot edges in these triangles will be partitioned into different sets with a constant probability (under the choice of $\alpha\in [\frac{3-\sqrt{5}}{2},0.5]$).
Hence, it should not be surprising to prove Lemma~\ref{lem:edge-charging-for-B} for these triangles, and we shall defer their proofs to Appendix~\ref{app:edge-charging-for-B}.

Recall our charging scheme. 
Conditioning on $\mathcal{A}_t$, we charge an edge in $t$ only when the edge is non-pivot. 
If an edge $(i',j')$ is not charged in $t$, we automatically assign $0$ to $B_{i',j',t}$.
For convenience, we will only specify $B_{i',j',t}$ when $(i',j')$ is the non-pivot edge in the rest of the analysis.

\paragraph{$(d,d,d)$-triangles.}
% \label{subsec:ddd-triangles}
Because all three edges in the triangle are deterministic, we have $\Delta y_{\ell^*(i,j)}(i,j)$, $\Delta y_{\ell^*(i,k)}(i,k)$, $\Delta y_{\ell^*(j,k)}(j,k)> 1-\alpha$. 
% We can prove $M_{i,j,t}=M_{i,k,t}=M_{j,k,t}=0$ by showing that $w_{\ell^*(y,i,j)}, w_{\ell^*(y,i,k)}, w_{\ell^*(y,j,k)}$ satisfies the ultrametric inequality. Note that the violation of the ultrametric inequality happens only when two weights are strictly less than the other one. Since the edges are deterministic, we have
Note that this implies that
\begin{align}
\label{eqn:y-of-a-det-edge}
    \forall i'\neq j' \in t, %\{i,j,k\}, 
    \quad
    y_{\ell^*(i',j')-1}(i',j') < \alpha, 
    \;
    y_{\ell^*(i',j')}(i',j') > 1-\alpha~.
\end{align}

Note that the distance of any pivot edge $(u,v)\in t$ is set to its dominant distance level $d_{\ell^*(y,u,v)}$.
Because of Lemma~\ref{lem:low-error-modify}, the input distance of any low-cost edge $(u,v)\in t$ is also $d_{\ell^*(y,u,v)}$.
Therefore, a low-cost edge is modified only when it is non-pivot in the call and the ultrametric inequality is violated on the three dominant distance levels, i.e., there exists a permutation $(i',j',k')$ of $i,j,k$ such that $\ell^*(y,i',j')<\min\{\ell^*(y,j',k'), \ell^*(y,k',i')\}$ (equivalently, $d_{\ell^*(y,i',j')}>\max\{d_{\ell^*(y,j',k')}, d_{\ell^*(y,k',i')}\}$).

If there is no violation of the ultrametric inequality on these three deterministic levels, $M_{i,j,t}=M_{i,k,t}=M_{j,k,t}=0$. 
For this case, we define $B_{i,j,t},B_{i,k,t},B_{j,k,t}\defeq 0$, and Lemma~\ref{lem:edge-charging-for-B} then clearly holds. 

Otherwise, w.l.o.g., we assume that $\ell^*(y,j,k)<\min\{\ell^*(y,i,j), \ell^*(y,i,k)\}$. 
In this case, for any permutation $(i',j',k')$ of $i,j,k$, 
\[
    M_{i',j',t}=\ind(\text{$k'$ is the pivot vertex})\cdot \ind((i',j')\in E_{\lowerr}).
\]
Because each vertex in $t$ is chosen as the pivot with equal probability $1/3$ conditioning on $\mathcal{A}_t$, in this case,
\begin{align*}
    \sum_{(i',j')\in t} \expect[M_{i',j',t}] = \frac{1}3 \cdot \Big(\ind((i,j)\in E_{\lowerr}) + \ind((i,k)\in E_{\lowerr}) + \ind((j,k)\in E_{\lowerr})\Big)~.
\end{align*}
Because $y$ is feasible in~\eqref{eqn:umvd-lp} and because of Lemma~\ref{lem:stronger-y-bounded-by-c*},
\begin{align*}
c^*(i,j) + c^*(i,k)
&
\geq 
y_{\ell^*(i,j)-1}(i,j)+y_{\ell^*(i,k)-1}(i,k) 
\tag{Lemma~\ref{lem:stronger-y-bounded-by-c*}}
\\
&
\geq 
y_{\ell^*(j,k)}(i,j)+y_{\ell^*(j,k)}(i,k) 
\tag{LP constraint~\eqref{eqn:umvd-lp-increasing}}
\\
&
\geq
y_{\ell^*(j,k)}(j,k) 
>
1-\alpha
~.
\tag{LP constraint~\eqref{eqn:umvd-lp-triangle} \& Eqn.~\eqref{eqn:y-of-a-det-edge}}
\end{align*}
\begin{Remark}
If $\alpha\leq \frac{1}{3}$, $y_{\ell^*(i,j)-1}(i,j)+y_{\ell^*(i,k)-1}(i,k) \leq  2\alpha\leq 1-\alpha$ and one can easily find a contradiction. 
This implies no violation of the ultrametric inequality in this type of triangle if the parameter $\alpha$ is less than $\frac13$. 
\end{Remark}
Because of Lemma~\ref{lem:y-bounded-by-0}, Eqn.~\eqref{eqn:y-of-a-det-edge} and $\alpha\leq \frac{1}{2}$, if $(i,j)\in E_{\highrand}$ or $(i,k)\in E_{\highrand}$, $y_{\ell^*(i,j)-1}(i,j)+y_{\ell^*(i,k)-1}(i,k)<\alpha\leq 1-\alpha$, which violates the above inequality. 
Therefore, edges $(i,j),(i,k)\notin E_{\highrand}$.
Based on this fact, when the corresponding edge is non-pivot in the recursive call, we shall define $B_{i,j,k},B_{i,k,t},B_{j,k,t}$ as follows:
\begin{align*}
    B_{i,j,t}
    \defeq
    \begin{cases}
        \frac{3}{1-\alpha} & \text{if }(i,j)\in E_{\lowerr}\\
        \frac{2}{1-\alpha} & \text{if }(i,j)\in E_{\highdet}\\
    \end{cases}
    ~, 
    \quad
    B_{i,k,t}
    \defeq
    \begin{cases}
        \frac{3}{1-\alpha} & \text{if }(i,k)\in E_{\lowerr}\\
        \frac{2}{1-\alpha} & \text{if }(i,k)\in E_{\highdet}\\
    \end{cases}
    ~,
    \quad
    B_{j,k,t}
    \defeq
    0
    ~,
\end{align*}
which implies the first bullet of Lemma~\ref{lem:edge-charging-for-B} for this case:
\begin{align*}
    \sum_{(i',j')\in t} \expect[B_{i',j',t}]\cdot c^*(i',j')
    &
    =
    \frac{2+\ind((i,j)\in E_{\lowerr})}{3(1-\alpha)}\cdot c^*(i,j)
    +
    \frac{2+\ind((i,k)\in E_{\lowerr})}{3(1-\alpha)}\cdot c^*(i,k)
    \\
    &
    \geq
    \frac{\ind((i,j)\in E_{\lowerr}) + \ind((i,k)\in E_{\lowerr}) + \ind((j,k)\in E_{\lowerr})}{3(1-\alpha)} \cdot \big(c^*(i,j)+c^*(i,k)\big)
    \\
    &
    >
    \frac{\ind((i,j)\in E_{\lowerr}) + \ind((i,k)\in E_{\lowerr}) + \ind((j,k)\in E_{\lowerr})}{3} 
    =
    \sum_{(i',j')\in t} \expect[M_{i',j',t}]~,
\end{align*}
Further, when the edge $(i,j)$ (or $(i,k)$) is non-pivot, two pivot edges in the triangle are both deterministic but have different dominant levels. 
Because of Lemma~\ref{lem:prob-disappear}, the probability $i,j$ (or $i,k$) are then partitioned into different sets on line~\ref{line:partition-non-pivots} is at least $1$, which is at least
\begin{align*}
    \begin{cases}
        \frac{1-\alpha}{3}\cdot B_{i,j,t} & \text{if $(i,j)\in E_{\lowerr}$}\\
        \frac{1-\alpha}{2}\cdot B_{i,j,t} & \text{if $(i,j)\in E_{\highdet}$}
    \end{cases}
    \quad 
    \text{and}
    \quad
    \begin{cases}
        \frac{1-\alpha}{3}\cdot B_{i,k,t} & \text{if $(i,k)\in E_{\lowerr}$}\\
        \frac{1-\alpha}{2}\cdot B_{i,k,t} & \text{if $(i,k)\in E_{\highdet}$}
    \end{cases}
    ~.
\end{align*}
This definition satisfies the second bullet of Lemma~\ref{lem:edge-charging-for-B} for this case.

% Further, because the dominant level of $(j,k)$ is different with those of $(i,j)$ and $(i,k)$, according to Lemma~\ref{lem:prob-disappear}, if $j$ (or $k$) is the pivot vertex in the call, $i,k$ (resp., $i,j$) are partitioned into different sets on line~\ref{line:partition-non-pivots} with probability $1$. 
% Note that when edge $(i,j)$ (or $(i,k)$) is non-pivot, $1$ equals $\frac{\alpha}{3}\cdot B_{i,k,t}$ if $(i,k)\in E_{\lowerr}$, or equals $\frac{\alpha}{2}\cdot B_{i,k,t}$ if $(i,k)\in E_{\highdet}$ (resp., for $B_{i,j,t}$), 
% the defintion satisfies the second bullet of Lemma~\ref{lem:edge-charging-for-B} for this case.

\paragraph{$(d,d,r)$-same-triangles.} 
% \label{subsec:ddr-same-triangles}
W.l.o.g., we assume $(i,j)$ and $(i,k)$ are the deterministic edges. 
Because of Lemma~\ref{lem:low-error-modify} and $\ell^*(y,i,j)=\ell^*(y,i,k)$, the input distances of edges $(i,j)$ and $(i,k)$ are both $d_{\ell^*(y,i,j)}$.
Therefore, $M_{i,j,t}$ (or $M_{i,k,t}$) equals $1$ only when the pivot is $k$ (resp., $j$), it is low-cost, and the random distance $x'(j,k)>d_{\ell^*(y,i,j)}$, i.e.,
\begin{align*}
    M_{i,j,t} &= \ind(\text{$k$ is the pivot vertex})\cdot \ind((i,j)\in E_{\lowerr}) \cdot \ind(x'(j,k)>d_{\ell^*(y,i,j)})~,
    \\
    M_{i,k,t} &= \ind(\text{$j$ is the pivot vertex})\cdot \ind((i,k)\in E_{\lowerr}) \cdot \ind(x'(j,k)>d_{\ell^*(y,i,j)})~.
\end{align*}
According to the CCDF~\eqref{eqn:complete-distance-ccdf} of the random distance, the probability that $x'({j,k})>d_{\ell^*(y,i,j)}$ is 
\begin{align}
    y_{\ell^*(i,j)-1}(j,k) 
    &
    \leq 
    y_{\ell^*(i,j)-1}(i,j) + y_{\ell^*(i,j)-1}(i,k)
    \tag{LP constraint~\eqref{eqn:umvd-lp-triangle}}
    \\
    &
    =
    y_{\ell^*(i,j)-1}(i,j) + y_{\ell^*(i,k)-1}(i,k)~.
    \label{eqn:ddr-same-modify-prob}
    % \\
    % &
    % \leq 
    % c(i,j) + c(i,k)
    % \leq 
    % c^*(i,j) + c^*(i,k)
\end{align}
Because each vertex in $t$ is chosen as the pivot with equal probability $1/3$ conditioning on $\mathcal{A}_t$, in this class of triangles,
\begin{align*}
    \sum_{(i',j')\in t} \expect[M_{i',j',t}] \leq \frac{2(y_{\ell^*(i,j)-1}(i,j) + y_{\ell^*(i,k)-1}(i,k))}{3}~.
\end{align*}
% \begin{align*}
%     \sum_{i'\neq j'\in t} \expect[M_{i',j',t}] \leq \frac{y_{\ell^*(i,j)-1}(i,j) + y_{\ell^*(i,k)-1}(i,k)}{3} \cdot \Big(\ind((i,j)\in E_{\lowerr}) + \ind((i,k)\in E_{\lowerr})\Big)~.
% \end{align*}

Accordingly, when the corresponding edge is non-pivot, we define $B_{i,j,t},B_{i,k,t},B_{j,k,t}$ as follows:
\begin{align*}
    B_{i,j,t}
    \defeq
    \begin{cases}
     2 & \text{if } (i,j)\in E_{\lowerr}\\
     \frac{2\alpha}{1-\alpha} & \text{if }(i,j)\in E_{\highdet}\\
     0 & \text{if } (i,j)\in E_{\highrand}
    \end{cases}
    ~, 
    \quad
    B_{i,k,t} 
    \defeq
    \begin{cases}
     2 & \text{if } (i,k)\in E_{\lowerr}\\
     \frac{2\alpha}{1-\alpha} & \text{if }(i,k)\in E_{\highdet}\\
     0 & \text{if } (i,k)\in E_{\highrand}
    \end{cases}
    ~,
    \quad
    B_{j,k,t}
    \defeq
    0
\end{align*}
% \begin{align*}
%     B_{i,j,t}, ~
%     B_{i,k,t} 
%     \defeq
%     \begin{cases}
%      2 & \text{if } (i,k)\in E_{\lowerr}\\
%      \frac{\alpha}{1-\alpha} & \text{if }(i,k)\in E_{\highdet}\\
%      0 & \text{if } (i,k)\in E_{\highrand}
%     \end{cases},
%     \quad\quad\quad
%     B_{j,k,t}
%     \defeq
%     0
% \end{align*}
Because of Lemma~\ref{lem:stronger-y-bounded-by-c*} and~\ref{lem:y-bounded-by-0}, this definition implies the first bullet of Lemma~\ref{lem:edge-charging-for-B}: 
\begin{align*}
    \sum_{(i',j')\in t}
    \expect[B_{i',j',t}] \cdot c^*(i',j')
    &
    \geq
    \frac{2}{3} \cdot y_{\ell^*(i,j)-1}(i,j) 
    +
    \frac{2}{3} \cdot y_{\ell^*(i,k)-1}(i,k)
    \geq
    \sum_{(i', j')\in t} \expect[M_{i',j',t}]~.
\end{align*} 
% \begin{align*}
%     \sum_{i'\neq j'\in t}
%     \expect[B_{i',j',t}] \cdot c^*(i',j')
%     &
%     \geq
%     \frac{1+\ind(\text{$(i,j)$ is $\lowerr$})}{3} \cdot y_{\ell^*(i,j)-1}(i,j) 
%     +
%     \frac{1+\ind(\text{$(i,j)$ is $\lowerr$})}{3} \cdot y_{\ell^*(i,k)-1}(i,k)
%     \\
%     &
%     \geq
%     \frac{\ind((i,j)\in E_{\lowerr})+\ind((i,k)\in E_{\lowerr})}{3} \cdot \big(y_{\ell^*(i,j)-1}(i,j) + y_{\ell^*(i,k)-1}(i,k)\big)
%     \\
%     &
%     =
%     \sum_{i'\neq j'\in t} \expect[M_{i',j',t}]~.
% \end{align*} 

Further, when the edge $(i,j)$ (or $(i,k)$) is non-pivot, one of the pivot edges in the triangle is random.
Because of Lemma~\ref{lem:prob-disappear}, the probability $i,j$ (or $i,k$) are then partitioned into different sets on line~\ref{line:partition-non-pivots} is at least $\alpha$, which is at least
\begin{align*}
    \begin{cases}
        \frac{\alpha}{2}\cdot B_{i,j,t} & \text{if $(i,j)\in E_{\lowerr}$}\\
        \frac{1-\alpha}{2}\cdot B_{i,j,t} & \text{if $(i,j)\in E_{\highdet}$}\\
        1\cdot B_{i,j,t} & \text{if $(i,j)\in E_{\highrand}$}
    \end{cases}
    \quad
    \text{and}
    \quad
    \begin{cases}
        \frac{\alpha}{2}\cdot B_{i,k,t} & \text{if $(i,k)\in E_{\lowerr}$}\\
        \frac{1-\alpha}{2}\cdot B_{i,k,t} & \text{if $(i,k)\in E_{\highdet}$}\\
        1\cdot B_{i,k,t} & \text{if $(i,k)\in E_{\highrand}$}
    \end{cases}
    ~.
\end{align*}
% \begin{align*}
%     \begin{cases}
%         \frac{\alpha}{2}\cdot B_{i,j,t} & \text{if $(i,j)\in E_{\lowerr}$}\\
%         (1-\alpha)\cdot B_{i,j,t} & \text{if $(i,j)\in E_{\highdet}$}\\
%         1\cdot B_{i,j,t} & \text{if $(i,j)\in E_{\highrand}$}
%     \end{cases}
%     ~,
%     \quad 
%     \text{ or }
%     \begin{cases}
%         \frac{\alpha}{2}\cdot B_{i,k,t} & \text{if $(i,k)\in E_{\lowerr}$}\\
%         (1-\alpha)\cdot B_{i,k,t} & \text{if $(i,k)\in E_{\highdet}$}\\
%         1\cdot B_{i,j,t} & \text{if $(i,j)\in E_{\highrand}$}
%     \end{cases}
% \end{align*}
This definition satisfies the second bullet of Lemma~\ref{lem:edge-charging-for-B} for this class of triangles.

\section{$\boldsymbol{O(\min\{L,\log n\})}$-Approximation Analysis for Weighted Instances with Triangle Inequality Constraints}
\label{sec:ratio-s-weighted}
In this section, we prove that Algorithm~\ref{alg:pivot} is an $O(\min\{L,\log n\})$-approximation algorithm for weighted \UMVD~with triangle inequality constraints by setting $\beta=0$ and %tuning the parameter 
an appropriate choice of $\alpha$. 
More specifically, we will prove the following theorem:
\begin{theorem}
\label{thm:logn-app-s-weighted}
If $\alpha=\frac{1}{3},\beta=0$, Algorithm~\ref{alg:pivot} is a randomized polynomial-time $O(\min\{L,\log n\})$-approximation algorithm for weighted \UMVD~where weights satisfy the triangle inequality.
\end{theorem}

In Section~\ref{sec:prelim}, we have w.l.o.g. assumed that $E=\binom{[n]}{2}$ for weighted instances. 
Because of Lemma~\ref{lem:approx-ratio}, the proof of Theorem~\ref{thm:logn-app-s-weighted} directly follows the following lemma.

\begin{lemma}
    \label{lem:s-weighted-edge-charging-for-B}
    If $\alpha=\frac{1}{3}, \beta=0$, there exists a charging scheme on weighted cases with triangle inequality constraints such that 
    \begin{enumerate}
        \item for any triangle $t\in \binom{[n]}{3}$, 
        \begin{align*}
            % \label{eqn:improved-M-bound-by-B}
            \sum_{(i,j)\in t} w(i,j) \cdot \expect[M_{i,j,t}|\mathcal{A}_t]\leq \sum_{(i,j)\in t} \expect[B_{i,j,t}|\mathcal{A}_t]\cdot w(i,j) \cdot c^*(i,j)~.
        \end{align*}
        \item for any edge $(i,j)\in \binom{[n]}{2}$,
        $
            \sum_{t:i,j\in t} \expect[B_{i,j,t}] = O(\min\{L,\log n\})~.
        $
    \end{enumerate}
\end{lemma}

In the rest of this section, we will prove Lemma~\ref{lem:s-weighted-edge-charging-for-B}.
For convenience, we assume $\alpha=\frac{1}{3}$ and $\beta=0$.
Because $\beta=0$, we can still use the lemmas and inequalities established in Section~\ref{sec:ratio-complete}, including Lemma~\ref{lem:stronger-y-bounded-by-c*},~\ref{lem:y-bounded-by-0},~\ref{lem:prob-disappear} and inequalities~\eqref{eqn:y-of-a-det-edge},~\eqref{eqn:ddr-same-modify-prob}.%,~\eqref{eqn:ddr-diff-modify-prob}.

We classify the charge $B_{i,j,t}$ into two types.
Let $B_{i,j,t}^{(1)}$ and $B_{i,j,t}^{(2)}$ denote charges of the two types, where $B_{i,j,t}^{(1)}=B_{i,j,t},B_{i,j,t}^{(2)}=0$ if the charge is classified to the first type and $B_{i,j,t}^{(1)}=0,B_{i,j,t}^{(2)}=B_{i,j,t}$ if the charge is classified as the second type.
With the classification, 
the total charge on $(i,j)$ can be rewritten as the sum of its first-type charges and its second-type charges, i.e., $\sum_{t:i,j\in t} B_{i,j,t}^{(1)} + \sum_{t:i,j\in t} B_{i,j,t}^{(2)}$.
Therefore, we can upper bound for total charges of different types to upper bound the total charges:
\begin{align*}
    \sum_{t:i,j\in t} \expect[B_{i,j,t}^{(1)}]&=O(1)
    \quad\text{and}\quad
    % \label{eqn:k-weighted-1st-charge}
    \sum_{t:i,j\in t} \expect[B_{i,j,t}^{(2)}]=O(\min\{L,\log n\})~.
    % \label{eqn:k-weighted-2nd-charge}
\end{align*}
For the first type, we will use the same method in Section~\ref{sec:ratio-complete} by showing the condition in Lemma~\ref{lem:bound-of-sum-B} with $q=\Omega(1)$.
For the second type, we will use the technique of~\cite{DBLP:conf/focs/Cohen-AddadFLM22} to analyze the $O(\log\{L,\log n\})$ approximation ratio of their pivot-based algorithm.

Similarly, as in Section~\ref{sec:ratio-complete}, we prove the lemma by presenting the charging scheme for each class of triangles.
Consider any triangle $t=(i,j,k)$ and any recursive call with a vertex set involving all of $i,j,k$.
Suppose that the recursive call has an upper bound level $u$ and a truncated LP solution $y$.
For convenience, we will condition on $\mathcal{A}_t$ in the rest of this section, and all expectations will be automatically conditioned on $\mathcal{A}_t$. 
Further, we similarly defer the proofs for $(d,d,r)$-diff-triangles and $(d,r,r)$-triangles to Appendix~\ref{app:s-weighted-edge-charging-for-B}.

\paragraph{$(d,d,d)$-triangles.} Note that the distance of any pivot edge $(u,v)\in t$ is set to its dominant distance level $d_{\ell^*(y,u,v)}$.
Because of Lemma~\ref{lem:low-error-modify}, any low-cost edge $(u,v)$ has an input distance $x(u,v)$ equal to the dominant distance level $d_{\ell^*(y,u,v)}$. 
If the non-pivot edge is low-cost, it is modified only when the ultrametric inequality is violated at the three dominant distance levels, i.e., there exists a permutation $(i',j',k')$ of $i,j,k$ such that $\ell^*(y,i',j')<\min\{\ell^*(y,i',k'),\ell^*(y,j',k')\}$. 
However, $\ell^*(y,i',j')<\min\{\ell^*(y,i',k'),\ell^*(y,j',k')\}$ implies
\begin{align}
    \alpha
    &
    > 
    1 - y_{\ell^*(i',j')}(i',j')
    \tag{Eqn.~\eqref{eqn:y-of-a-det-edge}}
    \\
    &
    \geq 
    1 - y_{\ell^*(i',j')}(i',k') - y_{\ell^*(i',j')}(j',k') 
    \tag{LP constraint~\eqref{eqn:umvd-lp-triangle}}
    \\
    &
    \geq 
    1 - y_{\ell^*(i',k')-1}(i',k') - y_{\ell^*(j',k')-1}(j',k')
    \tag{LP constraint~\eqref{eqn:umvd-lp-increasing}}
    \\
    &
    >
    1 - 2\alpha
    \tag{Eqn.~\eqref{eqn:y-of-a-det-edge}}
    ~,
\end{align}
which violates our choice of parameter $\alpha=\frac{1}{3}$.
Therefore, the low-cost non-pivot edges cannot be modified, and we have $M_{i,j,t}=M_{i,k,t}=M_{j,k,t}=0$ in this case.
Accordingly, by always defining $B_{i,j,t},B_{i,k,t},B_{j,k,t}\defeq 0$ and classifying the charges into the first type, we show the first bullet of Lemma~\ref{lem:s-weighted-edge-charging-for-B} and the condition in Lemma~\ref{lem:bound-of-sum-B} for this class of triangles.

\paragraph{$(d,d,r)$-same-triangles.} W.l.o.g. we assume $(i,j)$ and $(i,k)$ are deterministic.
Because of Corollary~\ref{cor:low-err-are-det}, we have $(j,k)\in E_{\higherr}$, $M_{j,k,t}=0$, and $c^*(j,k)\geq \alpha$. 
Recall that we have shown for this class of triangles in Section~\ref{sec:ratio-complete}: 
\begin{align*}
    M_{i,j,t} &= \ind(\text{$k$ is the pivot vertex})\cdot \ind((i,j)\in E_{\lowerr}) \cdot \ind(x'(j,k)>d_{\ell^*(y,i,j)})~,
    \\
    M_{i,k,t} &= \ind(\text{$j$ is the pivot vertex})\cdot \ind((i,k)\in E_{\lowerr}) \cdot \ind(x'(j,k)>d_{\ell^*(y,i,j)})~,
\end{align*}
and the probability that $x'(j,k)>d_{\ell^*(y,i,j)}$ is at most $y_{\ell^*(i,j)-1}(i,j) + y_{\ell^*(i,k)-1}(i,k)$. 
Because of Lemma~\ref{lem:stronger-y-bounded-by-c*}, $y_{\ell^*(i,j)-1}(i,j)\leq c^*(i,j)$, $y_{\ell^*(i,k)-1}(i,k)\leq c^*(i,k)$.
Because edges $(i,j)$ and $(i,k)$ are deterministic and Eqn.~\eqref{eqn:y-of-a-det-edge}, $y_{\ell^*(i,j)-1}(i,j), y_{\ell^*(i,k)-1}(i,k)< \alpha \leq c^*(j,k)$.
Therefore, in this class of triangles, we have 
\begin{align*}
    \sum_{(i', j')\in t} w(i',j')\cdot \expect[M_{i',j',t}] 
    &
    \leq 
    \frac{y_{\ell^*(i,j)-1}(i,j) + y_{\ell^*(i,k)-1}(i,k)}{3} \cdot \big(w(i,j)+w(i,k)\big)
    \\
    &
    \leq 
    \frac{w(i,j)\cdot c^*(i,j)+w(i,k)\cdot c^*(i,k)}{3} 
    + \frac{y_{\ell^*(i,j)-1}(i,j)\cdot w(i,k)}{3} 
    \\
    &
    \quad 
    + \frac{y_{\ell^*(i,k)-1}(i,k)\cdot w(i,j)}{3}
    \\
    &
    \leq 
    \frac{w(i,j)\cdot c^*(i,j)+w(i,k)\cdot c^*(i,k)}{3} 
    + \frac{y_{\ell^*(i,j)-1}(i,j)\cdot (w(i,j)+w(j,k))}{3} 
    \\
    &
    \quad 
    + \frac{y_{\ell^*(i,k)-1}(i,k)\cdot (w(i,k)+w(j,k))}{3}
    \tag{triangle inequality for $w$}
    \\
    &
    \leq 
    \frac{w(i,j)\cdot c^*(i,j)+w(i,k)\cdot c^*(i,k)}{3} 
    + 
    \frac{c^*(i,j)\cdot w(i,j)+c^*(j,k)\cdot w(j,k)}{3} 
    \\
    &
    \quad 
    + 
    \frac{c^*(i,k)\cdot w(i,k)+c^*(j,k)\cdot w(j,k)}{3}
    \\
    &
    = 
    \frac{2\big(w(i,j)\cdot c^*(i,j)+w(j,k)\cdot c^*(j,k) + w(i,k)\cdot c^*(i,k)\big)}{3} 
    ~.
\end{align*}
On the other hand, if $\ell^*(y,i,j)=u$, according to the definition of the truncated LP solution, $y_{\ell^*(i,j)-1}(i,j)=0$.
In this case, $\sum_{(i', j')\in t} w(i',j')\cdot \expect[M_{i',j',t}]=0$.
Accordingly, when the corresponding edge is non-pivot, we can define $B_{i,j,t},B_{i,k,t},B_{j,k,t}\defeq 2\cdot \ind(\ell^*(y,i,j)>u)$ to prove the first bullet of Lemma~\ref{lem:s-weighted-edge-charging-for-B}.
Furthermore, we classify the charges on $(i,j)$ and $(i,k)$ into the first type, while we classify the charge on $(j,k)$ into the second type. 
Note that in this class of triangles, when $(i,j)$ or $(i,k)$ is non-pivot, there is a random pivot edge.
Because of Lemma~\ref{lem:prob-disappear}, when being non-pivot, $i,j$ (or $i,k$) will be partitioned into different sets on line~\ref{line:partition-non-pivots} with probability at least $\alpha=\frac{1}{3}$, which is at least $\Omega(1)\cdot B_{i,j,t}$ (resp., $\Omega(1) \cdot B_{i,k,t}$). 
Therefore, this definition for $(i,j)$ and $(i,k)$ satisfies the condition of Lemma~\ref{lem:bound-of-sum-B} with $q=\Omega(1)$. 
Next, we will upper bound the expected total charge of the second type on the random edges in this class of triangles.

\paragraph{Bounding the total charges of the second type.}
Fix an edge $(j,k)\in E$.
In any recursive call with an upper bound level $u$, if $B_{j,k,t}^{(2)}\neq 0$, where the pivot in the call is $i$ and $t=(i,j,k)$, we have $\ell^*(y,i,j)>u$.
At the same time, the algorithm modifies the distance of $(j,k)$ to $d_{\ell^*(y,i,j)}$ and partitions $j,k$ to a child call with upper bound level $u'=\ell^*(y,i,j)>u$. 
Therefore, the recursive calls that charge $(j,k)$ for the second type must have distinct upper bound levels $u$, and thus $\sum_{t:(i,j)\in t} B_{j,k,t}^{(2)} \leq O(L)$.

On the other hand, we have the following characterization for the dominant levels of each edge in each recursive call: the dominant level equals either its dominant level in $y^*$ or the upper bound level $u$ of the call. 
Its proof follows the fact that the truncated LP solution satisfies $\forall \ell<u, \Delta y_{\ell}(i,j)=0$, $\Delta y_{u}(i,j)=y^*_u(i,j)$ and $\forall \ell>u, \Delta y_{\ell}(i,j)=\Delta y^*_{\ell}(i,j)$ in a call with upper bound level $u$.
Furthermore, because descendant calls have higher upper bound levels, if the dominant level of edge $(i,j)$ does not equal $\ell^*(y^*,i,j)$ in a call, it will no longer equal $\ell^*(y^*,i,j)$ in all its descendant calls.
Formally, all previous arguments can be summarized by the following lemma.
\begin{lemma}
    \label{lem:bad-pivots}
    In any recursive call with an upper bound level $u$, the dominant distance level of any edge $(i,j)$ satisfies $\ell^*(y,i,j)=u$ or $\ell^*(y,i,j)=\ell^*(y^*,i,j)$.
    In particular, if $\ell^*(y,i,j)=u$ in any recursive call with truncated LP solution $y$ and upper bound level $u$, $\ell^*(y',i,j)=u'$ in any of its descendant calls with truncated LP solution $y'$ and upper bound $u'$.
\end{lemma}

% Consider a fixed edge $(j,k)$.
In a recursive call, we call vertices $i\neq j,k$ a bad pivot of $(j,k)$ if $\ell^*(y,i,j)=\ell^*(y,i,k)>u$. 
Otherwise, we call $i$ a good pivot.
Because of Lemma~\ref{lem:bad-pivots}, any bad pivot $i$ of $(j,k)$ satisfies $\ell^*(y,i,j)=\ell^*(y^*,i,j)$, and good pivots in a call never become bad in its descendant calls. 
Further, if the pivot $i$ is bad in the recursive call, $j,k$ will appear in a child call with an upper bound level $u' = \ell^*(y^*,i,j)$.
Therefore, any bad pivot $i'$ with $\ell^*(y^*,i',j)\geq \ell^*(y^*,i,j)$ will become good (or disappear) in the descendant calls.
Let $\overline{B}^{(2)}(n')$ denote an upper bound for the expected total charges of the second type on $(j,k)$ in a recursive call and its descendant calls, where there are at most $n'$ bad pivots for the edge.
It is clear that $\sum_{t:j,k\in t}\expect[B^{(2)}_{j,k,t}]\leq \overline{B}^{(2)}(n-2)$. 
We claim that $\overline{B}^{(2)}(0)=0$, and that $\overline{B}^{(2)}(n')\leq 4\ln{n'} + 2$ for any $n'\geq 1$ to finish proving that $\sum_{t:j,k\in t}\expect[B^{(2)}_{j,k,t}]=O(\log n)$.
We prove the claim by induction.
The base case is when $n'=0$.
It is clear that we have $\overline{B}^{(2)}(0)=0$ because there is no bad pivot and there is no charge of the second type on $(j,k)$. 
For $n'\geq 1$, 
let $n'_{\ell}$ be the number of bad pivots $i$ with $\ell^*(y,i,j)\geq \ell$ in the recursive. 
Based on our previous discussion, we have 
\begin{align*}
    \overline{B}^{(2)}(n') 
    &
    \leq
    2 + \sum_{\ell>u} \frac{n'_{\ell}-n'_{\ell+1}}{n'} \cdot \overline{B}^{(2)}(n'_{\ell+1})
    \\
    &
    \leq
    2 + \sum_{i\in [n']} \frac{1}{n'}\cdot \overline{B}^{(2)}(i-1)
    \\
    &
    \leq 
    2 + \sum_{i\in [n'-1]} \frac{4\ln{i}+2}{n'} 
    \\
    &
    \leq
    4 - \frac{2}{n'} + \frac{4}{n'} \int_{1}^{n'} \ln{x} \;dx
    \\
    &
    =
    4 - \frac{2}{n'} + \Big(4\ln{n'} - 4 + \frac{4}{n'}\Big)
    \leq
    4\ln{n'}+2~.
\end{align*}

\section{16-Approximation Analysis for k-Partite Graphs}
\label{sec:ratio-k-part}
% \Ruiquan{I know a (slightly) better $8+4\sqrt{3}\approx 14.93$-approximation for this case, with $\alpha=\frac{\sqrt{3}-1}{2}$ and $\beta=1/\alpha-2$. However, it would make the proof full of square roots of $3$. Do we need it?}

In this section, we prove that Algorithm~\ref{alg:pivot} is a $16$-approximation algorithm for unweighted \UMVD~on complete $k$-partite graphs by appropriate choices of the parameters $\alpha,\beta$.
More specifically, we will prove the following theorem:
\begin{theorem}
\label{thm:16-approx-k-part}
If $\alpha=\frac{3}{8}, \beta=\frac{2}{3}$, Algorithm~\ref{alg:pivot} is a randomized polynomial-time $16$-approximation algorithm for \UMVD ~on complete $k$-partite graphs.
\end{theorem}

% And, the triangles are classified similarly according to the number of deterministic edges and the dominant levels of the deterministic edges in the triangles:
% \begin{description}
%     \item[$(d,d,d)$-triangles] have three deterministic edges.
%     \item[$(d,d,r)$-same-triangles] have two deterministic edges and one random edge, and the dominant levels of the deterministic edges are the same.
%     \item[$(d,d,r)$-diff-triangles] have two deterministic edges and one random edge, and the dominant levels of the deterministic edges are different.
%     \item[$(d,r,r)$-triangles] have one deterministic edge and two random edges.
%     \item[$(r,r,r)$-triangles] have three random edges.
% \end{description}
We shall first prove the following analog Lemma~\ref{lem:k-prob-disappear0} of Lemma~\ref{lem:prob-disappear}.
To ensure that endpoints of a non-pivot edge are partitioned into different sets on line~\ref{line:partition-non-pivots} when one of the pivot edges is random, which is a frequently used property in the proof in Section~\ref{sec:ratio-complete}, 
our choice of the parameter $\beta$ follows the criteria $\alpha\beta$ and $\frac{\alpha(1-\beta)}{1-\alpha\beta}$ can both be lower bounded by some constant.

\begin{lemma}
    \label{lem:k-prob-disappear0}
    Consider any recursive call $\algoname(V, x, u)$ with $|V|>2$.
    Suppose $i$ is the pivot vertex of the call. 
    For any $(j,k)\in \binom{V\setminus\{i\}}{2}$, the probability that $j,k$ are partitioned into different sets on line~\ref{line:partition-non-pivots} can be lower bounded by
    \begin{itemize}
        \item $1$ if both $(i,j)$ and $(i,k)$ are deterministic in the call but they have different dominant levels, 
        \item $\alpha\beta$ if one of $(i,j)$ and $(i,k)$ is random in $E_{\varnothing}$, or
        \item $\frac{\alpha(1-\beta)}{1-\alpha\beta}$ if one of $(i,j)$ and $(i,k)$ is random in $E$.
    \end{itemize}
\end{lemma}
\begin{proof}
    For the first bullet, according to the algorithm, the distances of the pivot edges $(i,j)$ and $(i,k)$ are deterministically different, and thus $j,k$ are partitioned into different sets with probability $1$ on line~\ref{line:partition-non-pivots}.
    
    For the second bullet, suppose the truncated LP solution in the call is $y$. 
    W.l.o.g., we assume that $(i,j)$ is random in $E_{\varnothing}$. 
    Note that this implies $\forall \ell\in [L], \Delta y_{\ell}(i,j)\leq \Delta y_{\ell^*(i,j)}(i,j)\leq 1-\alpha\beta$.
    According to the CCDF~\eqref{eqn:empty-distance-ccdf}, the probability that $j,k$ are partitioned into different sets on line~\ref{line:partition-non-pivots} is
    \begin{align*}
        1-\sum_{\ell\in [L]} \Pr[j,k\in V_{\ell}] 
        &
        = 
        1-\sum_{\ell\in [L]} \Pr[x'(i,j)=d_{\ell}]\cdot \Pr[x'(i,k)=d_{\ell}]
        \\
        &
        =
        1-\sum_{\ell\in [L]} \Delta y_{\ell}(i,j)\cdot \Pr[x'(i,k)=d_{\ell}]
        \\
        &
        \geq 
        1-\sum_{\ell\in [L]} (1-\alpha\beta)\cdot \Pr[x'(i,k)=d_{\ell}]
        \\
        &
        =
        1-(1-\alpha\beta)
        =
        \alpha\beta
        ~.
    \end{align*}
    
    For the third bullet, suppose the truncated LP solution in the call is $y$. 
    W.l.o.g., we assume that $(i,j)$ is random in $E$. 
    Note that this implies $\forall \ell\in [L], \Delta y_{\ell}(i,j)\leq \Delta y_{\ell^*(i,j)}(i,j)\leq 1-\alpha$.
    According to the CCDF~\eqref{eqn:distance-ccdf}, the probability that $j,k$ are partitioned into different sets on line~\ref{line:partition-non-pivots} is
    \begin{align*}
        1-\sum_{\ell\in [L]} \Pr[j,k\in V_{\ell}] 
        &
        = 
        1-\sum_{\ell\in [L]} \Pr[x'(i,j)=d_{\ell}]\cdot \Pr[x'(i,k)=d_{\ell}]
        \\
        &
        =
        1-\sum_{\ell\in [L]} \frac{(y_{\ell}(i,j)-\alpha\beta)^+ - (y_{\ell-1}(i,j)-\alpha\beta)^+}{1-\alpha\beta} \cdot \Pr[x'(i,k)=d_{\ell}]
        \\
        &
        \geq 
        1-\sum_{\ell\in [L]} \frac{\Delta y_{\ell}(i,j)}{1-\alpha\beta} \cdot \Pr[x'(i,k)=d_{\ell}]
        \tag{LP constraint~\eqref{eqn:umvd-lp-increasing}}
        \\
        &
        \geq 
        1-\sum_{\ell\in [L]} \frac{1-\alpha}{1-\alpha\beta} \cdot \Pr[x'(i,k)=d_{\ell}]
        \\
        &
        =
        1-\frac{1-\alpha}{1-\alpha\beta}
        =
        \frac{\alpha(1-\beta)}{1-\alpha\beta}
        % =
        % \frac{1}{6}
        ~.
    \end{align*}
\end{proof}

% We assume $w(i,j)=1$ for $(i,j)\in E$ and assume $w(i,j)=0$ for $(i,j)\in E_{\varnothing}$.
Next, we present the key Lemma~\ref{lem:edge-charging-for-k-partite-B} of this section, which gives the condition we need in Lemma~\ref{lem:approx-ratio} and Lemma~\ref{lem:bound-of-sum-B} under the specific choices of $\alpha,\beta$. 
As the proof of the lemma is quite technically involved by the proofs we have shown in the previous two sections, we defer it to Appendix~\ref{app:edge-charging-for-k-partite-B}.
\begin{lemma}
    \label{lem:edge-charging-for-k-partite-B}
    If $\alpha=\frac{3}{8},\beta=\frac{2}{3}$,
    there exists a charging scheme such that 
    \begin{enumerate}
        \item for any triangle $t\in \binom{[n]}{3}$, 
        \begin{align}
            \label{eqn:improved-M-bound-by-k-paritite-B}
            \sum_{(i,j)\in t} \expect[M_{i,j,t}|\mathcal{A}_t]\leq \sum_{(i,j)\in t\cap E} \expect[B_{i,j,t}|\mathcal{A}_t] \cdot c^*(i,j)~.
        \end{align}
        \item for any $b>0$, $(i,j)\in E$ and any $t\owns i,j$, in any recursive call, conditioning on $\mathcal{A}_t$ and that $(i,j)$ is a non-pivot edge charged by $B_{i,j,t}=b$, the probability that $i$ and $j$ are partitioned into different sets on line~\ref{line:partition-non-pivots} of Algorithm~\ref{alg:pivot} is at least $q_{i,j}\cdot b$, where 
            \begin{align*}
                q_{i,j} 
                =
                \begin{cases}
                   1/16  & \text{if $(i,j)\in E_{\lowerr}$}\\
                   1/8  & \text{if $(i,j)\in E_{\highdet}$}\\
                   3/32  & \text{if $(i,j)\in E_{\highrand}$}
                \end{cases}
            \end{align*}
    \end{enumerate}
\end{lemma}

A technical highlight in this section is that we round differently for edges in $E$ versus edges in $E_{\varnothing}$. 
This overcomes the challenge we face for $(d,d,r)$-same-triangles when analyzing the weighted cases with triangle inequality constraints. 
Specifically, in a recursive call with truncated LP solution $y$, when considering a $(d,d,r)$-same-triangle $t=(i,j,k)$, in which $(j,k)$ is the random edge, we encounter the issue if the triangle has a deterministic edge in $E_{\varnothing}$, say $(i,k)\in E_{\varnothing}$. 
% For simplicity, we abuse the notation to use $\ell^*$ for $\ell^*(y,i,j)$ here.
We can upper bound the probability that $(i,j)$ is modified by $c^*(i,j)+c^*(i,k)$ (shown by Eqn.~\eqref{eqn:ddr-same-modify-prob} and Lemma~\ref{lem:stronger-y-bounded-by-c*}).
% Because the sum of the LP cost of , we manage to avoid charging the edge $(j,k)$, whose charges can lead to the $\min\{L,\log n\}$ factor approximation in the weighted case with triangle inequality constraints. 
However, on $k$-partite graphs, the LP contribution of $(i,k)$ is $0$, and we need to charge the edge $(j,k)$ when $c^*(i,j) \ll c^*(i,k)$. 
As in the analysis of the weighted case, this only results in an approximation ratio of $O(\min\{L,\log n\})$.
To get a constant factor approximation, we modify the algorithm so as to upper bound the probability that $(i,j)$ is modified solely by $c^*(i,j)$.
We introduce different thresholds to determine whether edges in $E$ and $E_{\varnothing}$ are deterministic and round random edges in $E$ in a (slightly) different way.
As shown in Section~\ref{sec:algo}, the threshold for edges in $E$ is $1-\alpha$ while the threshold for edges in $E_{\varnothing}$ is $1-\alpha\beta$. 
Suppose $\ell^*$ is the dominant level of $(i,j)$.
According to the CCDF \eqref{eqn:distance-ccdf}, the distance of the random edge $(j,k)\in E$ is strictly greater than $d_{\ell^*}$ with probability $\frac{(y_{\ell^*-1}(j,k)-\alpha\beta)^+}{1-\alpha\beta}$. 
Using the LP constraint~\eqref{eqn:umvd-lp-triangle} and the fact that the deterministic edge $(i,k)$ has $y_{\ell^*-1}(i,k)<\alpha\beta$, we can upper bound the probability that $(i,j)$ is modified solely by $O(y_{\ell^*-1}(i,j))$. 
The detailed analysis can be found in Case 2 of Appendix~\ref{app:k-part-ddr-same-triangles}.
See Figure~\ref{fig:sol-for-ddr-same} for an illustration of why our algorithm resolves this issue.

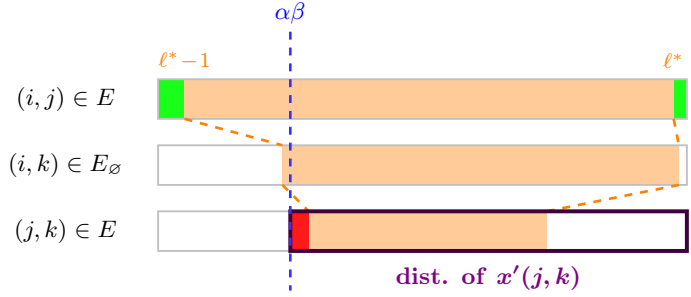
\begin{figure}[!t]
    \centering
    \begin{tikzpicture}       
        % \fill[fill=blue!40] ()
        \fill[fill=red!90] (50pt, 0pt) rectangle ++ (7pt, 15pt);
        \fill[fill=orange!40] (57pt, 0pt) rectangle ++ (90pt, 15pt);
        \draw[draw=lightgray, line width=0.7pt] (0 pt, 0 pt) rectangle ++ (200pt, 15pt);

        \fill[fill=orange!40] (47pt, 25pt) rectangle ++ (150pt, 15pt);
        \draw[draw=lightgray, line width=0.7pt] (0 pt, 25 pt) rectangle ++ (200pt, 15pt);

        \fill[fill=green!90] (0 pt, 50 pt) rectangle ++ (10pt, 15pt);
        \fill[fill=green!90] (195 pt, 50 pt) rectangle ++ (5pt, 15pt);
        \fill[fill=orange!40] (10 pt, 50 pt) rectangle ++ (185pt, 15pt);
        \draw[draw=lightgray, line width=0.7pt] (0 pt, 50 pt) rectangle ++ (200pt, 15pt);

        \node at (10pt, 72pt) {\color{orange} \scriptsize $\ell^*\!-\!1$};
        \node at (195pt, 72pt) {\color{orange} \scriptsize $\ell^*$};
        \draw[draw=orange, line width=1.0pt, dashed, -] (10pt, 50pt) -- (47pt, 40pt);
        \draw[draw=orange, line width=1.0pt, dashed, -] (47pt, 25pt) -- (57pt, 15pt);
        \draw[draw=orange, line width=1.0pt, dashed, -] (195pt, 50pt) -- (197pt, 40pt);
        \draw[draw=orange, line width=1.0pt, dashed, -] (197pt, 25pt) -- (147pt, 15pt);

        \draw[draw=violet!150, line width=1.5pt] (50 pt, 0 pt) rectangle ++ (150 pt, 15pt);
        \node at (125pt, -10pt) {\color{violet} \footnotesize \textbf{dist. of }$\boldsymbol{x'(j,k)}$};
    %    \foreach \x in {0, 36, 84}
    %    \filldraw[draw=black, fill=red, line width=0.5pt] (\x pt, 0pt) rectangle ++(12pt, 12pt);
    %    \foreach \x in {0, 36}
    %    \filldraw[draw=black, fill=blue, line width=0.5pt] (\x pt, -14pt) rectangle ++(12pt, 12pt);
    %    \foreach \x in {144}
    %    \filldraw[draw=black, fill=blue, line width=0.5pt] (\x pt, 0pt) rectangle ++(12pt, 12pt);
    %    \foreach \x in {12, 24, 96, 108, 120, 132}
    %    \filldraw[draw=black, fill=red!15!white, line width=0.5pt] (\x pt, 0pt) rectangle ++(12pt, 12pt);
    %    \foreach \x in {12, 24}
    %    \filldraw[draw=black, fill=blue!15!white, line width=0.5pt] (\x pt, -14pt) rectangle ++(12pt, 12pt);
    %    \foreach \x in {156, 168}
    %    \filldraw[draw=black, fill=blue!15!white, line width=0.5pt] (\x pt, 0pt) rectangle ++(12pt, 12pt);
    %    \node at (66pt, 6pt) {$\cdots$};
    %    \node at (66pt, -8pt) {$\cdots$};
    %    \node at (196pt, 6pt) {$\cdots$};
       
       \draw[draw=blue!80, line width=1.0pt, dashed, -] (50pt, 83pt) -- (50pt, -15pt);
       \node at (50pt, 90pt) {\color{blue} \footnotesize $\mathbf{\alpha\beta}$};
    %    \draw[draw=gray, line width=1.0pt, -] (138pt, 24pt) -- (138pt, 12pt);
    %    \draw[draw=gray, line width=1.0pt, -] (150pt, -12pt) -- (150pt, 0pt);
       
       \node at (-35pt, 7.5pt) {\footnotesize $(j,k)\in E$};
       \node at (-35pt, 32.5pt) {\footnotesize $(i,k)\in E_{\varnothing}$};
       \node at (-35pt, 57.5pt) {\footnotesize $(i,j)\in E$};
    \end{tikzpicture}
    \caption{An example of the rounding process in $(d,d,r)$-same-triangles with a deterministic edge in $E_{\varnothing}$. 
    The length of the red part indicates the probability that $(i,j)$ is modified in this triangle (up to a constant factor).
    The total length of the green parts indicates the LP cost of $(i,j)$.}
    \label{fig:sol-for-ddr-same}
\end{figure}

Finally, we establish the approximation ratio for $k$-partite graphs.
\begin{proof}[Proof of Theorem~\ref{thm:16-approx-k-part}]
Because of Lemma~\ref{lem:bound-of-sum-B} and the second bullet of Lemma~\ref{lem:edge-charging-for-k-partite-B}, we can present the second condition of Lemma~\ref{lem:approx-ratio} as follows:
\begin{align*}
\forall (i,j)\in E, \quad 
\sum_{t:i,j\in t} \expect[B_{i,j,t}]
\leq
\begin{cases}
    16 & \text{if $(i,j)\in E_{\lowerr}$}~,\\
    8  & \text{if $(i,j)\in E_{\highdet}$}~,\\
    32/3  & \text{if $(i,j)\in E_{\highrand}$}~.
\end{cases}
\end{align*}
Because the first bullet of Lemma~\ref{lem:edge-charging-for-k-partite-B} meets that of Lemma~\ref{lem:approx-ratio}, if $\alpha=\frac{8}{3},\beta=\frac{2}{3}$, Algorithm~\ref{alg:pivot} is at most $\max\{16, 8+\frac{8}{5}, \frac{32}{3}+\frac{8}{3}\}=16$-approximation for \UMVD on complete $k$-partite graphs.
\end{proof}

\section*{Acknowledgment}
Moses Charikar was supported by a Simons Investigator award. Ruiquan Gao was supported by a Stanford Graduate Fellowship.

% \section{Conclusions}
% \input{future}

\bibliographystyle{alpha}
\bibliography{bib}

\newcommand{\etalchar}[1]{$^{#1}$}
\begin{thebibliography}{JKMM20}

\bibitem[AALZ12]{DBLP:journals/siamcomp/AilonALZ12}
Nir Ailon, Noa Avigdor{-}Elgrabli, Edo Liberty, and Anke~{van} Zuylen.
\newblock Improved approximation algorithms for bipartite correlation
  clustering.
\newblock {\em {SIAM} J. Comput.}, 41(5):1110--1121, 2012.

\bibitem[ABF{\etalchar{+}}98]{agarwala1998approximability}
Richa Agarwala, Vineet Bafna, Martin Farach, Mike Paterson, and Mikkel Thorup.
\newblock On the approximability of numerical taxonomy (fitting distances by
  tree metrics).
\newblock {\em SIAM Journal on Computing}, 28(3):1073--1085, 1998.

\bibitem[AC11]{DBLP:journals/siamcomp/AilonC11}
Nir Ailon and Moses Charikar.
\newblock Fitting tree metrics: Hierarchical clustering and phylogeny.
\newblock {\em {SIAM} J. Comput.}, 40(5):1275--1291, 2011.

\bibitem[ACN08]{DBLP:journals/jacm/AilonCN08}
Nir Ailon, Moses Charikar, and Alantha Newman.
\newblock Aggregating inconsistent information: Ranking and clustering.
\newblock {\em J. {ACM}}, 55(5):23:1--23:27, 2008.

\bibitem[Ami04]{amit2004bicluster}
Noga Amit.
\newblock {\em The bicluster graph editing problem}.
\newblock PhD thesis, Tel Aviv University, 2004.

\bibitem[BBC04]{bansal2004correlation}
Nikhil Bansal, Avrim Blum, and Shuchi Chawla.
\newblock Correlation clustering.
\newblock {\em Machine learning}, 56:89--113, 2004.

\bibitem[BDST08]{brickell2008metric}
Justin Brickell, Inderjit~S Dhillon, Suvrit Sra, and Joel~A Tropp.
\newblock The metric nearness problem.
\newblock {\em SIAM Journal on Matrix Analysis and Applications},
  30(1):375--396, 2008.

\bibitem[CDK{\etalchar{+}}21]{DBLP:conf/focs/Cohen-Addad0KPT21}
Vincent Cohen{-}Addad, Debarati Das, Evangelos Kipouridis, Nikos Parotsidis,
  and Mikkel Thorup.
\newblock Fitting distances by tree metrics minimizing the total error within a
  constant factor.
\newblock In {\em 62nd {IEEE} Annual Symposium on Foundations of Computer
  Science}, pages 468--479, 2021.

\bibitem[CDL21]{cohen2021improving}
Vincent Cohen{-}Addad, R{\'e}mi De~Joannis {De Verclos}, and Guillaume Lagarde.
\newblock Improving ultrametrics embeddings through coresets.
\newblock In {\em International Conference on Machine Learning}, pages
  2060--2068. PMLR, 2021.

\bibitem[CE67]{cavalli1967phylogenetic}
Luigi~L Cavalli{-}Sforza and Anthony~WF Edwards.
\newblock Phylogenetic analysis. models and estimation procedures.
\newblock {\em American journal of human genetics}, 19(3 Pt 1):233, 1967.

\bibitem[CFLM22]{DBLP:conf/focs/Cohen-AddadFLM22}
Vincent Cohen{-}Addad, Chenglin Fan, Euiwoong Lee, and Arnaud~{de} Mesmay.
\newblock Fitting metrics and ultrametrics with minimum disagreements.
\newblock In {\em 63rd {IEEE} Annual Symposium on Foundations of Computer
  Science}, pages 301--311, 2022.

\bibitem[CGW05]{DBLP:journals/jcss/CharikarGW05}
Moses Charikar, Venkatesan Guruswami, and Anthony Wirth.
\newblock Clustering with qualitative information.
\newblock {\em J. Comput. Syst. Sci.}, 71(3):360--383, 2005.

\bibitem[CKK{\etalchar{+}}06]{DBLP:journals/cc/ChawlaKKRS06}
Shuchi Chawla, Robert Krauthgamer, Ravi Kumar, Yuval Rabani, and D.~Sivakumar.
\newblock On the hardness of approximating multicut and sparsest-cut.
\newblock {\em Comput. Complex.}, 15(2):94--114, 2006.

\bibitem[CKL20]{cohen2020efficient}
Vincent Cohen{-}Addad, CS~Karthik, and Guillaume Lagarde.
\newblock On efficient low distortion ultrametric embedding.
\newblock In {\em International Conference on Machine Learning}, pages
  2078--2088. PMLR, 2020.

\bibitem[CLLN23]{cohen2023handling}
Vincent Cohen{-}Addad, Euiwoong Lee, Shi Li, and Alantha Newman.
\newblock Handling correlated rounding error via preclustering: A
  1.73-approximation for correlation clustering.
\newblock In {\em 64th {IEEE} Annual Symposium on Foundations of Computer
  Science}, 2023.

\bibitem[CLN22]{DBLP:conf/focs/Cohen-AddadLN22}
Vincent Cohen{-}Addad, Euiwoong Lee, and Alantha Newman.
\newblock Correlation clustering with sherali-adams.
\newblock In {\em 63rd {IEEE} Annual Symposium on Foundations of Computer
  Science}, pages 651--661, 2022.

\bibitem[CMSY15]{DBLP:conf/stoc/ChawlaMSY15}
Shuchi Chawla, Konstantin Makarychev, Tselil Schramm, and Grigory Yaroslavtsev.
\newblock Near optimal {LP} rounding algorithm for correlationclustering on
  complete and complete k-partite graphs.
\newblock In {\em Proceedings of the Forty-Seventh Annual {ACM} Symposium on
  Theory of Computing}, pages 219--228, 2015.

\bibitem[Day87]{day1987computational}
William~HE Day.
\newblock Computational complexity of inferring phylogenies from dissimilarity
  matrices.
\newblock {\em Bulletin of mathematical biology}, 49(4):461--467, 1987.

\bibitem[DEFI06]{demaine2006correlation}
Erik~D Demaine, Dotan Emanuel, Amos Fiat, and Nicole Immorlica.
\newblock Correlation clustering in general weighted graphs.
\newblock {\em Theoretical Computer Science}, 361(2-3):172--187, 2006.

\bibitem[DPS{\etalchar{+}}13]{duggal2013resolving}
Geet Duggal, Rob Patro, Emre Sefer, Hao Wang, Darya Filippova, Samir Khuller,
  and Carl Kingsford.
\newblock Resolving spatial inconsistencies in chromosome conformation
  measurements.
\newblock {\em Algorithms for Molecular Biology}, 8:1--10, 2013.

\bibitem[DPS15]{di2015finding}
Marco {Di Summa}, David Pritchard, and Laura Sanit{\`a}.
\newblock Finding the closest ultrametric.
\newblock {\em Discrete Applied Mathematics}, 180:70--80, 2015.

\bibitem[Far72]{farris1972estimating}
James~S Farris.
\newblock Estimating phylogenetic trees from distance matrices.
\newblock {\em The American Naturalist}, 106(951):645--668, 1972.

\bibitem[FGR{\etalchar{+}}20]{fan2020generalized}
Chenglin Fan, Anna~C Gilbert, Benjamin Raichel, Rishi Sonthalia, and Gregory
  Van~Buskirk.
\newblock Generalized metric repair on graphs.
\newblock In {\em 17th Scandinavian Symposium and Workshops on Algorithm Theory
  (SWAT 2020)}. Schloss Dagstuhl-Leibniz-Zentrum f{\"u}r Informatik, 2020.

\bibitem[FKW93]{farach1993robust}
Martin Farach, Sampath Kannan, and Tandy Warnow.
\newblock A robust model for finding optimal evolutionary trees.
\newblock In {\em Proceedings of the Twenty-Fifth Annual ACM Symposium on
  Theory of Computing}, pages 137--145, 1993.

\bibitem[FRVB18]{fan2018metric}
Chenglin Fan, Benjamin Raichek, and Gregory Van~Buskirk.
\newblock Metric violation distance: Hardness and approximation.
\newblock In {\em Proceedings of the Twenty-Ninth Annual ACM-SIAM Symposium on
  Discrete Algorithms}, pages 196--209. SIAM, 2018.

\bibitem[GJ17]{gilbert2017if}
Anna~C Gilbert and Lalit Jain.
\newblock If it ain't broke, don't fix it: Sparse metric repair.
\newblock In {\em 2017 55th Annual Allerton Conference on Communication,
  Control, and Computing (Allerton)}, pages 612--619. IEEE, 2017.

\bibitem[GMT07]{gionis2007clustering}
Aristides Gionis, Heikki Mannila, and Panayiotis Tsaparas.
\newblock Clustering aggregation.
\newblock {\em Acm transactions on knowledge discovery from data (tkdd)},
  1(1):4--es, 2007.

\bibitem[Har67]{hartigan1967representation}
John~A Hartigan.
\newblock Representation of similarity matrices by trees.
\newblock {\em Journal of the american Statistical Association},
  62(320):1140--1158, 1967.

\bibitem[HKM05]{harb2005approximating}
Boulos Harb, Sampath Kannan, and Andrew McGregor.
\newblock Approximating the best-fit tree under l p norms.
\newblock In {\em Approximation, Randomization and Combinatorial Optimization.
  Algorithms and Techniques: 8th International Workshop on Approximation
  Algorithms for Combinatorial Optimization Problems}, pages 123--133.
  Springer, 2005.

\bibitem[JJS67]{jardine1967structure}
CJ~Jardine, Nicholas Jardine, and Robin Sibson.
\newblock The structure and construction of taxonomic hierarchies.
\newblock {\em Mathematical Biosciences}, 1(2):173--179, 1967.

\bibitem[JKMM20]{jafarov2020correlation}
Jafar Jafarov, Sanchit Kalhan, Konstantin Makarychev, and Yury Makarychev.
\newblock Correlation clustering with asymmetric classification errors.
\newblock In {\em International Conference on Machine Learning}, pages
  4641--4650. PMLR, 2020.

\bibitem[Joh67]{johnson1967hierarchical}
Stephen~C Johnson.
\newblock Hierarchical clustering schemes.
\newblock {\em Psychometrika}, 32(3):241--254, 1967.

\bibitem[JS71]{jardine1971mathematical}
Nicholas Jardine and Robin Sibson.
\newblock Mathematical taxonomy.
\newblock Technical report, 1971.

\bibitem[Kip23]{DBLP:conf/esa/Kipouridis23}
Evangelos Kipouridis.
\newblock Fitting tree metrics with minimum disagreements.
\newblock In Inge~Li G{\o}rtz, Martin Farach{-}Colton, Simon~J. Puglisi, and
  Grzegorz Herman, editors, {\em 31st Annual European Symposium on Algorithms,
  {ESA} 2023, September 4-6, 2023, Amsterdam, The Netherlands}, volume 274 of
  {\em LIPIcs}, pages 70:1--70:10. Schloss Dagstuhl - Leibniz-Zentrum f{\"{u}}r
  Informatik, 2023.

\bibitem[KM86]{kvrivanek1986np}
Mirko K{\v{r}}iv{\'a}nek and Jaroslav Mor{\'a}vek.
\newblock Np-hard problems in hierarchical-tree clustering.
\newblock {\em Acta informatica}, 23:311--323, 1986.

\bibitem[Kra44]{krasner1944nombres}
Marc Krasner.
\newblock Nombres semi-r{\'e}els et espaces ultram{\'e}triques.
\newblock {\em Comptes-Rendus de l’Acad{\'e}mie des Sciences}, 2:219, 1944.

\bibitem[Man99]{mantegna1999hierarchical}
Rosario~N Mantegna.
\newblock Hierarchical structure in financial markets.
\newblock {\em The European Physical Journal B-Condensed Matter and Complex
  Systems}, 11:193--197, 1999.

\bibitem[RTV86]{rammal1986ultrametricity}
Rammal Rammal, G{\'e}rard Toulouse, and Miguel~Angel Virasoro.
\newblock Ultrametricity for physicists.
\newblock {\em Reviews of Modern Physics}, 58(3):765, 1986.

\bibitem[SS62]{sneath1962numerical}
Peter~HA Sneath and Robert~R Sokal.
\newblock Numerical taxonomy.
\newblock {\em Nature}, 193:855--860, 1962.

\bibitem[SWW17]{sidiropoulos2017metric}
Anastasios Sidiropoulos, Dingkang Wang, and Yusu Wang.
\newblock Metric embeddings with outliers.
\newblock In {\em Proceedings of the Twenty-Eighth Annual ACM-SIAM Symposium on
  Discrete Algorithms}, pages 670--689. SIAM, 2017.

\bibitem[WSSB77]{waterman1977additive}
M.S. Waterman, T.F Smith, M.~Singh, and W.A. Beyer.
\newblock Additive evolutionary trees.
\newblock {\em Journal of Theoretical Biology}, 64(2):199--213, 1977.

\bibitem[ZW09]{DBLP:journals/mor/ZuylenW09}
Anke~{van} Zuylen and David~P. Williamson.
\newblock Deterministic pivoting algorithms for constrained ranking and
  clustering problems.
\newblock {\em Math. Oper. Res.}, 34(3):594--620, 2009.

\end{thebibliography}

% \newpage
\appendix

% \section{Further Related Works}
% \input{app-related}

\section{Examples Regarding Different Rounding Schemes}
\label{app:counter}
In this appendix section, we present examples demonstrating the issues with the purely randomized rounding scheme and how our rounding scheme resolves them.
Our first example is an unweighted \UMVD instance on complete $k$-partite graphs with $L=2$, where triple-based analysis encounters issues when analyzing the purely randomized rounding scheme to a constant factor approximation.

\begin{Example}
\label{exp:counter-for-random}
Consider an \UMVD instance on complete $k$-partite graphs with $L=2$.
Consider three distinct vertices $i,j,k\in [n]$.
Suppose $(i,j),(i,k)\in E, (j,k)\in E_{\varnothing}$, and the input satisfies $x_{\text{in}}(i,j)=d_{2}, x_{\text{in}}(i,k)=x_{\text{in}}(j,k)=d_{1}$.
In the language of \CClust, $(i,j)$ is a + edge while $(i,k)$ and $(j,k)$ are - edges.
Because we assume $y_{L}(u,v)=1$ for any distinct $u,v\in [n]$, we can use $a,b,c$ to respectively denote the only variables $y_{1}(i,j),y_{1}(i,k),y_{1}(j,k)$ and we call them the lengths of the edges.
Pick a huge $m=\omega(1)$ and a tiny $\eps=o(1/m)$.
Suppose the LP solution gives $a=m\eps$, $b=1-\eps$ and $c=1-(m+1)\epsilon$.

The LP contributions on these three edges are respectively $m\eps, \eps, 0$ because edges in $E_{\varnothing}$ do not contribute to the LP objective.
On the other hand, if we use the purely randomized rounding, when the pivot is $k$, with probability $b(1-c)=(1-o(1))\cdot (m+1)\eps$, $x'(i,k)=d_1,x'(j,k)=d_2$ and we need to modify the distance of $(i,j)$ to $d_1$ accordingly.
Hence, the expected cost of the algorithm in this triangle is at least $(1-o(1))\cdot (m+1)\eps$.
If we want the cost to be upper bounded by the charges on the edges, we need: either charge $(i,j)$ by a constant multiple of its LP contribution; or charge $(i,k)$ by a $\Omega(m)$ multiple of its LP contribution.
If we want to prove a constant factor approximation, we can only place our hope on charging $(i,j)$.
However, with probability $bc=1-o(1)$, $x'(i,k)=x'(j,k)=d_1$ and $(i,j)$ will appear in another recursive call.
This means that the total charge on $(i,j)$ can be a $\Omega(1)/(1-bc)=\omega(1)$ multiple of its LP contribution, which does not lead to a constant factor approximation analysis.
\end{Example}

Notice that in the first example, the three edges are deterministic if $\alpha$ is some positive constant and they have very low LP contributions. 
Our second example is an arbitrary weighted \UMVD instance with $L=2$, which is more general than the $k$-partite cases.
In this example, our rounding scheme will not modify any edge in such triangles.
This example is a formal version of the example we present in the Introduction.

\begin{Example}
\label{exp:no-violation}
Consider any weighted \UMVD instance with $L=2$ and the triangle inequality constraint.
Suppose we choose parameters $\alpha=\frac{1}{3},\beta=0$ in Algorithm~\ref{alg:pivot}. 
Consider three distinct vertices $i,j,k\in [n]$.
Because we have w.l.o.g. assumed $y_{L}(u,v)=1$ for any distinct $u,v\in [n]$, we can use $a,b,c$ to denote the only variables $y_{1}(i,j),y_{1}(i,k),y_{1}(j,k)$ and we call them the lengths of the edges.
Consider the scenario where $(i,j), (i,k)\text{ and }(j,k)$ are all deterministic, i.e., $a,b,c\in [0,\frac{1}{3})\cup (\frac{2}{3},1]$.
If the length of an edge is less than $\frac{1}{3}$, the dominant level is 2; otherwise, the dominant level is 1.
In addition, we assume that all three edges have LP contributions less than $\frac{1}{3}$, which implies their input equals their dominant distance level.
To show that our algorithm does not modify any edge in this triangle, it suffices to show that the ultrametric inequality is automatically satisfied if we set the distances to the dominant levels.

By setting the distances to the dominant distance levels, the only possible violation of the ultrametric inequality is when two of them have dominant level 2 and the other one has dominant level 1, w.l.o.g. say, $(j,k)$ is the only edge with dominant level 1.
If this violation happens, we will have $a,b<\frac{1}{3}$ and $c>\frac{2}{3}$, violating the LP constraint~\eqref{eqn:umvd-lp-triangle} by $a+b<c$.
\end{Example}

\section{Omitted Proofs}
\subsection{Proof of Lemma~\ref{lem:output-ultrametric}}
\label{app:output-ultrametric}

First, we can observe that the upper bound level $u$ automatically implies an upper bound for the input:
\begin{lemma}
    \label{lem:upper-bound-for-input}
    In any call $\algoname(V,x,u)$ of our algorithm, the input $x$ satisfies $x(i,j)\leq d_u$ for any $(i,j)\in \binom{V}{2}$.
\end{lemma}
\begin{proof}
    We prove the lemma by induction. 
    The base case is clear for the root call because $u=1$ and $d_1$ equals the largest entry in $x_{\text{in}}$.
    Suppose the lemma holds for $\algoname(V,x,u)$, whose random pivot equals $i$, and consider any of its children calls $\algoname(V_{\ell},x',\ell)$. 
    According to the definition of $V_{\ell}$ (line~\ref{line:partition-non-pivots}), for any $j,k\in V_{\ell}$, $x'(i,j)=x'(i,k)=d_{\ell}$. 
    Because the algorithm sets $x'(j,k)$ to $\min\{x(j,k),x'(i,j)\}$ in this case, we have $x'(j,k)\leq d_{\ell}$ and we prove the lemma for this child call.
\end{proof}

Next, we show a stronger version of the lemma by induction: 
each call $\algoname(V,x,u)$ returns an ultrametric $x'$ satisfying $x'(i,j)\leq d_{u}$ for any $(i,j)\in \binom{V}{2}$. 
The base cases are clear when $|V|\leq 2$ or $u=L$.

Consider a call with $|V|=n_0$ ($n_0\geq 3$) and $u<L$. 
Suppose we have shown the stronger version for any $|V|<n_0$.
Observe that the set of return values in $x'$ is a subset of those in $x$.
Because of Lemma~\ref{lem:upper-bound-for-input}, we have $x'(i,j)\leq d_u$ for any $(i,j)\in \binom{V}{2}$.
Suppose $i$ is the random pivot of the call.
We discuss two cases to show that the returned $x'$ satisfies the ultrametric inequality for any triangle in $V$.
\paragraph{The triangle involves $i$.} 
Fix any distinct $j,k\in V\setminus\{i\}$. 
If $x'(i,j)=x'(i,k)$, $j,k$ will appear in a child call, and thus the ultrametric inequality is satisfied because $x'(j,k)\leq x(i,j)$ according to the induction hypothesis.
If $x'(i,j)\neq x'(i,k)$, the algorithm sets $x'(j,k)$ to $\max\{x'(i,j),x'(i,k)\}$ and thus the ultrametric inequality is satisfied.
\paragraph{The triangle does not involve $i$.} 
Fix any distinct $j,k,r\in V\setminus\{i\}$. 
If $x'(i,j)=x'(i,k)=x'(i,r)$, the vertices appear in a child call and thus satisfy the ultrametric inequality according to the induction hypothesis.
Otherwise, either one of them is strictly greater than two others or two of them are equal and strictly greater than the other one. 
For the first case, w.l.o.g., suppose $x'(i,r)>x'(i,j),x'(i,k)$. 
The algorithm then sets $x'(j,r)$ and $x'(k,r)$ to $x'(i,r)$. 
Therefore, $x'(j,k)=\max\{x'(i,j),x'(i,k)\}$ or $j,k$ will appear in a child call, where we both have $x'(j,k)<x'(j,r)$ according to the induction hypothesis.
For the second case, w.l.o.g., suppose $x'(i,r)<x'(i,j)=x'(i,k)$. 
The algorithm then sets $x'(j,r)$ and $x'(k,r)$ to $x'(i,j)$. 
And, $j,k$ will appear in a child call and thus $x'(j,k)\leq x'(i,j)$ according to the induction hypothesis.

\subsection{Proof of Lemma~\ref{lem:highdet-error-bound}}
\label{app:highdet-error-bound}
Because $(i,j)$ is initially deterministic, if the dominant level $\tilde{\ell}(i,j)=\ell^*(y^*,i,j)$, $c^*(i,j)=1-\Delta y^*_{\ell^*(i,j)}(i,j)<\alpha$, violating $(i,j)$ is a high-cost edge. Therefore, $\tilde{\ell}(i,j)\neq \ell^*(y^*,i,j)$. Since $y^*_0(i,j)=0$ and $y^*_L(i,j)=1$, $\sum_{\ell=1}^L \Delta y^*_{\ell}(i,j)=1$. Hence, $c^*(i,j)=1-\Delta y^*_{\tilde{\ell}(i,j)}(i,j) \geq \Delta y^*_{\ell^*(i,j)}(i,j)>1-\alpha$.

\subsection{Proof of Lemma~\ref{lem:delta-y-inc}}
\label{app:delta-y-inc}
Suppose $u$ is the level in the recursive call and $u'$ is the level in the child call. 
%\Moses{Should we say ``index'' or ``level''?}
%\Ruiquan{I think ``level'' is better. I name $u$ as the input ``upper bound level'' of the recursive calls.}
According to the algorithm, we have $u'\geq u$. 
Fix any edge $(i,j)\in \binom{V'}{2}$. 

%\Moses{Can explain the intuition behind the proof, but maybe not worth it for this simple Lemma.}
If $\ell^*(y,i,j)>u'$, since the algorithm sets both $y_{\ell}(i,j)$ and $y'_{\ell}(i,j)$ to $y^*_{\ell}(i,j)$ for any $\ell\in\{u',u'+1,\cdots,L\}$, we have $\Delta y'_{\ell^*(y,i,j)}(i,j)=\Delta y^*_{\ell^*(y,i,j)}(i,j)=\Delta y_{\ell^*(y,i,j)}(i,j)$.
Therefore, according to the definition of the dominant level, $\Delta y'_{\ell^*(i,j)}(i,j)\geq \Delta y'_{\ell^*(y,i,j)}(i,j) = \Delta y_{\ell^*(i,j)}(i,j)$. 

Notice that $\Delta y_{\ell}(i,j)\leq y^*_{\ell}(i,j)$ for any $\ell\in [L]$ and $\Delta y'_{u'}(i,j)=y^*_{u'}(i,j)$. 
On the other hand, if $\ell^*(y,i,j)\leq u'$, because of Lemma~\ref{lem:feasible-y-in-calls} and the LP constraint~\eqref{eqn:umvd-lp-increasing}, we have $\Delta y'_{u'}(i,j)=y^*_{u'}(i,j)\geq y^*_{\ell^*(y,i,j)}(i,j) \geq \Delta y_{\ell^*(i,j)}(i,j)$ and thus $\Delta y'_{\ell^*(i,j)}(i,j)\geq \Delta y_{\ell^*(i,j)}(i,j)$.

In particular, in the root call, because $u=1$, $y^*$ equals the truncated LP solution. Therefore, for any $(i,j)\in \binom{[n]}{2}$, $\Delta y_{\ell^*(i,j)}(i,j)\geq \Delta y^*_{\ell^*(i,j)}(i,j)$.

\subsection{Proof of Lemma~\ref{lem:low-error-modify}}
\label{app:low-error-modify}
    We prove the lemma by a top-down induction on the recursive calls.
    % By top-down induction on the recursive calls (TBD). 
    The base case is the root call $\algoname([n], x_{\text{in}}, 1)$. 
    In the root call, because $u=1$, the truncated LP solution $y$ equals $y^*$.
    For any $(i,j)\in E_{\lowerr}$, according to the definition of low-cost edges, $\Delta y_{\tilde{\ell}(i,j)}(i,j) = 1-c^*(i,j)>1-\alpha$. 
    From Lemma~\ref{lem:feasible-y-in-calls} and $\alpha\leq \frac12$, $\tilde{\ell}(i,j)=\ell^*(y,i,j)$. 
    Since $x(i,j)=d_{\tilde{\ell}(i,j)}$ in the root call, we prove the lemma for the base case.

%\Moses{Will read and check rest of proof later.}
    Consider any call $\algoname(V, x, u)$ and any child call $\algoname(V_{\ell}, x', \ell)$. 
    Note that $\ell\geq u$ and the input distance $x'$ of the child call is consistent with the $x'$ in the parent call before partitioning non-pivot vertices (line~\ref{line:partition-non-pivots}). 
    Assume we have proved the lemma for the parent call. 
    Suppose the truncated LP solutions in the calls are respectively $y$ and $y'$ and the pivot vertex of the parent call is $i$.
    Consider any edge $(j,k)\in \binom{V_{\ell}}{2}$. 
    According to the definition of $V_{\ell}$, $x'(i,j)=x'(i,k)=d_{\ell}$. 
    Because the algorithm sets $x'(j,k)$ to $\min\{x(j,k),d_{\ell}\}$ in the parent call, we can then prove the lemma for the child call by discussing the following two cases.
    \begin{enumerate}
        \item If $x(j,k)\geq d_{\ell}$, $x'(j,k)=d_{\ell}$. 
        According to the induction hypothesis, $d_{\ell^*(y,j,k)}=x(j,k)\geq d_{\ell}$ and thus $\ell^*(y,i,j)\leq \ell$. 
        Because of the definition of the truncated LP solutions, the LP constraint~\eqref{eqn:umvd-lp-increasing} and Corollary~\ref{cor:low-err-are-det}, we have
        $
            \Delta y'_{\ell}(i,j) = y^*_{\ell}(i,j)\geq y^*_{\ell^*(y,i,j)}(i,j) \geq \Delta y_{\ell^*(i,j)}(i,j) > 1-\alpha.
        $
        Since $\sum_{\ell\in [L]} \Delta y'_{\ell}(j,k)=1$ (Lemma~\ref{lem:feasible-y-in-calls}) and $\alpha\leq \frac12$, $\ell^*(y',j,k)=\ell$ and thus $x'(j,k)=d_{\ell^*(y',j,k)}$.
        \item If $x(j,k)<d_{\ell}$, $x'(j,k)=x(j,k)$. 
        According to the induction hypothesis, $d_{\ell^*(y,j,k)}=x(j,k)$ and thus $x'(j,k)=d_{\ell^*(y,i,j)}$ and $\ell^*(y,j,k)>\ell\geq u$. 
        Because of the definition of the truncated LP solutions and Corollary~\ref{cor:low-err-are-det}, $\Delta y'_{\ell^*(y,j,k)}(j,k)=\Delta y^*_{\ell^*(y,j,k)}(j,k)=\Delta y_{\ell^*(j,k)}(j,k)>1-\alpha$. 
        Since $\sum_{\ell\in [L]} \Delta y'_{\ell}(j,k)=1$ (Lemma~\ref{lem:feasible-y-in-calls}) and $\alpha\leq \frac12$, $\ell^*(y',i,j)=\ell^*(y,i,j)$ and thus $x'(j,k)=d_{\ell^*(y',i,j)}$. 
    \end{enumerate}

    In particular, because the algorithm sets the distances of deterministic pivot edges to its dominant distance level, any edge $(i,j)\in E_{\lowerr}$ is not modified in any recursive call where it appears as a pivot edge.

\subsection{Missing proof of Lemma~\ref{lem:edge-charging-for-B}}
\label{app:edge-charging-for-B}
\subsubsection{$(d,d,r)$-diff-triangles} 
% \label{subsec:ddr-diff-triangles}
W.l.o.g. we assume $(i,j)$ and $(i,k)$ are deterministic, and $\ell^*(y,i,j)<\ell^*(y,i,k)$ (equivalently, $d_{\ell^*(y,i,j)}>d_{\ell^*(y,i,k)}$). 
Because of Lemma~\ref{lem:low-error-modify}, the input distances of $(i,j)$ and $(i,k)$ are respectively $d_{\ell^*(y,i,j)}$ and $d_{\ell^*(y,i,k)}$.
Since $d_{\ell^*(y,i,j)}>d_{\ell^*(y,i,k)}$, $M_{i,j,t}$ (or $M_{i,k,t}$) equals $1$ only when the pivot is $k$ (resp., $j$), it is low-cost and the random distance $x'(j,k)\neq d_{\ell^*(y,i,j)}$, i.e.,
\begin{align*}
    M_{i,j,t} &= \ind(\text{$k$ is the pivot vertex})\cdot \ind((i,j)\in E_{\lowerr}) \cdot \ind(x'(j,k)\neq d_{\ell^*(y,i,j)})~,
    \\
    M_{i,k,t} &= \ind(\text{$j$ is the pivot vertex})\cdot \ind((i,k)\in E_{\lowerr}) \cdot \ind(x'(j,k)\neq d_{\ell^*(y,i,j)})~.
\end{align*}
According to the CCDF~\eqref{eqn:complete-distance-ccdf} of the random distance, because $y$ is feasible in \eqref{eqn:umvd-lp}, the probability that $x'({j,k})\neq d_{\ell^*(y,i,j)}$ is 
\begin{align}
    1 - \Delta y_{\ell^*(i,j)}(j,k)
    &
    = 
    1-y_{\ell^*(i,j)}(j,k)+y_{\ell^*(i,j)-1}(j,k)
    \notag
    \\
    &
    \leq 
    1 - (y_{\ell^*(i,j)}(i,j) - y_{\ell^*(i,j)}(i,k)) + y_{\ell^*(i,j)-1}(i,j) + y_{\ell^*(i,j)-1}(i,k)
    \notag
    \\
    &
    \leq 
    (1 - \Delta y_{\ell^*(i,j)}(i,j)) + 2y_{\ell^*(i,k)-1}(i,k) ~.
    \label{eqn:ddr-diff-modify-prob}
    % \\
    % &
    % \leq 
    % (1 - \Delta y_{\ell^*(i,j)}(i,j))+c^*(j,k)+y_{\ell^*(i,k)-1}(i,k)~.
    % \tag{$c^*(j,k)\ge 1-\alpha > y_{\ell^*(i,k)-1}(i,k)$}
    % \min\Big\{1,\frac{1}{4(1-\alpha)}\Big\} \cdot (1 - \Delta y_{\ell^*(i,j)}(i,j))+\max\Big\{1,2-\frac{1}{4(1-\alpha)}\Big\} \cdot c^*(j,k)+y_{\ell^*(i,k)-1}(i,k)
    % \tag{$c^*(j,k)\ge 1-\alpha > y_{\ell^*(i,k)-1}(i,k), 1-\Delta y_{\ell^*(i,j)}(i,j)$}
\end{align}
Because of Corollary~\ref{cor:low-err-are-det}, $(j,k)\in E_{\highrand}$ and thus $c^*(i,j)\geq \alpha$. 
Note that the deterministic edge $(i,k)$ satisfies $y_{\ell^*(i,k)-1}(i,k)< \alpha$ (Eqn.~\eqref{eqn:y-of-a-det-edge}).
We can upper bound $1 - \Delta y_{\ell^*(i,j)}(j,k)$ by $1 - \Delta y_{\ell^*(i,j)}(i,j)+y_{\ell^*(i,k)-1}(i,k)+c^*(j,k)$.
Therefore, in this class of triangles,
\begin{align*}
    \sum_{(i',j')\in t} \expect[M_{i',j',t}] 
    \leq 
    \frac{2\big(1 - \Delta y_{\ell^*(i,j)}(i,j)+y_{\ell^*(i,k)-1}(i,k)+c^*(j,k)\big)}{3} 
    % \cdot 
    % \big(\ind((i,j)\in E_{\lowerr}) + \ind((i,k)\in E_{\lowerr})\big)
    ~.
\end{align*}
% \begin{align*}
%     \sum_{i'\neq j'\in t} \expect[M_{i',j',t}] 
%     \leq 
%     \frac{1 - \Delta y_{\ell^*(i,j)}(i,j)+y_{\ell^*(i,k)-1}(i,k)+c^*(j,k)}{3} 
%     \cdot 
%     \big(\ind((i,j)\in E_{\lowerr}) + \ind((i,k)\in E_{\lowerr})\big)~.
% \end{align*}

Accordingly, when the corresponding edge is non-pivot, we define $B_{i,j,t},B_{i,k,t},B_{j,k,t}$ as follows:
\begin{align*}
    B_{i,j,t}
    % &
    \defeq
    \begin{cases}
         2 & \text{if $(i,j) \in E_{\lowerr}$}\\
         \frac{2\alpha}{1-\alpha} & \text{if $(i,j)\in E_{\highdet}$}\\
         1 & \text{if $(i,j)\in E_{\highrand}$}
    \end{cases}
    % \begin{cases}
    %     \min\big\{2,\frac{1}{2(1-\alpha)}\big\} & \text{$(i,j)$ is $\lowerr$ or $\highdet$}\\
    %     \frac{1-\alpha}{\alpha} \cdot \min\big\{2,\frac{1}{2(1-\alpha)}\big\} & \text{$(i,j)$ is $\lowerr$ or $\highdet$}
    % \end{cases}
    , \quad\quad
    B_{i,k,t}
    % &
    \defeq
    \begin{cases}
        2 & \text{if $(i,k)\in E_{\lowerr}$}\\
        \frac{2\alpha}{1-\alpha} & \text{if $(i,k)\in E_{\highdet}$}\\
        0 & \text{if $(i,k)\in E_{\highrand}$}
    \end{cases}
    , \quad\quad
    B_{j,k,t} 
    % &
    \defeq 
    2
\end{align*}
Because of Lemma~\ref{lem:stronger-y-bounded-by-c*} and~\ref{lem:y-bounded-by-0}, this definition satisfies
\begin{align*}
    \expect[B_{i,j,t}] \cdot c^*(i,j)
    \geq 
    \frac{2}{3} &\cdot (1 - \Delta y_{\ell^*(i,j)}(i,j))
    ~,
    \\
    \expect[B_{i,k,t}] \cdot c^*(i,k)
    \geq 
    \frac{2}{3} \cdot y_{\ell^*(i,k)-1}(i,k)
    ~,
    &
    \quad 
    \expect[B_{j,k,t}] \cdot c^*(j,k)
    \geq 
    \frac{2}{3} \cdot c^*(j,k)
    ~,
\end{align*} 
and thus implies the first bullet of Lemma~\ref{lem:edge-charging-for-B}.
% \begin{align*}
%     \sum_{(i',j')\in t}
%     \expect[B_{i',j',t}] \cdot c^*(i',j')
%     &
%     % \geq 
%     % \frac{1+\ind((i,j)\in E_{\lowerr})}{3} \cdot y_{\ell^*(i,j)-1}(i,j) +
%     % \frac{1+\ind((i,k)\in E_{\lowerr})}{3} \cdot y_{\ell^*(i,k)-1}(i,k)
%     % \\
%     % &
%     \geq
%     \frac{\ind((i,j)\in E_{\lowerr})+\ind((i,k)\in E_{\lowerr})}{3} \cdot \big(1 - \Delta y_{\ell^*(i,j)}(i,j)+c^*(j,k)+y_{\ell^*(i,k)-1}(i,k)\big)
%     \\
%     &
%     \geq
%     \sum_{i'\neq j'\in t} \expect[M_{i',j',t}]~.
% \end{align*} 

Further, when the edge $(i,j)$ (or $(i,k)$) is non-pivot, one of the pivot edges in the triangle is random.
Because of Lemma~\ref{lem:prob-disappear}, the probability $i,j$ (or $i,k$) are then partitioned into different sets on line~\ref{line:partition-non-pivots} is at least $\alpha$, which is at least
\begin{align*}
    \begin{cases}
        \frac{\alpha}{2}\cdot B_{i,j,t} & \text{if $(i,j)\in E_{\lowerr}$}\\
        \frac{1-\alpha}{2}\cdot B_{i,j,t} & \text{if $(i,j)\in E_{\highdet}$}\\
        \alpha \cdot B_{i,j,t} & \text{if $(i,j)\in E_{\highrand}$}
    \end{cases}
    \quad 
    \text{and}
    \quad 
    \begin{cases}
        \frac{\alpha}{2}\cdot B_{i,k,t} & \text{if $(i,k)\in E_{\lowerr}$}\\
        \frac{1-\alpha}{2}\cdot B_{i,k,t} & \text{if $(i,k)\in E_{\highdet}$}\\
        1\cdot B_{i,k,t} & \text{if $(i,k)\in E_{\highrand}$}
    \end{cases}
    ~.
\end{align*}
When the edge $(j,k)$ is non-pivot, two pivot edges are both deterministic but have different dominant levels.
Because of Lemma~\ref{lem:prob-disappear}, the probability $i,j$ (or $i,k$) are then partitioned into different sets on line~\ref{line:partition-non-pivots} is $1$, which is at least $\frac{1}{2}\cdot B_{j,k,t}\geq \max\{\frac{1-\alpha}{2},\alpha\}\cdot B_{j,k,t}$ because $\alpha\leq \frac12$.
Hence, this definition satisfies the second bullet of Lemma~\ref{lem:edge-charging-for-B} for this case.

% Further, when edge $(i,j)$ (or $(i,k)$) is non-pivot, one of the pivot edges in the triangle is random. 
% Because of Lemma~\ref{lem:prob-disappear}, the probability $i,j$ (or $i,k$) are partitioned into different sets on line~\ref{line:partition-non-pivots} is at least $1-\alpha$.
% Note that when edge $(i,j)$ (or $(i,k)$) is non-pivot, $1-\alpha$ equals $\frac{1-\alpha}{2}\cdot B_{i,j,t}$ if $(i,j)\in E_{\lowerr}$, or equals $\alpha\cdot B_{i,j,t}$ if $(i,j)\in E_{\highdet}$ (resp., for $B_{i,k,t}$), or is no less than $(1-\alpha)\cdot B_{i,j,t}$ if $(i,j)\in E_{\highrand}$ (resp., for $B_{i,k,t}$), this definition for $B_{i,j,t}$ and $B_{i,k,t}$ satisfies the second bullet of Lemma~\ref{lem:edge-charging-for-B} for this class of triangles.
% When edge $(j,k)$ is non-pivot, both the pivot edges in the triange are deterministic but they have different dominant levels.
% Because of Lemma~\ref{lem:prob-disappear}, the probability $j,k$ are partitioned into different sets on line~\ref{line:partition-non-pivots} is $1$.
% Note $\alpha\geq \frac12$, when edge $(j,k)$ is non-pivot, $1=\frac12\cdot B_{j,k,t}\geq (1-\alpha)\cdot B_{j,k,t}$.
% Because of Corollary~\ref{cor:low-err-are-det}, $(j,k)\in E_{\highrand}$ and this definition of $B_{j,k,t}$ satisfies the second bullet of Lemma~\ref{lem:edge-charging-for-B} for this class of triangles.

\subsubsection{$(d,r,r)$-triangles} 
% \label{subsec:drr-triangles}
W.l.o.g., we assume $(j,k)$ is the deterministic edge.
Because of Corollary~\ref{cor:low-err-are-det}, $(i,j),(i,k)\in E_{\highrand}$ and thus $M_{i,j,t}=M_{i,k,t}=0$.
If $(j,k)\notin E_{\lowerr}$, $M_{j,k,t}=0$. 
Accordingly, by defining $B_{i,j,t},B_{i,k,t},B_{j,k,t}\defeq 0$, we show Lemma~\ref{lem:edge-charging-for-B} for the case $(j,k)\notin E_{\lowerr}$.

Next, we consider the case $(j,k)\in E_{\lowerr}$.
Only when the pivot in the recursive call is $i$, the low-cost edge $(j,k)$ can be modified. 
Because of Lemma~\ref{lem:low-error-modify}, the input distance of the call satisfies $x(j,k)=d_{\ell^*(y,j,k)}$.
If the random distances $x'(i,j),x'(i,k)$ satisfy $x'(i,j)=x'(i,k)\geq d_{\ell^*(y,j,k)}$ or $x'(i,j)<d_{\ell^*(y,j,k)}=x'(i,k)$ or $x'(i,j)=d_{\ell^*(y,j,k)}>x'(i,k)$, edge $(j,k)$ is not modified, i.e., $M_{j,k,t}=0$. 
Therefore, according to the CCDF~\eqref{eqn:complete-distance-ccdf} of the random distance, in this case,
\begin{align}
    \label{eqn:type-4-error}
    \begin{split}
    \sum_{(i',j')\in t} \expect[M_{i',j',t}]
    =
    \frac{1}{3}\cdot \bigg(1 - \sum_{\ell=1}^{\ell^*(y,j,k)} \Delta y_{\ell}(i,j) \cdot \Delta y_{\ell}(i,k) &- \Delta y_{\ell^*(j,k)}(i,k) (1-y_{\ell^*(j,k)}(i,j)) \\
    &- \Delta y_{\ell^*(j,k)}(i,j) (1-y_{\ell^*(j,k)}(i,k))\bigg)~.
    \end{split}
\end{align}
Based on the above equality, we can further prove the following upper bound for the expected number of modifications on low-cost edges in this triangle.
\begin{lemma}
If $\alpha\in [\frac{3-\sqrt{5}}{2}, \frac{1}{2}]$,  $\sum_{(i',j')\in t} \expect[M_{i',j',t}] \leq \frac{1}{3}\cdot (c^*(i,j)+c^*(i,k)+c^*(j,k))$.
\end{lemma} 
\begin{proof}
We shall prove the lemma by discussing the magnitude of $\ell^*(y,i,j)$, $\ell^*(y,i,k)$, and $\ell^*(y,j,k)$.
    
\paragraph{If $\ell^*(y,i,j),\ell^*(y,i,k)>\ell^*(y,j,k)$, }
\begin{align*}
    c^*(i,j)+c^*(i,k)
    &
    \geq 
    y_{\ell^*(i,j)-1}(i,j) + y_{\ell^*(i,k)-1}(i,k)
    \tag{Lemma~\ref{lem:stronger-y-bounded-by-c*}}
    \\
    &
    \geq 
    y_{\ell^*(j,k)}(i,j) + y_{\ell^*(j,k)}(i,k)
    \tag{LP constraint~\eqref{eqn:umvd-lp-increasing}}
    \\
    &
    \geq 
    y_{\ell^*(j,k)}(j,k)
    \tag{LP constraint~\eqref{eqn:umvd-lp-triangle}}
    \\
    &
    \geq 
    1 - c^*(j,k)~.
    \tag{Lemma~\ref{lem:stronger-y-bounded-by-c*}}
\end{align*}
Because the right-hand side of Eqn.~\eqref{eqn:type-4-error} can be trivially upper bounded by $\frac13$, $\sum_{(i',j')\in t} \expect[M_{i',j',t}]\leq \frac13\cdot (c^*(i,j)+c^*(i,k)+c^*(j,k))$.

\paragraph{If $\ell^*(y,i,j)<\ell^*(y,j,k)$ (or $\ell^*(y,i,k)<\ell^*(y,j,k)$ in symmetry),} we can upper bound Eqn.~\eqref{eqn:type-4-error} by $\frac{1}{3}\cdot (1 - \Delta y_{\ell^*(i,j)}(i,j) \cdot \Delta y_{\ell^*(i,j)}(i,k))$. Because $y$ satisfies the LP constraint~\eqref{eqn:umvd-lp-triangle} and~\eqref{eqn:umvd-lp-increasing}, 
\begin{align*}
    y_{\ell^*(i,j)}(i,j) - y_{\ell^*(i,j)}(i,k) &\leq y_{\ell^*(i,j)}(j,k) \leq y_{\ell^*(j,k)-1}(j,k)~,\\
    y_{\ell^*(i,j)-1}(i,j) - y_{\ell^*(i,j)-1}(i,k) &\geq -y_{\ell^*(i,j)-1}(j,k) \geq -y_{\ell^*(j,k)-1}(j,k)~,
\end{align*}
and thus $\Delta y_{\ell^*(i,j)}(i,j) -\Delta y_{\ell^*(i,j)}(i,k) \leq 2y_{\ell^*(j,k)-1}(j,k)$. 
Because of Lemma~\ref{lem:stronger-y-bounded-by-c*}, $\Delta y_{\ell^*(i,j)}(i,j) \geq 1-c^*(i,j)$ and $y_{\ell^*(j,k)-1}(j,k)\leq c^*(j,k)$. 
Therefore, $\Delta y_{\ell^*(i,j)}(i,k) \geq 1- c^*(i,j) - 2c^*(j,k)$.
We can upper bound Eqn.~\eqref{eqn:type-4-error} in this case by 
\begin{align}
    \frac{1}{3}\cdot\big(1 - \Delta y_{\ell^*(i,j)}(i,j) \cdot \Delta y_{\ell^*(i,j)}(i,k)\big)
    &
    \leq
    \frac13 \cdot \big(1 - (1- c^*(i,j))\cdot (1- c^*(i,j) - 2c^*(j,k))\big)
    \label{eqn:in-lemma-for-drr}
    \\
    & 
    \leq 
    \frac13 \cdot \big(2(1-c^*(i,j))\cdot c^*(j,k) + 1 - (1-c^*(i,j))^2\big)
    \notag
    \\
    & 
    = 
    \frac13\cdot \big((2-2c^*(i,j)) \cdot c^*(j,k) + (2-c^*(i,j)) \cdot c^*(i,j)\big)
    \notag
    \\
    &
    \leq
    \frac{2-\alpha}{3} \cdot \big(c^*(i,j)+c^*(j,k)\big)~.
    \tag{$c^*(i,j)\geq \alpha$}
\end{align}
On the other hand, because the random edge $(i,k)\in E_{\highrand}$, we can also upper bound Eqn.~\eqref{eqn:type-4-error} by $\frac{c^*(i,k)}{3\alpha}$. Averaging the previous two upper bounds of Eqn.~\eqref{eqn:type-4-error}, we can get
\begin{align*}
    \sum_{(i', j')\in t} \expect[M_{i',j',t}]
    &
    \leq
    \frac{1}{1+2\alpha-\alpha^2} \cdot \frac{2-\alpha}{3} \cdot \big(c^*(i,j)+c^*(j,k)\big) + \frac{2\alpha-\alpha^2}{1+2\alpha-\alpha^2} \cdot \frac{c^*(i,k)}{3\alpha}
    \\
    &
    \leq 
    \frac{2-\alpha}{3(1+2\alpha-\alpha^2)} \cdot \big(c^*(i,j) + c^*(i,k) + c^*(j,k) \big)
    \\
    &
    \leq 
    \frac{1}{3}\cdot\big(c^*(i,j) + c^*(i,k) + c^*(j,k)\big) ~.
    \tag{$\forall \alpha\in \big[\frac{3-\sqrt{5}}{2}, \frac{1}{2}\big], 2-\alpha\leq 1+2\alpha-\alpha^2$}
\end{align*}

\paragraph{If $\ell^*(y,i,j)=\ell^*(y,j,k)$ (or $\ell^*(y,i,k)=\ell^*(y,j,k)$ in symmetry),} we can upper bound Eqn.~\eqref{eqn:type-4-error} by $1 - \Delta y_{\ell^*(i,j)}(i,j) \cdot (1 - y_{\ell^*(i,j)-1}(i,k))$. Because of Lemma~\ref{lem:stronger-y-bounded-by-c*} and because the truncated LP solution $y$ satisfies LP constraint~\eqref{eqn:umvd-lp-triangle}, 
\begin{align*}
    y_{\ell^*(i,j)-1}(i,k) \leq y_{\ell^*(i,j)-1}(i,j)+y_{\ell^*(i,j)-1}(j,k)\leq c^*(i,j) + c^*(j,k)~.
\end{align*}
Hence, we can further upper bound Eqn.~\eqref{eqn:type-4-error} by
\begin{align*}
    \frac{1}{3}\cdot \big(1- (1-c^*(i,j))\cdot (1-c^*(i,j)-c^*(j,k)) \big)~,
%     &
%     \leq 
%     c^*(j,k) + 1 - (1-c^*(i,j))^2 
%     \\
%     &
%     \leq 
%     c^*(j,k) + (1+\alpha) c^*(i,j)~.
\end{align*}
which is stronger than~\eqref{eqn:in-lemma-for-drr}.
Following the steps in the previous case, we can then similarly show $\sum_{(i', j')\in t} \expect[M_{i',j',t}]\leq \frac{1}{3}\cdot(c^*(i,j) + c^*(i,k) + c^*(j,k))$. 
\end{proof}
% In total, we can upper bound the number of modifications of low-error edges in $\mathcal{A}_t$ as follows:
% \begin{align*}
%     \sum_{(u,v)\in t} \expect[M_{u,v,t}|\mathcal{A}_t] \leq \frac{1}{3}\big(c^*(i,j)+c^*(i,k)+c^*(j,k)\big)
% \end{align*}

Accordingly, when the corresponding edge is non-pivot, we define $B_{i,j,t},B_{i,k,t},B_{j,k,t}$ as follows:
\begin{align*}
    B_{i,j,t},\, B_{i,k,t},\, B_{j,k,t} \defeq 1~,
\end{align*}
% we always define $B_{j,k,t}$ to be $0$, and define $B_{i,j,t}$ (or $B_{i,k,t}$) to be $\frac{1}{2(1-\alpha)}$ when it is backward. 
% Since edges $(i,j)$ and $(i,k)$ are random and thus high-error, 
% \begin{align*}
%     \sum_{(u,v)\in t} c^*(u,v) \cdot \expect[B_{u,v,t}|\mathcal{A}_t] = \frac{1}{6(1-\alpha)}(c^*(i,j)+c^*(i,k)) \geq \frac{1}{6(1-\alpha)}\cdot 2(1-\alpha)=\frac{1}{3}~. 
% \end{align*}
% Therefore, we derive Lemma~\ref{lem:improved-M-bound-by-B} for this type. 
which implies the first bullet of Lemma~\ref{lem:edge-charging-for-B}:
\begin{align*} 
    \sum_{(i', j')\in t} \expect[B_{i',j',t}]\cdot c^*(i',j') = \frac{1}{3}\cdot \big(c^*(i,j)+c^*(i,k)+c^*(j,k)\big) \geq \sum_{(i', j')\in t} \expect[M_{i',j',t}]~.
\end{align*}

Note that because of Lemma~\ref{lem:prob-disappear}, when being non-pivot, endpoints of all three edges in the triangle will be partitioned into different sets on line~\ref{line:partition-non-pivots} with probability at least $\alpha$, which is $\alpha \cdot B_{i,j,t}\geq \frac{1-\alpha}{2}\cdot B_{i,j,t}$ (resp., for $B_{i,k,t}$ and $B_{j,k,t}$) when $\alpha\geq \frac{3-\sqrt{5}}{2}$. 
Therefore, this definition of $B_{i',j',t}$s satisfies the second bullet of Lemma~\ref{lem:edge-charging-for-B} for this class of triangles.

% \subsubsection{$(r,r,r)$-triangles.} 
% % \label{subsec:rrr-triangles}
% Because of Corollary~\ref{cor:low-err-are-det}, there is no low-cost edges in this class of triangle, and thus $M_{i,j,t}=M_{i,k,t}=M_{j,k,t}=0$.
% Accordingly, by defining $B_{i,j,t},B_{i,k,t},B_{j,k,t}\defeq 0$, we can prove Lemma~\ref{lem:edge-charging-for-B} for this class of triangles. 

\subsection{Missing proof of Lemma~\ref{lem:s-weighted-edge-charging-for-B}}
\label{app:s-weighted-edge-charging-for-B}

\subsubsection{$(d,d,r)$-diff-triangles} 
W.l.o.g. we assume $(i,j)$ and $(i,k)$ are deterministic, and $(i,j)$ has a lower dominant level, i.e., $\ell^*(y,i,j)<\ell^*(y,i,k)$ (equivalently, $d_{\ell^*(y,i,j)}>d_{\ell^*(y,i,k)}$). 
Because of Corollary~\ref{cor:low-err-are-det}, we have $(j,k)\in E_{\lowerr}$, $M_ {j,k,t}=0$ and $c^*(j,k)\geq \alpha$. 
Recall that in Section~\ref{sec:ratio-complete}, we have shown for this class of triangles: 
\begin{align*}
    M_{i,j,t} &= \ind(\text{$k$ is the pivot vertex})\cdot \ind((i,j)\in E_{\lowerr}) \cdot \ind(x'(j,k)\neq d_{\ell^*(y,i,j)})~,
    \\
    M_{i,k,t} &= \ind(\text{$j$ is the pivot vertex})\cdot \ind((i,k)\in E_{\lowerr}) \cdot \ind(x'(j,k)\neq d_{\ell^*(y,i,j)})~,
\end{align*}
and the probability that $x'(j,k)\neq d_{\ell^*(y,i,j)}$ is at most $(1 - \Delta y_{\ell^*(i,j)}(i,j)) + 2y_{\ell^*(i,k)-1}(i,k)$. 
Because of Lemma~\ref{lem:stronger-y-bounded-by-c*}, $1 - \Delta y_{\ell^*(i,j)}(i,j)\leq c^*(i,j)$, $y_{\ell^*(i,k)-1}(i,k)\leq c^*(i,k)$.
Because edges $(i,j)$ and $(i,k)$ are deterministic, $1 - \Delta y_{\ell^*(i,j)}(i,j), y_{\ell^*(i,k)-1}(i,k) < \alpha \leq c^*(j,k)$ (Eqn.~\eqref{eqn:y-of-a-det-edge}).
Therefore, in this class of triangles, we have 
\begin{align*}
    \sum_{(i', j')\in t} w(i',j')\cdot \expect[M_{i',j',t}] 
    &
    \leq 
    \frac{(1 - \Delta y_{\ell^*(i,j)}(i,j)) + 2y_{\ell^*(i,k)-1}(i,k)}{3} \cdot \big(w(i,j)+w(i,k)\big)
    \\
    &
    \leq 
    \frac{w(i,j)\cdot c^*(i,j)+2w(i,k)\cdot c^*(i,k)}{3} + \frac{2y_{\ell^*(i,k)-1}(i,k)\cdot (w(i,k)+w(j,k))}{3}
    \\
    &
    \quad 
    + \frac{(1 - \Delta y_{\ell^*(i,j)}(i,j))\cdot (w(i,j)+w(j,k))}{3} 
    \tag{triangle inequality for $w$}
    \\
    &
    \leq 
    \frac{w(i,j)\cdot c^*(i,j) \!+\! 2w(i,k)\cdot c^*(i,k)}{3} 
    + 
    \frac{2\big(c^*(i,k)\cdot w(i,k) \!+\! c^*(j,k)\cdot w(j,k)\big)}{3}
    \\
    &
    \quad 
    + \frac{c^*(i,j)\cdot w(i,j)+c^*(j,k)\cdot w(j,k)}{3} 
    \\
    &
    = 
    \frac{2w(i,j)\cdot c^*(i,j)+3w(j,k)\cdot c^*(j,k) +4w(i,k)\cdot c^*(i,k)}{3} 
    ~.
\end{align*}
Accordingly, when the corresponding edge is non-pivot, we can define $B_{i,j,t},B_{i,k,t},B_{j,k,t}\defeq 4$ to prove the first bullet of Lemma~\ref{lem:s-weighted-edge-charging-for-B}.
Further, we will classify all charges into the first type, i.e., $\forall (u, v)\in t, B_{u,v,t}^{(1)}=B_{u,v,t}$.
Note that in this class of triangles, either there is a random pivot edge or the two deterministic pivot edges have different dominant levels. 
Because of Lemma~\ref{lem:prob-disappear}, when being non-pivot, the endpoints of all three edges in $t$ will be partitioned into different sets on line~\ref{line:partition-non-pivots} with probability at least $\alpha=\frac{1}{3}$, which is $\Omega(1)\cdot B_{u,v,t}^{(1)}$ for any $(u, v)\in t$.
Therefore, this definition satisfies the condition of Lemma~\ref{lem:bound-of-sum-B} with $q=\Omega(1)$.

\subsubsection{$(d,r,r)$-triangles} 
W.l.o.g., we assume $(j,k)$ is the deterministic edge.
Because of Corollary~\ref{cor:low-err-are-det}, $(i,j),(i,k)\in E_{\highrand}$, and we have $M_{i,j,t}=M_{i,k,t}=0$ and $c^*(i,j),c^*(i,k)\geq \alpha=\frac{1}{3}$.
Therefore, in this class of triangles,
\begin{align*}
    \sum_{(i', j')\in t} w(i',j')\cdot \expect[M_{i',j',t}] 
    \leq 
    \frac{w(j,k)}{3} 
    \leq 
    \frac{w(i,j)+w(i,k)}{3} 
    \leq 
    w(i,j)\cdot c^*(i,j)+w(i,k)\cdot c^*(i,k)
    ~.
\end{align*}
Accordingly, when the corresponding edge is non-pivot, we can define $B_{i,j,t}\defeq 0$ and $B_{i,k,t},B_{j,k,t}\defeq 3$ to prove the first bullet of Lemma~\ref{lem:s-weighted-edge-charging-for-B}. 
Further, we will classify all non-trivial charges into the first type, i.e., $B_{i,k,t}^{(1)}=B_{i,k,t}, B_{j,k,t}^{(1)}=B_{j,k,t}$.
Because of Lemma~\ref{lem:prob-disappear}, when being non-pivot, $i,k$ (or $j,k$) will be partitioned into different sets on line~\ref{line:partition-non-pivots} with probability at least $\alpha=\frac{1}{3}$, which is $\Omega(1)\cdot B_{i,k,t}^{(1)}$ (resp., for $B_{j,k,t}^{(1)}$).
Therefore, this definition satisfies the condition of Lemma~\ref{lem:bound-of-sum-B} with $q=\Omega(1)$.
% \subsection{Proof of Lemma~\ref{lem:bad-pivots}}

\subsection{Proof of Lemma~\ref{lem:edge-charging-for-k-partite-B}}
\label{app:edge-charging-for-k-partite-B}
For convenience, we assume $\alpha=\frac{3}{8},\beta=\frac{2}{3}$ in the proof.
With this assumption, the CCDF functions~\eqref{eqn:distance-ccdf} and~\eqref{eqn:empty-distance-ccdf} respectively become:
\begin{align}
    \label{eqn:k-part-distance-ccdf}
    \text{for $(i,j)\in E$:} 
    &
    \qquad
    \forall \,\ell\in [L],\quad \Pr[x'(i,j)\geq d_{\ell}] = \frac{(4\cdot y_\ell(i,j)-1)^+}{3}~.
    \\
    \label{eqn:k-part-empty-distance-ccdf}
    \text{for $(i,j)\in E_{\varnothing}$:} 
    &
    \qquad
    \forall \,\ell\in [L],\quad \Pr[x'(i,j)\geq d_{\ell}] = y_\ell(i,j)~.
\end{align}
Given the truncated LP solution $y$, we define an edge $(i,j)$ to be deterministic if 
\begin{itemize}
    \item $(i,j)\in E$ and $\Delta y_{\ell^*(i,j)}(i,j)>1-\alpha=\frac{5}{8}$, or
    \item $(i,j)\in E_{\varnothing}$ and $\Delta y_{\ell^*(i,j)}(i,j)>1-\alpha\beta=\frac{3}{4}$.
\end{itemize}
Otherwise, we define it to be random.
In addition, we restate Lemma~\ref{lem:k-prob-disappear0} with the specific choice of $\alpha,\beta$ in the following lemma.
\begin{lemma}[Restatement of Lemma~\ref{lem:k-prob-disappear0}]
    \label{lem:k-prob-disappear}
    Consider any recursive call $\algoname(V, x, u)$ with $|V|>2$.
    Suppose $i$ is the pivot vertex of the call. 
    For any $(j,k)\in \binom{V\setminus\{i\}}{2}$, the probability that $j,k$ are partitioned into different sets on line~\ref{line:partition-non-pivots} can be lower bounded by
    \begin{itemize}
        \item $1$ if both $(i,j)$ and $(i,k)$ are deterministic in the call but they have different dominant levels, 
        \item $1/4$ if one of $(i,j)$ and $(i,k)$ is random in $E_{\varnothing}$, or
        \item $1/6$ if one of $(i,j)$ and $(i,k)$ is random in $E$.
    \end{itemize}
\end{lemma}
\begin{corollary}
    \label{cor:k-prob-disappear}
    Consider any recursive call $\algoname(V, x, u)$ with $|V|>2$.
    Suppose $i$ is the pivot vertex of the call. 
    For any $(j,k)\in \binom{V\setminus\{i\}}{2}$, the probability that $j,k$ are partitioned into different sets on line~\ref{line:partition-non-pivots} can be lower bounded by $1/6$ if at least one of $(i,j)$ and $(i,k)$ is random.
\end{corollary}

Note that each triangle cannot have exactly one edge in $E$ because $E$ is a $k$-partite graph.
When a triangle does not have any edge in $E$, no edge in the triangle is low-cost, and thus $M_{i,j,t}=M_{i,k,t}=M_{j,k,t}=0$.
By trivially defining $B_{i,j,t},B_{i,k,t},B_{j,k,t}\defeq 0$, we can prove Lemma~\ref{lem:edge-charging-for-k-partite-B}. 
Therefore, it suffices to only consider the triangles with at least two edges in $E$ in the rest of the analysis. 

We prove the lemma by presenting the charging scheme for each class of triangles.
Consider any triangle $t=(i,j,k)$ and any recursive call with a vertex set involving all of $i,j,k$.
Suppose the recursive call has a truncated LP solution $y$.
For convenience, we will condition on $\mathcal{A}_t$ in the rest of this section, and all expectations will be automatically conditioned on $\mathcal{A}_t$.

\subsubsection{$(d,d,d)$-triangles}
Note that the choice of $\beta$ does not influence how Algorithm~\ref{alg:pivot} rounds deterministic pivot edges in $E$.
In Section~\ref{sec:ratio-complete}, when all three edges are all in $E$, we have shown Lemma~\ref{lem:edge-charging-for-B} for this class of triangles, which gives a stronger condition of Lemma~\ref{lem:bound-of-sum-B} %for
% \begin{align*}
%     q_{i,j} 
%     =
%     \begin{cases}
%        3/16  & \text{if $(i,j)\in E_{\lowerr}$}\\
%        5/16  & \text{if $(i,j)\in E_{\highdet}$}\\
%        3/8  & \text{if $(i,j)\in E_{\highrand}$}
%     \end{cases}
% \end{align*}
than that of Lemma~\ref{lem:edge-charging-for-k-partite-B} under $\alpha=\frac{3}{8}$.
Next, we only need to analyze the case when two edges are in $E$ for this class of triangles.

W.l.o.g., we assume $(i,j),(i,k)\in E$ and $(j,k)\in E_{\varnothing}$.
Note that the distance of any pivot edge $(i',j')\in t$ is set to its dominant distance level $d_{\ell^*(y,i',j')}$.
Because of Lemma~\ref{lem:low-error-modify}, any low-cost edge $(i',j')$ has an input distance $x(i',j')=d_{\ell^*(y,i',j')}$. 
If the non-pivot edge is low-cost, it is modified only when the ultrametric inequality is violated on the three dominant distance levels, i.e., there exists a permutation $(i',j',k')$ of $i,j,k$ such that $\ell^*(y,i',j')<\min\{\ell^*(y,i',k'),\ell^*(y,j',k')\}$. 
Because $\ell^*(y,i',j')<\min\{\ell^*(y,i',k'),\ell^*(y,j',k')\}$ implies
\begin{align}
    1-\Delta y_{\ell^*(i',j')}(i',j')
    &
    \geq 
    1 - y_{\ell^*(i',j')}(i',j')
    \notag
    \\
    &
    \geq 
    1 - y_{\ell^*(i',j')}(i',k') - y_{\ell^*(i',j')}(j',k') 
    \tag{LP constraint~\eqref{eqn:umvd-lp-triangle}}
    \\
    &
    \geq 
    1 - y_{\ell^*(i',k')-1}(i',k') - y_{\ell^*(j',k')-1}(j',k')
    \tag{LP constraint~\eqref{eqn:umvd-lp-increasing}}
    \\
    &
    \geq 
    1 - (1-\Delta y_{\ell^*(i',k')-1}(i',k')) - (1-\Delta y_{\ell^*(j',k')-1}(j',k'))
    \label{eqn:det-no-violation}
    ~.
\end{align}
However, because all three edges are deterministic in the recursive call, $1-\Delta y_{\ell^*(i,j)}(i,j), 1-\Delta y_{\ell^*(i,k)}(i,k)<\frac{3}{8}$ and $1-\Delta y_{\ell^*(j,k)}(j,k)<\frac{1}{4}$.
Because $\frac{3}{8}\cdot 2+\frac{1}{4}=1$, Inequality~\eqref{eqn:det-no-violation} cannot hold.
Therefore, the low-cost non-pivot edges are not modified in this case, and we have $M_{i,j,t}=M_{i,k,t}=M_{j,k,t}=0$.
Accordingly, by always defining $B_{i,j,t},B_{i,k,t},B_{j,k,t}\defeq 0$, we show Lemma~\ref{lem:edge-charging-for-k-partite-B} for this case.

\subsubsection{$(d,d,r)$-same-triangles} 
\label{app:k-part-ddr-same-triangles}
W.l.o.g., we assume $(i,j)$ and $(i,k)$ are deterministic.
This implies that $M_{j,k,t}=0$.
Because of Lemma~\ref{lem:low-error-modify} and $\ell^*(y,i,j)=\ell^*(y,i,k)$, the input distances of $(i,j)$ and $(i,k)$ equal $d_{\ell^*(y,i,j)}$.
Therefore, $M_{i,j,t}$ (or $M_{i,k,t}$) equals $1$ only when the pivot is $k$ (resp., $j$), it is low-cost and the random distance $x'(j,k)>d_{\ell^*(y,i,j)}$, i.e.,
\begin{align*}
    M_{i,j,t} &= \ind(\text{$k$ is the pivot vertex})\cdot \ind((i,j)\in E_{\lowerr}) \cdot \ind(x'(j,k)>d_{\ell^*(y,i,j)})~,
    \\
    M_{i,k,t} &= \ind(\text{$j$ is the pivot vertex})\cdot \ind((i,k)\in E_{\lowerr}) \cdot \ind(x'(j,k)>d_{\ell^*(y,i,j)})~.
\end{align*}
Note that triangles in this class with at least two edges in $E$ can be divided into the following three cases: (1) the random edge is in $E_{\varnothing}$, (2) a deterministic edge is in $E_{\varnothing}$, and (3) all three edges are in $E$.
Next, we shall discuss these three cases to prove Lemma~\ref{lem:edge-charging-for-k-partite-B} for this class of triangles. 

\paragraph{Case 1: the random edge is in $E_{\varnothing}$.}
According to the CCDF~\eqref{eqn:k-part-empty-distance-ccdf} of the random distance, the probability that $x'({j,k})>d_{\ell^*(y,i,j)}$ is 
\begin{align*}
    y_{\ell^*(i,j)-1}(j,k)
    &
    \leq
    y_{\ell^*(i,j)-1}(i,j) + y_{\ell^*(i,j)-1}(i,k)
    \tag{LP constraint~\eqref{eqn:umvd-lp-triangle}}
    \\
    &
    =
    y_{\ell^*(i,j)-1}(i,j) + y_{\ell^*(i,k)-1}(i,k)
    ~.
    % \\
    % &
    % \leq 
    % c(i,j) + c(i,k)
    % \leq 
    % c^*(i,j) + c^*(i,k)
\end{align*}
Therefore, in this case, 
\begin{align*}
    \sum_{(i', j')\in t} \expect[M_{i',j',t}] \leq \frac{y_{\ell^*(i,j)-1}(i,j) + y_{\ell^*(i,k)-1}(i,k)}{3} \cdot \Big(\ind((i,j)\in E_{\lowerr}) + \ind((i,k)\in E_{\lowerr})\Big)~.
\end{align*}
Accordingly, when the corresponding edge is non-pivot, we define $B_{i,j,t},B_{i,k,t},B_{j,k,t}$ as follows:
\begin{align*}
    B_{i,j,t}
    \defeq
    \begin{cases}
     2 & \text{if } (i,j)\in E_{\lowerr}\\
     1 & \text{if }(i,j)\in E_{\highdet}\\
     0 & \text{if } (i,j)\in E_{\highrand}
    \end{cases}
    ~, 
    \quad
    B_{i,k,t} 
    \defeq
    \begin{cases}
     2 & \text{if } (i,k)\in E_{\lowerr}\\
     1 & \text{if }(i,k)\in E_{\highdet}\\
     0 & \text{if } (i,k)\in E_{\highrand}
    \end{cases}
    ~,
    \quad
    B_{j,k,t}
    \defeq
    0
\end{align*}
Because of Lemma~\ref{lem:stronger-y-bounded-by-c*} and~\ref{lem:y-bounded-by-0}, this definition implies
% \begin{align*}
%     \expect[B_{i,j,t}] \cdot c^*(i,j)
%     &
%     \geq 
%     \frac{1+\ind(\text{$(i,j)$ is $\lowerr$})}{3} \cdot y_{\ell^*(i,j)-1}(i,j)
%     ~,
%     \\
%     \expect[B_{i,k,t}] \cdot c^*(i,k)
%     &
%     \geq 
%     \frac{1+\ind(\text{$(i,j)$ is $\lowerr$})}{3} \cdot y_{\ell^*(i,k)-1}(i,k)
%     ~,
% \end{align*}
% and thus implies 
the first bullet of Lemma~\ref{lem:edge-charging-for-k-partite-B}:
 \begin{align*}
    \sum_{(i', j')\in t}
    \expect[B_{i',j',t}] \cdot c^*(i',j')
    &
    \geq 
    \frac{1+\ind(\text{$(i,j)\in E_{\lowerr}$})}{3} \cdot y_{\ell^*(i,j)-1}(i,j) +
    \frac{1+\ind(\text{$(i,k)\in E_{\lowerr}$})}{3} \cdot y_{\ell^*(i,k)-1}(i,k)
    \\
    &
    \geq
    \frac{\ind((i,j)\in E_{\lowerr})+\ind((i,k)\in E_{\lowerr})}{3} \cdot \big(y_{\ell^*(i,j)-1}(i,j) + y_{\ell^*(i,k)-1}(i,k)\big)
    \\
    &
    \geq
    \sum_{(i', j')\in t} \expect[M_{i',j',t}]~.
\end{align*} 
Further, when the edge $(i,j)$ (or $(i,k)$) is non-pivot, one of the pivot edges in the triangle is random. 
Because of Corollary~\ref{cor:k-prob-disappear}, the probability that $i,j$ (or $i,k$) are then partitioned into different sets on line~\ref{line:partition-non-pivots} is at least $\frac{1}{6}$, which is at least
\begin{align*}
    \begin{cases}
        \frac{1}{12}\cdot B_{i,j,t} & \text{if $(i,j)\in E_{\lowerr}$}\\
        \frac{1}{6}\cdot B_{i,j,t} & \text{if $(i,j)\in E_{\highdet}$}\\
        1\cdot B_{i,j,t} & \text{if $(i,j)\in E_{\highrand}$}
    \end{cases}
    \quad 
    \text{and}
    \quad
    \begin{cases}
        \frac{1}{12}\cdot B_{i,k,t} & \text{if $(i,k)\in E_{\lowerr}$}\\
        \frac{1}{6}\cdot B_{i,k,t} & \text{if $(i,k)\in E_{\highdet}$}\\
        1\cdot B_{i,k,t} & \text{if $(i,k)\in E_{\highrand}$}
    \end{cases}
\end{align*}
Hence, this definition satisfies the second bullet of Lemma~\ref{lem:edge-charging-for-k-partite-B} for this case.

\paragraph{Case 2: one deterministic edge is in $E_{\varnothing}$.}
W.l.o.g., we assume that the deterministic edge $(i,k)\in E_{\varnothing}$,
which implies $M_{j,k,t}=0$ and $y_{\ell^*(i,j)-1}(i,k)=y_{\ell^*(i,k)-1}(i,k)\leq 1-\Delta y_{\ell^*(i,k)-1}(i,k)<\frac{1}{4}$.
According to the CCDF~\eqref{eqn:k-part-distance-ccdf} of the random distance, the probability that $x'({j,k})>d_{\ell^*(y,i,j)}$ is 
\begin{align*}
    \frac{\big(4\cdot y_{\ell^*(i,j)-1}(j,k)-1\big)^+}{3}
    &
    \leq 
    \frac{\big(4\cdot(y_{\ell^*(i,j)-1}(i,j) + y_{\ell^*(i,j)-1}(i,k))-1\big)^+}{3}
    \tag{LP constraint~\eqref{eqn:umvd-lp-triangle}}
    \\
    &
    \leq 
    \frac{4\cdot y_{\ell^*(i,j)-1}(i,j)}{3}
    ~.
    % \\
    % &
    % \leq 
    % c(i,j) + c(i,k)
    % \leq 
    % c^*(i,j) + c^*(i,k)
\end{align*}
Therefore, in this case, 
\begin{align*}
    \sum_{(i', j')\in t} \expect[M_{i',j',t}] \leq \frac{4\cdot \ind((i,j)\in E_{\lowerr})\cdot y_{\ell^*(i,j)-1}(i,j)}{9}~.
\end{align*}
Accordingly, when the corresponding edge is non-pivot, we define $B_{i,j,t},B_{i,k,t},B_{j,k,t}$ as follows:
\begin{align*}
    B_{i,j,t}\defeq 
    \begin{cases}
        4/3 & \text{if $(i,j)\in E_{\lowerr}$}\\
        0 & \text{otherwise}
    \end{cases}~,
    \quad\quad
    B_{i,k,t}, B_{j,k,t} 
    \defeq
    0~.
\end{align*}
Because of Lemma~\ref{lem:stronger-y-bounded-by-c*}, this definition implies the first bullet of Lemma~\ref{lem:edge-charging-for-k-partite-B}:
\begin{align*}
    \sum_{(i',j')\in t\cap E} \expect[B_{i,j,t}]\cdot c^*(i,j) 
    % &
    % =
    % \frac{4\cdot \ind((i,j)\in E_{\lowerr})\cdot c^*(i,j)}{9}
    % \\
    &
    \geq 
    \frac{4\cdot \ind((i,j)\in E_{\lowerr})\cdot y_{\ell^*(i,j)-1}(i,j)}{9}
    \geq 
    \sum_{(i', j')\in t} \expect[M_{i',j',t}]~.
    % \tag{Lemma~\ref{lem:stronger-y-bounded-by-c*}}
\end{align*}
Further, when the edge $(i,j)$ is non-pivot, one of the pivot edges in the triangle is random. 
Because of Corollary~\ref{cor:k-prob-disappear}, the probability that $i,j$ are then partitioned into different sets on line~\ref{line:partition-non-pivots} is at least $\frac{1}{6}$, which is at least $\frac{1}{8}\cdot B_{i,j,t}$.
Hence, this definition satisfies the second bullet of Lemma~\ref{lem:edge-charging-for-k-partite-B} for this class of triangles.

\paragraph{Case 3: all three edges are in $E$.}
According to the CCDF~\eqref{eqn:distance-ccdf} of the random distance, the probability that $x'({j,k})>d_{\ell^*(y,i,j)}$ is 
\begin{align*}
    \frac{\big(4\cdot y_{\ell^*(i,j)-1}(j,k)-1\big)^+}{3}
    \leq
    \frac{4\cdot y_{\ell^*(i,j)-1}(j,k)}{3}
    &
    \leq 
    \frac{4\cdot (y_{\ell^*(i,j)-1}(i,j) + y_{\ell^*(i,j)-1}(i,k))}{3}
    \tag{LP constraint~\eqref{eqn:umvd-lp-triangle}}
    \\
    &
    =
    \frac{4\cdot (y_{\ell^*(i,j)-1}(i,j) + y_{\ell^*(i,k)-1}(i,k))}{3}
    ~.
    % \\
    % &
    % \leq 
    % c(i,j) + c(i,k)
    % \leq 
    % c^*(i,j) + c^*(i,k)
\end{align*}
Therefore, in this case,
\begin{align*}
    \sum_{(i', j')\in t} \expect[M_{i',j',t}] 
    \leq 
    \frac{4\cdot(y_{\ell^*(i,j)-1}(i,j) + y_{\ell^*(i,k)-1}(i,k))}{9} 
    \cdot 
    \big(\ind((i,j)\in E_{\lowerr}) + \ind((i,k)\in E_{\lowerr})\big)~.
\end{align*}
% Because each vertex is chosen as the pivot with equal probability $\frac{1}{3}$ conditioned on $\mathcal{A}_t$, $\expect[M_{i,j,t}|\mathcal{A}_t]$ and $\expect[M_{i,k,t}|\mathcal{A}_t]$ can be upper bounded by $\frac{1}{3}(y_{\ell^*(i,j)-1}(i,j) + y_{\ell^*(i,k)-1}(i,k))$. Therefore, 
% \begin{align*}
%     \sum_{(u,v)\in t} \expect[M_{u,v,t}|\mathcal{A}_t^2]
%     \leq 
%     \frac{2}{3}(y_{\ell^*(i,j)-1}(i,j) + y_{\ell^*(i,k)-1}(i,k))
%     ~.
% \end{align*}

Accordingly, when the corresponding edge is non-pivot, we define $B_{i,j,t},B_{i,k,t},B_{j,k,t}$ as follows:
\begin{align*}
    B_{i,j,t}
    \defeq
    \begin{cases}
     8/3 & \text{if } (i,j)\in E_{\lowerr}\\
     4/3 & \text{if }(i,j)\in E_{\highdet}\\
     0 & \text{if } (i,j)\in E_{\highrand}
    \end{cases}
    ~,
    \quad
    B_{i,k,t} 
    \defeq
    \begin{cases}
     8/3 & \text{if } (i,k)\in E_{\lowerr}\\
     4/3 & \text{if }(i,k)\in E_{\highdet}\\
     0 & \text{if } (i,k)\in E_{\highrand}
    \end{cases}
    ~,
    \quad
    B_{j,k,t}
    \defeq
    0
\end{align*}
Because of Lemma~\ref{lem:stronger-y-bounded-by-c*} and~\ref{lem:y-bounded-by-0}, this definition 
implies
\begin{align*}
    \expect[B_{i,j,t}] \cdot c^*(i,j)
    &
    \geq 
    \frac{4\cdot (1+\ind(\text{$(i,j)$ is $\lowerr$}))}{9} \cdot y_{\ell^*(i,j)-1}(i,j)
    ~,
    \\
    \expect[B_{i,k,t}] \cdot c^*(i,k)
    &
    \geq 
    \frac{4\cdot (1+\ind(\text{$(i,j)$ is $\lowerr$}))}{9} \cdot y_{\ell^*(i,k)-1}(i,k)
    ~,
\end{align*}
and thus 
implies the first bullet of Lemma~\ref{lem:edge-charging-for-k-partite-B}:
 \begin{align*}
    \sum_{(i',j')\in t\cap E}
    \expect[B_{i',j',t}] \cdot c^*(i',j')
    &
    % \geq 
    % \frac{1+\ind((i,j)\in E_{\lowerr})}{3} \cdot y_{\ell^*(i,j)-1}(i,j) +
    % \frac{1+\ind((i,k)\in E_{\lowerr})}{3} \cdot y_{\ell^*(i,k)-1}(i,k)
    % \\
    % &
    \geq
    \frac{4\cdot(\ind((i,j)\in E_{\lowerr})+\ind((i,k)\in E_{\lowerr}))}{9} \cdot \big(y_{\ell^*(i,j)-1}(i,j) + y_{\ell^*(i,k)-1}(i,k)\big)
    \\
    &
    \geq
    \sum_{(i', j')\in t} \expect[M_{i',j',t}]~.
\end{align*} 

Further, when the edge $(i,j)$ (or $(i,k)$) is non-pivot, one of the pivot edges in the triangle is random. 
Because of Corollary~\ref{cor:k-prob-disappear}, the probability $i,j$ (or $i,k$) are then partitioned into different sets on line~\ref{line:partition-non-pivots} is at least $\frac{1}{6}$, which is at least
\begin{align*}
    \begin{cases}
        \frac{1}{16}\cdot B_{i,j,t} & \text{if $(i,j)\in E_{\lowerr}$}\\
        \frac{1}{8}\cdot B_{i,j,t} & \text{if $(i,j)\in E_{\highdet}$}\\
        1\cdot B_{i,j,t} & \text{if $(i,j)\in E_{\highrand}$}
    \end{cases}
    \quad 
    \text{and}
    \quad
    \begin{cases}
        \frac{1}{16}\cdot B_{i,k,t} & \text{if $(i,k)\in E_{\lowerr}$}\\
        \frac{1}{8}\cdot B_{i,k,t} & \text{if $(i,k)\in E_{\highdet}$}\\
        1\cdot B_{i,k,t} & \text{if $(i,k)\in E_{\highrand}$}
    \end{cases}
\end{align*}
Hence, this definition satisfies the second bullet of Lemma~\ref{lem:edge-charging-for-k-partite-B} for this case.

\subsubsection{$(d,d,r)$-diff-triangles} 
W.l.o.g. we assume $(i,j)$ and $(i,k)$ are deterministic, and $(i,j)$ has a lower dominant level, i.e., $\ell^*(y,i,j)<\ell^*(y,i,k)$ (equivalently, $d_{\ell^*(y,i,j)}>d_{\ell^*(y,i,k)}$). 
This implies $M_{j,k,t}=0$.
Because of Lemma~\ref{lem:low-error-modify}, the input distance of $(i,j)$ (or $(i,k)$) equals $d_{\ell^*(y,i,j)}$ (resp. $d_{\ell^*(y,i,k)}$).
Therefore, $M_{i,j,t}$ (or $M_{i,k,t}$) equals $1$ only when the pivot is $k$ (resp., $j$), it is low-cost and the random distance $x'(j,k)\neq d_{\ell^*(y,i,j)}$, i.e.,
\begin{align*}
    M_{i,j,t} &= \ind(\text{$k$ is the pivot vertex})\cdot \ind((i,j)\in E_{\lowerr}) \cdot \ind(x'(j,k)\neq d_{\ell^*(y,i,j)})~,
    \\
    M_{i,k,t} &= \ind(\text{$j$ is the pivot vertex})\cdot \ind((i,k)\in E_{\lowerr}) \cdot \ind(x'(j,k)\neq d_{\ell^*(y,i,j)})~.
\end{align*}
Note that triangles in this class with at least two edges in $E$ can be divided into the following two cases: (1) the random edge is in $E_{\varnothing}$, (2) the random edge is in $E$.
Next, we shall discuss these two cases to prove Lemma~\ref{lem:edge-charging-for-k-partite-B} for this class of triangles. 

\paragraph{Case 1: the random edge is in $E_{\varnothing}$.}
According to the CCDF~\eqref{eqn:k-part-distance-ccdf} of the random distance, the probability that $x'({j,k})\neq d_{\ell^*(y,i,j)}$ is 
\begin{align*}
    1 - \Delta y_{\ell^*(i,j)}(j,k)
    &
    = 
    1-y_{\ell^*(i,j)}(j,k)+y_{\ell^*(i,j)-1}(j,k)
    \\
    &
    \leq 
    1 - (y_{\ell^*(i,j)}(i,j) - y_{\ell^*(i,j)}(i,k)) + y_{\ell^*(i,j)-1}(i,j) + y_{\ell^*(i,j)-1}(i,k)
    \\
    &
    \leq 
    (1 - \Delta y_{\ell^*(i,j)}(i,j)) + 2y_{\ell^*(i,k)-1}(i,k) 
\end{align*}
Therefore, in this class of triangles,
\begin{align*}
    \sum_{(i', j')\in t} \expect[M_{i',j',t}] 
    \leq
    \frac{(1 - \Delta y_{\ell^*(i,j)}(i,j))+2y_{\ell^*(i,k)-1}(i,k)}{3} \cdot \big(\ind((i,j)\in E_{\lowerr}) + \ind((i,k)\in E_{\lowerr})\big)~.
\end{align*}

Accordingly, when the corresponding edge is non-pivot, we define $B_{i,j,t},B_{i,k,t},B_{j,k,t}$ as follows:
\begin{align*}
    B_{i,j,t}
    % &
    \defeq
    \begin{cases}
         2 & \text{if $(i,j) \in E_{\lowerr}$}\\
         1 & \text{if $(i,j)\in E_{\highdet}$}\\
         1 & \text{if $(i,j)\in E_{\highrand}$}
    \end{cases}
    % \begin{cases}
    %     \min\big\{2,\frac{1}{2(1-\alpha)}\big\} & \text{$(i,j)$ is $\lowerr$ or $\highdet$}\\
    %     \frac{1-\alpha}{\alpha} \cdot \min\big\{2,\frac{1}{2(1-\alpha)}\big\} & \text{$(i,j)$ is $\lowerr$ or $\highdet$}
    % \end{cases}
    ~, 
    \qquad
    B_{i,k,t}
    % &
    \defeq
    \begin{cases}
        4 & \text{if $(i,k)\in E_{\lowerr}$}\\
        2 & \text{if $(i,k)\in E_{\highdet}$}\\
        0 & \text{if $(i,k)\in E_{\highrand}$}
    \end{cases}
    ~, 
    \qquad
    B_{j,k,t} 
    % &
    \defeq 
    0
    ~.
\end{align*}
Because of Lemma~\ref{lem:stronger-y-bounded-by-c*} and~\ref{lem:y-bounded-by-0}, this definition implies
\begin{align*}
    c^*(i,j) \cdot \expect[B_{i,j,t}]
    &
    \geq 
    \frac{1+\ind(\text{$(i,j)$ is $\lowerr$})}{3} \cdot (1 - \Delta y_{\ell^*(i,j)}(i,j))
    ~,
    \\
    c^*(i,k) \cdot \expect[B_{i,k,t}]
    &
    \geq 
    \frac{1+\ind(\text{$(i,j)$ is $\lowerr$})}{3} \cdot 2y_{\ell^*(i,k)-1}(i,k)
    ~,
\end{align*} 
and thus implies the first bullet of Lemma~\ref{lem:edge-charging-for-k-partite-B}:
\begin{align*}
    \sum_{(i',j')\in t\cap E}
    \expect[B_{i',j',t}] \cdot c^*(i',j')
    &
    % \geq 
    % \frac{1+\ind((i,j)\in E_{\lowerr})}{3} \cdot y_{\ell^*(i,j)-1}(i,j) +
    % \frac{1+\ind((i,k)\in E_{\lowerr})}{3} \cdot y_{\ell^*(i,k)-1}(i,k)
    % \\
    % &
    \geq
    \frac{\ind((i,j)\in E_{\lowerr})+\ind((i,k)\in E_{\lowerr})}{3} \cdot \big(1 - \Delta y_{\ell^*(i,j)}(i,j)+2y_{\ell^*(i,k)-1}(i,k)\big)
    \\
    &
    \geq
    \sum_{(i', j')\in t} \expect[M_{i',j',t}]~.
\end{align*} 

Further, when the edge $(i,j)$ (or $(i,k)$) is non-pivot, one of the pivot edges in the triangle is random in $E_{\varnothing}$. 
Because of Lemma~\ref{lem:k-prob-disappear}, the probability that $i,j$ (or $i,k$) are then partitioned into different sets on line~\ref{line:partition-non-pivots} is at least $\frac{1}{4}$, which is at least
\begin{align*}
    \begin{cases}
        \frac{1}{8} \cdot B_{i,j,t} & \text{if $(i,j)\in E_{\lowerr}$}\\
        \frac{1}{4} \cdot B_{i,j,t} & \text{if $(i,j)\in E_{\highdet}$}\\
        \frac{1}{4} \cdot B_{i,j,t} & \text{if $(i,j)\in E_{\highrand}$}
    \end{cases}
    \quad 
    \text{and}
    \quad
    \begin{cases}
        \frac{1}{16} \cdot B_{i,k,t} & \text{if $(i,k)\in E_{\lowerr}$}\\
        \frac{1}{8} \cdot B_{i,k,t} & \text{if $(i,k)\in E_{\highdet}$}\\
        1 \cdot B_{i,k,t} & \text{if $(i,k)\in E_{\highrand}$}
    \end{cases}
\end{align*}
Hence, this definition satisfies the second bullet of Lemma~\ref{lem:edge-charging-for-k-partite-B} for this case.

\paragraph{Case 2: the random edge is in $E$.} 
In this case, we always have $\sum_{(i', j')\in t} \expect[M_{i',j',t}]\leq \frac{2}{3}$.
Accordingly, when the corresponding edge is non-pivot, we define $B_{i,j,t},B_{i,k,t},B_{j,k,t}$ as follows:
\begin{align*}
    B_{i,j,t}, B_{i,k,t}
    \defeq
    0
    ~,
    \qquad
    B_{j,k,t}
    \defeq
    16/3
    ~.
\end{align*}
Because of Corollary~\ref{cor:low-err-are-det}, $(j,k)\in E_{\highrand}$ and thus $c^*(j,k)\geq \frac{3}{8}$.
This definition implies the first bullet of Lemma~\ref{lem:edge-charging-for-k-partite-B}:
\begin{align*}
    \sum_{(i',j')\in t\cap E} \expect[B_{i',j',t}]\cdot c^*(i',j') 
    \geq 
    \frac{1}{3}\cdot \frac{16}{3} \cdot \frac{3}{8} 
    = 
    \frac{2}{3}
    \geq
    \sum_{(i', j')\in t} \expect[M_{i',j',t}]~.
\end{align*}
Further, when the edge $(j,k)$ is non-pivot, both pivot edges in the triangle are deterministic but have different dominant levels. 
Because of Lemma~\ref{lem:k-prob-disappear}, the probability that $j,k$ are then partitioned into different sets on line~\ref{line:partition-non-pivots} is 1, which is at least $\frac{3}{16}\cdot B_{j,k,t}$. 
Hence, this definition satisfies the second bullet of Lemma~\ref{lem:edge-charging-for-k-partite-B} for this case.

\subsubsection{$(d,r,r)$-triangles} 
W.l.o.g., we assume $(j,k)$ is the deterministic edge.
Because of Corollary~\ref{cor:low-err-are-det}, $(i,j),(i,k)\notin E_{\lowerr}$ and thus $M_{i,j,t}=M_{i,k,t}=0$.
If $(j,k)\notin E_{\lowerr}$, no low-cost edges are modified in this triangle, and thus $M_{j,k,t}=0$. 
Accordingly, by always defining $B_{i,j,t}=B_{i,k,t}=B_{j,k,t}=0$, we show Lemma~\ref{lem:edge-charging-for-k-partite-B} for the case $(j,k)\notin E_{\lowerr}$.

Next, we consider the case $(j,k)\in E_{\lowerr}$.
It is clear to upper bound $M_{i,j,t}$ by $1$ and thus
\begin{align*}
    \sum_{(i', j')\in t} \expect[M_{i',j',t}] \leq \frac13.
\end{align*}
Since $E_{\lowerr}\subseteq E$, $(j,k)\in E$. 
Because $E$ is $k$-partite, at least one of $(i,j)$ and $(i,k)$ is in $E$.
W.l.o.g., we assume $(i,j)\in E$.
We shall prove Lemma~\ref{lem:edge-charging-for-k-partite-B} by discussing whether $(i,k)\in E$.

\paragraph{Case 1: $(i,k)\notin E$.}
Because $(i,j)$ is random in $E$ and because of Corollary~\ref{cor:low-err-are-det}, $(i,j)\in E_{\highrand}$ and thus $c^*(i,j)\geq \frac{3}{8}$.
Accordingly, when the corresponding edge is non-pivot, we define $B_{i,j,t},B_{i,k,t},B_{j,k,t}$ as follows:
\begin{align*}
    B_{i,j,t}\defeq 8/3, \qquad B_{i,k,t},\, B_{j,k,t} \defeq 0~,
\end{align*}
% we always define $B_{j,k,t}$ to be $0$, and define $B_{i,j,t}$ (or $B_{i,k,t}$) to be $\frac{1}{2(1-\alpha)}$ when it is backward. 
% Since edges $(i,j)$ and $(i,k)$ are random and thus high-error, 
% \begin{align*}
%     \sum_{(u,v)\in t} c^*(u,v) \cdot \expect[B_{u,v,t}|\mathcal{A}_t] = \frac{1}{6(1-\alpha)}(c^*(i,j)+c^*(i,k)) \geq \frac{1}{6(1-\alpha)}\cdot 2(1-\alpha)=\frac{1}{3}~. 
% \end{align*}
% Therefore, we derive Lemma~\ref{lem:improved-M-bound-by-B} for this type. 
which implies the first bullet of Lemma~\ref{lem:edge-charging-for-k-partite-B} for this class of triangles:
\begin{align*} 
    \sum_{(i',j')\in t\cap E} \expect[B_{i',j',t}]\cdot c^*(i',j') \geq \frac{1}{3}\cdot \frac{8}{3} \cdot \frac{3}{8} = \frac{1}{3} \geq \sum_{(i', j')\in t} \expect[M_{i',j',t}]~.
\end{align*}
Further, when the edge $(i,j)$ is non-pivot, one of the pivot edges in the triangle is random in $E_{\varnothing}$. 
Because of Lemma~\ref{lem:k-prob-disappear}, the probability that $i,j$ are then partitioned into different sets on line~\ref{line:partition-non-pivots} is at least $\frac{1}{4}$, which equals $\frac{3}{32}\cdot B_{i,j,t}$. 
Because $(i,j)\in E_{\highrand}$, this definition satisfies the second bullet of Lemma~\ref{lem:edge-charging-for-k-partite-B} for this case.

\paragraph{Case 2: $(i,k)\in E$.}
Because $(i,j)$ and $(i,k)$ are random in $E$ and because of Corollary~\ref{cor:low-err-are-det}, $(i,j),(i,k) \in E_{\highrand}$ and thus $c^*(i,j), c^*(i,k)\geq \frac{3}{8}$.
Accordingly, when the corresponding edge is non-pivot, we define $B_{i,j,t},B_{i,k,t},B_{j,k,t}$ as follows:
\begin{align*}
    B_{i,j,t},\,B_{i,k,t}\defeq 4/3, \qquad B_{j,k,t} \defeq 0~,
\end{align*}
% we always define $B_{j,k,t}$ to be $0$, and define $B_{i,j,t}$ (or $B_{i,k,t}$) to be $\frac{1}{2(1-\alpha)}$ when it is backward. 
% Since edges $(i,j)$ and $(i,k)$ are random and thus high-error, 
% \begin{align*}
%     \sum_{(u,v)\in t} c^*(u,v) \cdot \expect[B_{u,v,t}|\mathcal{A}_t] = \frac{1}{6(1-\alpha)}(c^*(i,j)+c^*(i,k)) \geq \frac{1}{6(1-\alpha)}\cdot 2(1-\alpha)=\frac{1}{3}~. 
% \end{align*}
% Therefore, we derive Lemma~\ref{lem:improved-M-bound-by-B} for this type. 
which implies the first bullet of Lemma~\ref{lem:edge-charging-for-k-partite-B} for this class of triangles:
\begin{align*} 
    \sum_{(i',j')\in t\cap E} \expect[B_{i',j',t}]\cdot c^*(i',j') \geq \frac{1}{3}\cdot \frac{4}{3} \cdot \frac{3}{8} + \frac{1}{3}\cdot \frac{4}{3} \cdot \frac{3}{8} = \frac{1}{3} \geq \sum_{(i', j')\in t} \expect[M_{i',j',t}]~.
\end{align*}
Further, when the edge $(i,j)$ (or $(i,k)$) is non-pivot, one of the pivot edges in the triangle is random in $E$. 
Because of Lemma~\ref{lem:k-prob-disappear}, the probability that $i,j$ (or $i,k$) are then partitioned into different sets on line~\ref{line:partition-non-pivots} is at least $\frac{1}{6}$, which equals $\frac{1}{8}\cdot B_{i,j,t}$ (resp., $\frac{1}{8}\cdot B_{i,k,t}$). 
Because $(i,j), (i,k)\in E_{\highrand}$, this definition satisfies the second bullet of Lemma~\ref{lem:edge-charging-for-k-partite-B} for this case.

% \subsubsection{$(r,r,r)$-triangles} 
% Note that there is no low-cost edges in this class of triangle, $M_{i',j',t}=0$ for any $i'\neq j'\in t$. 
% Accordingly, we can always define $B_{i,j,t}=B_{i,k,t}=B_{j,k,t}=0$, and Lemma~\ref{lem:edge-charging-for-k-partite-B} trivially holds for this class. 

\end{document}